\documentclass[11pt,reqno, a4paper]{amsart} 

\usepackage{epsfig}
\usepackage{graphicx}
\usepackage{enumerate}
\usepackage{color} 
\usepackage[toc,page]{appendix}

\usepackage{setspace}

\usepackage{lipsum}
\makeatletter
\g@addto@macro{\endabstract}{\@setabstract}
\newcommand{\authorfootnotes}{\renewcommand\thefootnote{\@fnsymbol\c@footnote}}%
\makeatother

\usepackage[dvipsnames]{xcolor}
\usepackage{bm,bbm}
\usepackage{cancel}
\usepackage{enumitem}

\usepackage{physics}

\usepackage{amsmath,euscript}
\usepackage{amsthm}
\usepackage{amssymb}
\usepackage{graphicx}
\usepackage{mathrsfs}
\usepackage{epsf}
\usepackage{xcolor}
\usepackage{psfrag}
\usepackage{enumerate}
\usepackage{amsfonts,amssymb,amsthm,amsmath} 
\usepackage{hyperref}

\usepackage[numbers]{natbib}
\usepackage{enumitem} 
\usepackage[margin=1.37in]{geometry}

\usepackage{soul}

\newtheorem{thm}{Theorem}[section]

\newtheorem{lemma}{Lemma}[section]
\newtheorem{prop}{Proposition}[section]
\newtheorem{cor}{Corollary}[section]

\theoremstyle{remark}
\newtheorem{remark}{Remark}

\numberwithin{equation}{section}

\usepackage{soul}
\setstcolor{red}

\definecolor{green}{rgb}{0.0, 0.5, 0.5}
\definecolor{yellow}{rgb}{0.5, 0.5, 0}
\definecolor{lgray}{gray}{0.9}
\definecolor{llgray}{gray}{0.95}
\definecolor{lllgray}{gray}{0.975}

\newcommand{\nc}{\newcommand}

\nc{\la}{\label}
\nc{\ba}{\begin{array}}
\nc{\ea}{\end{array}}
\nc{\bs}{\begin{split}}
\nc{\es}{\end{split}}

\newcommand{\R}{\mathbb{R}}
\newcommand{\C}{\mathbb{C}}
\newcommand{\Z}{\mathbb{Z}}
\newcommand{\T}{\mathbb{T}}

\newcommand{\cB}{\mathcal{B}}
\newcommand{\cD}{\mathcal{D}}
\newcommand{\cC}{\mathcal{C}}
\newcommand{\cE}{\mathcal{E}}
\newcommand{\cF}{\mathcal{F}}
\newcommand{\cG}{\mathcal{G}}
\newcommand{\cH}{\mathcal{H}}

\newcommand{\cK}{\mathcal{K}}

\newcommand{\cN}{\mathcal{N}}       
\newcommand{\cQ}{\mathcal{Q}}
\newcommand{\cR}{\mathcal{R}}

\newcommand{\cW}{\mathcal{W}}

\nc{\ran}{\rangle}
\nc{\lan}{\langle}

\renewcommand{\Re}{\mathrm{Re}} 

\nc{\bfone}{{\bf 1}}
\nc{\1}{{\bf 1}}

\usepackage{centernot}

\usepackage[normalem]{ulem}

\newcommand{\DETAILS}[1]{}

\nc{\den}{\text{den}}
\nc{\ex}{\text{xc}}

\usepackage{soul}

\usepackage{tikz}
\usetikzlibrary{automata,positioning,arrows}

\begin{document}
\title[Collective behaviors of an electron gas]{Collective behaviors of an electron gas in the mean-field regime}

\maketitle

\begin{center}
	\authorfootnotes
	Dong Hao Ou Yang\footnote{\textit{Email}: ouyang@math.lmu.de}\textsuperscript{1} \par \bigskip
	
	\textsuperscript{1}LMU Munich, Dept of Mathematics, Theresienstr. 39, 80333 Munich, Germany \par\bigskip
\end{center}

\begin{abstract}
In this paper, we study the momentum distribution of an electron gas in a $3$-dimensional torus.  The goal is to compute the occupation number of Fourier modes for some trial state obtained through random phase approximation.  We obtain the mean-field analogue of momentum distribution formulas for electron gas in [Daniel and Voskov, Phys. Rev. \textbf{120}, (1960)] in high density limit and [Lam, Phys. Rev. \textbf{3}, (1971)] at metallic density.  The analysis in the present paper is majorly based on the work [Christiansen, Hainzl, Nam, Comm. Math. Phys. \textbf{401}, (2023)].  Our findings are related to recent results obtained independently by Benedikter, Lill and Naidu, and the analysis applies to a general class of singular potentials rather than just the Coulomb case.
\end{abstract}

\section{Introduction and main result}\label{sec:intro-main}
\subsection{Introduction}\label{sec:intro}
We consider a system of $N$ electron gas in the torus $\T^{3}=[0,2\pi]^{3}$ in the mean field limit, whose Hamiltonian is given by (in the unit $\hbar=1$)
\begin{align}
	H_{N}&=\sum_{j=1}^{N}(-\Delta_{x_{j}})+k_{F}^{-1}\sum_{1\leq i<j\leq N}V(x_{i}-x_{j})\quad\text{ on }\cH_{N}:=L_{a}^{2}(\T^{3N}).
\end{align}
Here we denote $L_{a}^{2}(\T^{3N})$ as the space of square-integrable, totally antisymmetric functions of $3N$-dimensional torus $\T^{3N}$.  The standard Laplacian operator $-\Delta$ describes the kinetic energy for each individual particle, $k_{F}$ denotes the Fermi momentum, and $V$ is given by the Comloub potential function
\begin{align}\label{Coulomb}
	V(x)&=\frac{1}{(2\pi)^{3}}\sum_{k\in\Z_{*}^{3}}\hat{V}_{k}e^{ik\cdot x},\quad \hat{V}_{k}:=\begin{cases}
		g|k|^{-2}\quad&\text{ for }k\neq 0,\\
		0&\text{ for }k=0.
	\end{cases}
\end{align}
where $g\geq 0$ is the coupling constant and $\Z_{*}^{3}:=\Z^{3}\setminus\{0\}$.  Since $H_{N}$ is bounded from below, we define $H_{N}$ as a quadratic form by Friedrichs method with form domain $H^{1}(\T^{3})$.  The \textit{ground state energy} $E_{N}$ is then given by
\begin{align}\label{gs-energy}
	E_{N}&:=\inf\sigma(H_{N})=\inf_{\psi\in L_{a}^{2}(\T^{3N})}\frac{\langle\psi,H_{N}\psi\rangle}{\|\psi\|^{2}},
\end{align}
and any eigenvector of $H_{N}$ with eigenvalue $E_{N}$ is called a \textit{ground state}.  

In the non-interacting case (i.e., $g=0$), the ground states $\Psi_{FS}$ (which will be called \textit{Fermi state}) are given by the Slater determinants comprising $N$ plane waves with different momentum $k_{j}\in\Z^{3}$ for $j=1,...,N$ of minimized kinetic energy $|k_{j}|^{2}$, i.e., 
\begin{align}
	\Psi_{FS}&=u_{k_{1}}\wedge...\wedge u_{k_{N}},\quad u_{p}(x)=(2\pi)^{-3/2}e^{-ip\cdot x}.
\end{align}
This state is unique up to a phase if we assume the Fermi ball $B_{F}$ is completely filled with $N$ integer points in the momentum space, i.e., 
\begin{align}
	N&=|B_{F}|,\quad \text{where } B_{F}=\{k\in\Z^{3}\mid |k|\leq k_{F}\}\text{ for some }k_{F}>0.
\end{align}
This implies that the Fermi momentum $k_{F}$ is scaled like
\begin{align}
	k_{F}&=\Big[\Big(\frac{3\pi}{4}\Big)^{1/3}+O(N^{-1/3})\Big]N^{1/3}.
\end{align}
We denote the complement of Fermi ball by $B_{F}^{c}:=\Z^{3}\setminus B_{F}$.

In the interacting case, the Slater determinant $\Psi_{FS}$ is no longer a ground state for $H_{N}$.  In the present case, we are interested in the collective correction induced by the interaction.  This is to be compared with general quantum mechanical system in which linear combinations of Slater determinants are allowed.  
In the work \cite{CHN-23,CHN-24-1}, Christiansen, Hainzl and Nam estimated the ground state energy by: As $k_{F}\rightarrow\infty$, it holds that
\begin{align}\label{correlation-energy}
	E_{N}&=E_{FS}+E_{\rm corr,bos}+E_{\rm corr,ex}+O(k_{F}^{5/6+\epsilon})
\end{align}
for any $\epsilon>0$, where
\begin{align}\label{boson-energy}
	E_{\rm corr,bos}&=\frac{1}{\pi}\sum_{k\in\Z_{*}^{3}}\int_{0}^{\infty}F\big(Q_{k}(s)\big)ds,\quad F(x)=\log(1+x)-x,
\end{align}
is the energy contribution due to bosonization (i.e., electron-hole pair excitations) and 
\begin{align}\label{ex-energy}
	E_{\rm corr,ex}&=\frac{k_{F}^{-2}}{4(2\pi)^{6}}\sum_{k\in\Z_{*}^{3}}\sum_{p,q\in L_{k}}\frac{\hat{V}_{k}\hat{V}_{p+q-k}}{\lambda_{k,p}+\lambda_{k,q}},
\end{align}
is the exchange correlation (one should not confuse it with exchange energy in Hartree-Fock approximation), with $L_{k}$ denotes the \textit{lune} given by
\begin{align}\label{lune}
	L_{k}&:=B_{F}^{c}\cap (B_{F}+k)=\{p\in\Z^{3}\mid |p-k|\leq k_{F}< |p|\}.
\end{align}
Here we denote
\begin{align}
	\lambda_{k,p}:=\frac{1}{2}\big(|p|^{2}&-|p-k|^{2}\big),\quad Q_{k}(s):=\frac{k_{F}^{-1}\hat{V}_{k}}{(2\pi)^{3}}\sum_{p\in L_{k}}\frac{\lambda_{k,p}}{s^{2}+\lambda_{k,p}^{2}},
\end{align}
and $E_{FS}$ is the energy in Fermi state $\Psi_{FS}$ that can be computed explicitly:
\begin{align}
	E_{FS}&=\langle\Psi_{FS},H_{N}\Psi_{FS}\rangle=\sum_{p\in B_{F}}|p|^{2}+\frac{1}{2(2\pi)^{3}}\sum_{k\in\Z_{*}^{3}}\hat{V}_{k}(|L_{k}|-N).
\end{align}
We note the following important lower bound for $\lambda_{k,p}$: Denote for each $p\in\Z^{3}$ that (see \cite[Eq. (A.2)]{CHN-24-1})
\begin{align}\label{kappa}
	\kappa&:=\frac{1}{2}\Big(\inf_{p\in B_{F}^{c}}|p|^{2}+\sup_{q\in B_{F}}|q|^{2}\Big),\quad m(p)^{-1}:=\big||p|^{2}-\kappa\big|\geq\frac{1}{2}.
\end{align}
Then, since $|p-k|^{2}\leq \kappa\leq |p|^{2}$ for each $p\in L_{k}$, we have 
\begin{align}\label{gap}
	\lambda_{k,p}&=\frac{1}{2}\big(\big||p|^{2}-\kappa\big|+\big||p-k|^{2}-\kappa|\big|\big)\geq \frac{1}{2}\big(m(p)^{-1}+m(p-k)^{-1}\big)\geq\frac{1}{2}.
\end{align}
For convenience, we also introduce $L_{k}':=L_{k}-k$, the set of hole states for each $k\in\Z_{*}^{3}$.  We remark that concerning the upper bound in \eqref{correlation-energy}, the analysis in \cite{CHN-23} can be extended to any interaction potential with positive, square summable Fourier mode on $\Z_{*}^{3}$, with an error of order at most $O(\sqrt{k_{F}})$.


\begin{remark}
Note that for potential less singular than Coulomb potential, e.g., if $\hat{V}_{k}\leq O(|k|^{2+\epsilon})$ for some $\epsilon>0$, then $E_{\rm corr,bos}$ is of order $k_{F}$ and $E_{\rm corr,ex}$ is of order $o(k_{F})$.  However, for Coulomb potential, then $E_{\rm corr,bos}$ is of order $k_{F}\log(k_{F})$ and $E_{\rm corr,ex}$ is of order $k_{F}$; see \cite{CHN-23} for a detailed explanation.  Thus the Coulomb potential is critical, and the correlation energy $E_{\rm corr,bos}+E_{\rm corr,ex}$ in \eqref{correlation-energy} can be interpreted as the mean-field analogue of the Gell-Mann--Brueckner formula $c_{1}\rho\log(\rho)+c_{2}\rho$ for jellium model in thermodynamic limit with the particle density $\rho$ sufficiently high \cite{GMBru-57}.  This is a refinement of the random phase approximation due to Bohm and Pines \cite{BohmPines-51,BohmPines-52,BohmPines-53,Pines-53}.
\end{remark}

\begin{remark}
For less singular positive interaction potentials, e.g., $\sum_{k\in\Z_{*}^{3}}(1+|k|)\hat{V}_{k}<\infty$, one can show $\Psi_{FS}$ remains a unique minimizer for Hartree-Fock approximation and a more precise leading order correlation energy (i.e., the difference between true ground state energy and the energy calculated by Hartree-Fock approximation) are obtained; see \cite{BNPSS-20,BNPSS-21,BPSS-23,CHN-24} and references therein.
\end{remark}

In the present paper, we are interested in the properties of the ground states of the system.  In this direction, it is a fundamental question in the condensed matter physics to understand if a
similar behavior also holds for true ground states, and it is expected each true ground state has a superconducting part that will smooth out the jump discontinuity as we move from inside to outside of Fermi ball (due to Kohn-Littinger theorem \cite{Luttinger-60,LuttingerWard-60}).  On the mathematicla side, the structure of a ground state is very delicate and its momentum distribution is thus difficult to determine.  Inspired by the recent work \cite{BeneLill-25} by Benediter and Lill, we will study the momentum distrubution for the trial state constructed in \cite{CHN-23}.  




\subsection{Trial states}\label{sec:trial-states}
Let us recall the trial state constructed in \cite{CHN-23}, which gives the energy upper bound in \eqref{correlation-energy}.  To this end, it is convenient to pass to the fermonic Fock space
\begin{align*}
	\cF_{a}&\equiv \cF_{a}(L^{2}(\T^{3})):=\C\oplus\bigoplus_{N=1}^{\infty}L_{a}^{2}(\T^{3N}).
\end{align*}
Each element $\Psi\in\cF_{a}$ is given by a sequence $(\psi_{0},\psi_{1},...,\psi_{N},...)$ with $\psi_{0}\in\C$ and $\psi_{N}\in L_{a}^{2}(\T^{3N})$ for each $N\geq 1$.  We denote the annihilation and creation operators on $\cF_{a}$ associated with $f\in L^{2}(\T^{3})$ as $a(f)$ and $a^{*}(f)$, respectively, satisfying the canoncial anticommutation relation (CAR): For any $f,g\in L^{2}(\T^{3})$, 
\begin{align}\label{CAR}
	\{a(f),a(g)\}=\{a^{*}(f),a^{*}(g)\}=0,\quad\{a(f),a^{*}(g)\}=\langle f,g\rangle.
\end{align}
For each plane wave $u_{p}=(2\pi)^{-3/2}e^{ip\cdot x}$ with $p\in\Z^{3}$, we denote
\begin{align}
	a_{p}&=a(u_{p})\text{ and } a_{p}^{*}=a^{*}(u_{p}).
\end{align}

Next, for each $k\in\Z_{*}^{3}$ and $p\in L_{k}$, we define the \textit{excitation operator} $b_{k,p},b_{k,p}^{*}$ as
\begin{align}\label{excitation}
	b_{k,p}&=a_{p-k}^{*}a_{p},\quad b_{k,p}^{*}=a_{p}^{*}a_{p-k},\quad k\in\Z_{*}^{3},p\in L_{k}.
\end{align}
The name is due to the fact that $b_{k,p}^{*}$ acts on $\cF_{a}$ by creating a state with momentum $k\in B_{F}^{c}$ and destroying a state with momentum $p-k\in B_{F}$.  In other words, it excits the state $p-k\in B_{F}$ to the state $k\in B_{F}^{c}$.  

Note that the excitation operators behave ``\textit{quasi-bosonically}" in the sense that, for each $k,l\in\Z_{*}^{3}$ and $p\in L_{k},q\in L_{l}$, they satisfy the following commutation relation:
\begin{align}
	[b_{k,p},b_{l,q}]&=[b_{k,p}^{*},b_{l,q}^{*}]=0,\\
	\label{almost-bosonic}[b_{k,p},b_{l,q}^{*}]&=\delta_{k,p}\delta_{p,q}-\big(\delta_{p,q}a_{q-l}a_{p-k}^{*}+\delta_{p-k,q-l}a_{q}^{*}a_{p}\big),
\end{align}
i.e., $b_{k,p}^{*},b_{k,p}$ is exactly bosonic if the terms in the bracket on the r.h.s. of \eqref{almost-bosonic} vanish.  In our quasi-bosonic setting, it becomes non-trivial and gives exchange contribution $E_{\rm corr, ex}$ to the correlaton energy \eqref{correlation-energy}.  

For computational convenience, it is better to introduce a basis-independent way of writing quasi-bosonic operators.  For each $k\in\Z_{*}^{3}$, we define an auxilliary vector space $\ell^{2}(L_{k})$, which we will consider as a real vector space with the standard orthonormal basis $(e_{p})_{p\in L_{k}}$.  Then, for each $k\in\Z_{*}^{3}$ and $\varphi\in\ell^{2}(L_{k})$, we define the \textit{generalized excitation operators} by
\begin{align}
	b_{k}(\varphi)&=\sum_{p\in L_{k}}\langle \varphi,e_{p}\rangle b_{k,p},\quad b_{k}^{*}(\varphi)=\sum_{p\in L_{k}}\langle e_{p},\varphi\rangle b_{k,p}^{*}.
\end{align}

Now, we construct our trial state according to \cite{CHN-23} as
\begin{align}\label{trial-N}
	\Psi_{N}&=e^{-\cK}\Psi_{FS},
\end{align}
where $\cK$ is the \textit{quasi-Bogoliubov kernel} on $\cH_{N}$ defined as
\begin{align}\label{quasiB-Bog-ker}
	\cK\DETAILS{&=\frac{1}{2}\sum_{k\in\Z_{*}^{3}}\sum_{p,q\in L_{k}}\langle e_{p},K_{k}e_{q}\rangle\big(b_{k,p}b_{-k,-q}-b_{-k,-q}^{*}b_{k,p}^{*}\big)\nonumber\\}
	&=\frac{1}{2}\sum_{k\in\Z_{*}^{3}}\sum_{p\in L_{k}}\big(b_{k}(K_{k}e_{p})b_{-k,-p}-b_{-k,-p}^{*}b_{k}^{*}(K_{k}e_{p})\big),
\end{align}
with the associated family of symmetric operators $K_{k}:\ell^{2}(L_{k})\rightarrow\ell^{2}(L_{k})$ given explicitly as
\begin{align}\label{Kk-exp}
	K_{k}&=-\frac{1}{2}\log\Big[h_{k}^{-1/2}\Big(h_{k}^{1/2}(h_{k}+2P_{k})h_{k}^{1/2}\Big)^{1/2}h_{k}^{-1/2}\Big].
\end{align}
Here, for each $k\in\Z_{*}^{3}$, the operator $h_{k}:\ell^{2}(L_{k})\rightarrow \ell^{2}(L_{k})$ is defined by the relation $h_{k}e_{p}=\lambda_{k,p}e_{p}$ and
\begin{align}
	P_{k}&=\ket{v_{k}}\bra{v_{k}},\quad v_{k}=\sqrt{\frac{k_{F}^{-1}\hat{V}_{k}}{2(2\pi)^{3}}}\sum_{p\in L_{k}}e_{p}\in \ell^{2}(L_{k}).
\end{align}
Note that, by spectral theory, $K_{k}\leq 0$ for each $k\in\Z_{*}^{3}$ and, for each $p,q\in\ell^{2}(L_{k})$, we have
\begin{align}\label{Kk}
	\langle e_{p},K_{k}e_{q}\rangle&=\langle e_{-p},K_{-k}e_{q}\rangle,\quad\langle e_{p},P_{k}e_{q}\rangle=\langle e_{p},v_{k}\rangle\langle v_{k},e_{q}\rangle=\frac{k_{F}^{-1}\hat{V}_{k}}{2(2\pi)^{3}}.
\end{align}
One can check easily that $\Psi_{N}\in\cH_{N}$ for each $N$, and $\cK$ is anti-symmetric, so that $e^{-\cK}$ is a unitary operator on $\cF_{a}$. 

\subsection{Main result}\label{sec:main}
In the present work, we want to compute the following quantity for each $\xi\in\Z^{3}$:
\begin{align}\label{momentum}
	n(\xi):=\langle\Psi_{N},\rho_{\xi}\Psi_{N}\rangle,
\end{align}
for trial state $\Psi_{N}$ defined in \eqref{trial-N}, where
\begin{align}\label{momentum-op}
	\rho_{\xi}&:=\chi_{B_{F}}(\xi)a_{\xi}a_{\xi}^{*}+\chi_{B_{F}^{c}}(\xi)a_{\xi}^{*}a_{\xi}.
\end{align}
For $\xi\in B_{F}^{c}$, the quantity \eqref{momentum-trial} represents the probability of finding a electron with momentum $\xi$, whereas, for $\xi\in B_{F}$, it represents the probability of finding a hole with this momentum.  The quantity $n(\xi)$ can also be interpreted as the momentum deviation in state $\Psi_{N}$ from sharp distribution $\chi_{B_{F}}(\xi)$.

We also consider general collective behaviors of our trial state in momentum space by computing, for a given  symmetric function $f\in \ell^{1}(\Z^{3})$, 
\begin{align}\label{n(f)}
	n(f)&:=\sum_{\xi\in\Z^{3}}f(\xi)n(\xi)\equiv \langle\Psi_{N},\rho(f)\Psi_{N}\rangle,
\end{align}
where
\begin{align}
	\rho(f)&:=\sum_{\xi\in\Z^{3}}f(\xi)\rho_{\xi}=\sum_{\xi\in\Z^{3}}f(\xi)\big(\chi_{B_{F}}(\xi)a_{\xi}a_{\xi}^{*}+\chi_{B_{F}^{c}}(\xi)a_{\xi}^{*}a_{\xi}\big).
\end{align}
It follows that $n(\xi)$ is simply given by \eqref{n(f)} when $f$ is the delta function at point $\xi$.
\begin{remark}
By introducing the \textit{particle-hole transformation} $\cR$ on $\cF_{a}$ as in \cite{BNPSS-20,BNPSS-21,BPSS-23,HPR-20}
\begin{align}\label{ph-transf}
	\cR^{*}a_{\xi}\cR&=\chi_{B_{F}}(\xi)a_{\xi}^{*}+\chi_{B_{F}^{c}}(\xi)a_{\xi},
\end{align}
which is unitary on $\cF_{a}$ and satisfies $\cR=\cR^{*}=\cR^{-1}$, we can rewrite $n(f)$ as
\begin{align}\label{n(f)2}
	n(f)&=\langle\Phi_{N},d\Gamma(f)\Phi_{N}\rangle,\quad\Phi_{N}\equiv \cR\Psi_{N},
\end{align}
where $d\Gamma(f):=\sum_{\xi\in\Z^{3}}f(\xi)a_{\xi}^{*}a_{\xi}$ is the second quantization of one-body observable $f$ on momentum space $\ell^{2}(\Z^{3})$.  By the method in \cite[Lemma 4.3]{BNPSS-20}, one can show $\Phi_{N}\in\cH_{N}$ and, by passing back from $\cF_{a}$ to $\cH_{N}$, we have
\begin{align}
	n(f)&=\sum_{(\xi_{1},...,\xi_{N})\in\Z^{3N}}\Big(\sum_{j=1}^{N}f(\xi_{j})\Big)|\Phi_{N}(\xi_{1},...,\xi_{N})|^{2}.
\end{align}
Hence, we can interpret $n(f)$ as the expectation of observable $f$ on momentum space in the transformed trial state $\Phi_{N}$.
\end{remark}
  

Now, we have our main result:
\begin{thm}\label{thm:main}
	Assume $N=|B_{F}|=4\pi k_{F}^{3}/3$ and suppose $V$ satisfies
	\begin{align}\label{square-sum}
		0\leq \hat{V}_{k}=\hat{V}_{-k},\quad\hat{V}_{0}=0,\quad\|\hat{V}\|_{\ell^{2}}<\infty.
	\end{align}
	Then, for each observable $f\in\ell^{\infty}(\Z^{3})$ such that $f(-\xi)=f(\xi)$, the expectation $n(f)$ in the trial state $\Psi_{N}$ constructed in \eqref{trial-N} is given
	\begin{align}\label{momentum-trial}
		n(f)&=n_{\rm b}(f)+n_{\rm ex}(f)+\cE(f)=\sum_{\xi\in\Z^{3}}f(\xi)\big(n_{\rm b}(\xi)+n_{\rm ex}(\xi)+\cE(\xi)\big),
	\end{align}
	where we have the bosonization contribution
	\begin{align}\label{bosonization}
		n_{\rm b}(\xi)&:=\frac{k_{F}^{-1}}{8\pi^{4}}\sum_{k\in\Z_{*}^{3}}\sum_{\zeta\in\cD_{k,\xi}}\chi_{L_{k}}(\zeta)\hat{V}_{k}\int_{0}^{\infty}\frac{(s^{2}-\lambda_{k,\zeta}^{2})(s^{2}+\lambda_{k,\zeta}^{2})^{-2}}{1+Q_{k}(s)}ds,
	\end{align}
	the exchange contribution
	\begin{align}\label{exchange}
		n_{\rm ex}(\xi)&:=-\frac{k_{F}^{-2}}{8(2\pi)^{6}}\sum_{k\in\Z_{*}^{3}}\sum_{\zeta\in\cD_{k,\xi}}\sum_{p\in L_{k}}\chi_{L_{k}}(\zeta)\frac{\hat{V}_{k}\hat{V}_{p+\zeta-k}}{(\lambda_{k,p}+\lambda_{k,\zeta})^{2}},
	\end{align}
	and some error term $\cE(\xi)$.  In above, $\cD_{k,\xi}:=\{\pm\xi,k\pm\xi\}$ and $\chi_{A}$ denotes the characteristic function of set $A$.
	
	Moreover, for each $f\in\ell^{\infty}(\Z^{3})$ with support $A$ and any $\delta>0$, it holds that
	\begin{align}\label{boson-f-esti}
		|n_{\rm b}(f)|\leq C_{V}k_{F}^{-1}\|f\|_{\ell^{\infty}}\sum_{\xi\in A}m(\xi),
	\end{align}
	for some constant $C_{V}>0$ depends only on $\|\hat{V}\|_{\ell^{2}}$, and
	\begin{align}
		&|n_{\rm ex}(f)|\leq C_{\delta,V}k_{F}^{-1+\delta}\|f\|_{\ell^{\infty}}\sum_{\xi\in A}m(\xi),\\ \label{error-f-esti}&|\cE(f)|\leq C_{\delta,V}k_{F}^{-3/2+\delta}\|f\|_{\ell^{\infty}}\sum_{\xi\in A}m(\xi),
	\end{align}
	for some constant $C_{\delta,V}>0$, depends only on $\delta$ and $\|\hat{V}\|_{\ell^{2}}$.

%
\end{thm}
The rest of this paper is devoted to prove this theorem.  We will first use techniques in \cite{CHN-23} to extract contributions due to bosonization and exchange correlation, together with some additional error terms.  We then estimate these errors in Sections \ref{sec:class-errors}--\ref{sec:principal-errors}, using the analysis based on the works \cite{BeneLill-25,CHN-23,CHN-24-1}.  Finally, in Section \ref{sec:pf-main}, we will use these estimations to complete the proof of Theorem \ref{thm:main}.

Even though we consider trial states instead of ground states, our result is interesting for the following reason.  It gives the rigorous derivation of collective behaviors for electron gas in the mean field regime in terms of momentum distribution and extract the contributions due to bosonization and exchange correlation.  Moreover, the formula is generic in the sense that our error term does not depend on specific structures of Coulomb interaction and is meaningful for a very general class of interactions satisfying \eqref{square-sum}.  

Furthermore, we can regard \eqref{momentum-trial} as the exact mean-field analogue to the formula obtained in the works \cite{DanielVosko-60} by Daniel and Voskov for electron gas in high density limit and \cite{Lam-71} by Lam for the same system at metallic densities.  In Appendix \ref{sec:compare-DV}, we compare our result with formulas in \cite{DanielVosko-60,Lam-71}.

\begin{remark}
	Although Coulomb interaction is of the most physical significance, our analysis extends to a more general class of interactions satisfying \eqref{square-sum}.  If the potential is more regular, then it is possible to improve the bound on errors.  If, for instance, the Fourier mode of interaction potential is compactly supported, we refer to \cite{BeneLill-25} and mention our first term in \eqref{momentum-trial} corresponds exactly to the result in \cite[Theorem 1.1]{BeneLill-25}.  
\end{remark}

\begin{remark}
During the completion of this work, we learned that a similar result was recently derived by Benedikter, Lill and Naidu in \cite{BenLillNa-25}, and these two works are conducted independently and in parallel.  

For comparison, in \cite{BenLillNa-25}, the authors computed momentum distribution at single point as in \eqref{momentum}, while in the present work, we consider a more general form \eqref{n(f)2}, computing expectation value of observable $f$.  In the case of $f\in\ell^{\infty}(\Z^{3})$, results in \cite{BenLillNa-25} and ours are equivalent to each other in the sense that \cite{BenLillNa-25} considered the special case of our result in \eqref{n(f)} with $f$ is taken to be the delta function at single point $\xi\in\Z^{3}$, while, by taking summation as in \eqref{n(f)2}, one can obtain a similar formula as in \eqref{n(f)} from results in \cite{BenLillNa-25}.  

Moreover, concerning error estimates, \cite{BenLillNa-25} presented two types of estimates, one considering a class of potentials as singular as Coulomb interactions while the other considering a class of more regular potentials (e.g. $\sum_{k\in\Z_{*}^{3}}\hat{V}_{k}<\infty$).  For the former class, they obtained an error estimate $O(k_{F}^{-7/6+\delta})$ for any $\delta>0$, whereas the error estimate is much stronger and is of order $O(k_{F}^{-2+\delta})$ for the latter class.  In the present work, we consider a more general class of singular potentials satisfying \eqref{square-sum} (including Coulomb interactions) and obtain a stronger error estimate, of order $O(k_{F}^{-3/2+\delta})$, without referencing specific structure of $\hat{V}$ other than its $\ell^{2}$-norm.  The method in the present work leads to the same upper bound for error terms when $\hat{V}$ is more regular so we omit this case and refer to \cite{BenLillNa-25} for details in this direction.   
\end{remark}

\begin{remark}
Consider Coulomb interaction $\hat{V}_{k}=g|k|^{-2}$ for a constant $g>0$ and $f(\xi)=\chi_{B_{F}}(\xi)$.  Then the quantity 
\begin{align*}
	n(f)&=\sum_{\xi\in B_{F}}n(\xi)=\sum_{\xi\in B_{F}}\langle\Psi_{N},\rho_{\xi}\Psi_{N}\rangle
\end{align*}
measures the number of excited particle-hole pairs.  We argue formally that our main result implies
\begin{align}
	n_{\rm b}(f)\sim O(1),\quad n_{\rm ex}(f)\sim O(1),\quad \big|\cE(f)\big|\leq C_{\delta,V}k_{F}^{-1/2+\delta}.
\end{align}
for some constant $C_{\delta,V}$ depends only on $\delta$ and $\|\hat{V}\|_{\ell^{2}}$.  For the error term, according to \cite[Lemma 3.2]{CHN-24-1}, we obtain immediately that $\sum_{\xi\in B_{F}}m(\xi)\leq C_{\delta}k_{F}^{1+\delta}$ for each $\delta>0$, which implies that
\begin{align}
	\big|\cE(f)\big|&\leq C_{\delta,V}k_{F}^{-3/2+\delta}\sum_{\xi\in B_{F}}m(\xi)\leq C_{\delta,V}k_{F}^{-1/2+\delta}.
\end{align}
For the exchange term, by shifting $\xi\mapsto k+\xi=:q$ for each fixed $k\in\Z_{*}^{3}$, we obtain
\begin{align}
	n_{\rm ex}(f)&=-\frac{k_{F}^{-2}}{4(2\pi)^{6}}\sum_{k\in\Z_{*}^{3}}\sum_{p,q\in L_{k}}\frac{\hat{V}_{k}\hat{V}_{p+q-k}}{(\lambda_{k,p}+\lambda_{k,q})^{2}}.
\end{align}
We observe that this is in the same form as the exchange correlation $E_{\rm corr,ex}$ in \eqref{ex-energy}, except an additional factor $(\lambda_{k,p}+\lambda_{k,q})^{-1}$ in each summand.  Using the approximation $\lambda_{k,p}\sim |k|\max\{|k|,k_{F}\}$ (in an average sense) and $|L_{k}|\sim k_{F}^{2}\min\{|k|,k_{F}\}$, it was argued in \cite[Remark 1]{CHN-23} that $E_{\rm corr,ex}\sim O(k_{F})$.  The same approxiamtion yields $n_{\rm ex}(f)\sim O(1)$.  For bosonization contribution, from the expansion
\begin{align*}
	\frac{1}{1+x}\approx 1-x+O(x^{2})_{x\rightarrow 0},
\end{align*}
we obtain, using the above approximation, 
\begin{align}
	n_{\rm b}(f)&=\frac{k_{F}^{-1}}{8\pi^{4}}\sum_{k\in\Z_{*}^{3}}\sum_{q\in L_{k}}\int_{0}^{\infty}\frac{(s^{2}-\lambda_{k,q}^{2})(s^{2}+\lambda_{k,q}^{2})^{-2}}{1+Q_{k}(s)}ds\nonumber\\
	&\approx \frac{k_{F}^{-1}}{8\pi^{4}}\sum_{k\in\Z_{*}^{3}}\sum_{q\in L_{k}}\hat{V}_{k}\int_{0}^{\infty}\frac{s^{2}-\lambda_{k,q}^{2}}{(s^{2}+\lambda_{k,q}^{2})^{2}}\Big[1-\frac{k_{F}^{-1}\hat{V}_{k}}{(2\pi)^{3}}\sum_{p\in L_{k}}\frac{\lambda_{k,p}}{s^{2}+\lambda_{k,p}^{2}}\Big]ds\nonumber\\
	&=\frac{\pi^{2}}{8(2\pi)^{7}}\sum_{k\in\Z_{*}^{3}}(k_{F}^{-1}\hat{V}_{k})^{2}\sum_{p,q\in L_{k}}\frac{1}{(\lambda_{k,p}+\lambda_{k,q})^{2}}\DETAILS{\nonumber\\
	&\approx\frac{g^{2}\pi^{2}}{8(2\pi)^{7}}\sum_{k\in\Z_{*}^{3}}\frac{k_{F}^{-2}}{|k|^{4}}\frac{k_{F}^{4}\min\{|k|^{2},k_{F}\}}{|k|^{2}\max\{|k|^{2},k_{F}^{2}\}}}\sim O(1).
\end{align}
Hence, we obtain the asymptotic behavior $n_{\rm b}(f)\sim O(1)$.  We observe that, for general singular potentials satisfying \eqref{square-sum}, when $f=\chi_{B_{F}}$, the above formal argument shows that contributions from bosonization and exchange correlation are at least of order $O(1)$, whereas the error term always of order $O(k_{F}^{-1/2+\delta})$.  Hence, Theorem \ref{thm:main} always yields a meaningful result.

Finally, we emphasize again the above estimates are purely formal, and it is not clear how to improve the upper bound for exchange term claimed in our main result under the general condition \eqref{square-sum} (and finding a rigorous, non-trivial lower bound on the exchange term for this general class of potentials remains an interesting open problem).  
\end{remark}


\begin{remark}
	The lower bound $\lambda_{k,p}\geq 1/2$ (or equivalently, the gap between lattice points inside and outside of Fermi ball is of order $O(1)$) is crucial in our analysis.  This breaks down in thermodynamic limit.  More precisely, by replacing the underlying configuration space $[0,1]^{3}\rightarrow[0,L]^{3}$, the corresponding momentum space $\Z^{3}$ becomes $L^{-1}\Z^{3}$ and thus $\lambda_{k,p}\sim O(L^{-1})$ for each $p\in L_{k}$, which vanishes in thermodynamic limit as $L\rightarrow\infty$.  As we will see in the proof, this leads to divergences in our estimates for error terms.
\end{remark}

\smallskip

\noindent\textbf{Strategy of proof}.  We will begin by computing the momentum distribution $n(\xi)$ in our trial state $\Psi_{N}$ using methods from \cite{CHN-23}, which outputs the contributions from bosonization and exchange correlations, together with some error terms.  We then show the error terms can be bounded by products of $k_{F}^{-1}$ and the following quantity
\begin{align*}
	\cQ:=\sup_{\xi\in\Z^{3}}\sup_{0\leq \tau\leq 1}\|\tilde{a}_{\xi}\Phi_{\tau}\|,\quad\text{ where }\Phi_{\tau}=e^{-\tau\cK}\Psi_{FS}.
\end{align*}
Clearly, this quantity is of order $O(1)$.  Nevertheless, by computing $\|\tilde{a}_{\xi}\Phi_{\tau}\|^{2}$ for each fixed $\xi\in\Z^{3}$ and $0\leq\tau\leq 1$ through Bogoliubov transformation, one can show it is bounded by $k_{F}^{-1}$ and $\cQ$ again, which, by a bootstrap argument as in \cite{BeneLill-25}, implies that $\cQ$ is of order $O(k_{F}^{-1/2+\delta})$ for each $\delta>0$.  One then can use this new bound of $\cQ$ to improve bounds for the error terms to obtain our claim result.


\section{Preliminaries}\label{sec:preliminaries}
In this section, we collect and establish some preliminary estimates that will be useful in the proof of our main result.  To simplify notation, in the rest of this paper, we denote any generic constant that depends on some parameters $a,b,c,...$ and $\|\hat{V}\|_{\ell^{2}}$ or $\|\hat{V}\|_{\ell^{1}}$ by $C_{a,b,c,..., V}$.  If the generic constant is independent of any paramter, then we denote it simply by $C$.  Moreover, we will always assume $\hat{V}$ satisfies \eqref{square-sum}.

\subsection{Quasi-bosonic operators}\label{sec:quasi-kernel}
For convenience, we introduce the unitary operator $I_{k}:\ell^{2}(L_{k})\rightarrow\ell^{2}(K_{-k})$ given by $I_{k}e_{p}=e_{-p}$ for each $k\in\Z_{*}^{3}$.  Note that the operators $K_{k}$ given in \eqref{Kk} satisfy the property
\begin{align}\label{unitary-Kk}
	I_{k}K_{k}=K_{-k}I_{k}.
\end{align}
\begin{lemma}[\cite{CHN-23}, Lemma 1.3]\label{lem:CR-excitation}
	For each $k,l\in\Z_{*}^{3}$, $\varphi\in\ell^{2}(L_{k})$ and $\psi\in\ell^{2}(L_{l})$, it holds that 
	\begin{align}
		\label{CR-excitation1}&[b_{k}(\varphi),b_{l}(\psi)]=[b_{k}^{*}(\varphi),b_{l}^{*}(\psi)]=0,\\
		\label{CR-excitation2}&[b_{k}(\varphi),b_{l}^{*}(\psi)]=\delta_{k,l}\langle\varphi,\psi\rangle+\varepsilon_{k,l}(\varphi;\psi),
	\end{align}
	where the {\rm exchange correction} is given by
	\begin{align}\label{exchange-correction}
		\varepsilon_{k,l}(\varphi;\psi)&=-\sum_{p\in L_{k}}\sum_{q\in L_{l}}\langle\varphi,e_{p}\rangle\langle e_{q},\psi\rangle\big(\delta_{p,q}a_{q-l}a_{p-k}^{*}+\delta_{p-k,q-l}a_{q}^{*}a_{p}\big).
	\end{align}
\end{lemma}
\begin{prop}[\cite{CHN-23}, Proposition A.2]\label{prop:Kcom-excit}
For any $k\in\Z_{*}^{3}$ and $\varphi\in \ell^{2}(L_{k})$, it holds that
\begin{align}\label{Kcom-excit}
	[\cK,b_{k}(\varphi)]&=b_{-k}^{*}(I_{k}K_{k}\varphi)+\cE_{k}(\varphi),\quad[\cK,b_{k}^{*}(\varphi)]=b_{-k}(I_{k}K_{k}\varphi)+\cE_{k}(\varphi)^{*},
\end{align}
where
\begin{align}
	\cE_{k}(\varphi)=\frac{1}{2}\sum_{l\in\Z_{*}^{3}}\sum_{q\in L_{l}}\{\varepsilon_{k,l}(\varphi;e_{q}),b_{-l}^{*}(K_{-l}e_{-q})\}.
\end{align}
\end{prop}
Furthermore, for any symmetric operator $T_{k}$ on $\ell^{2}(L_{k})$, we define the following associated quasi-bosonic operators:
\begin{align}
	Q_{1}^{k}(T_{k})\DETAILS{&=\sum_{p,q\in L_{k}}\langle e_{p},T_{k}e_{q}\rangle b_{k,p}^{*}b_{k,q}}=\sum_{p\in L_{k}}b_{k}^{*}(T_{k}e_{p})b_{k,p},
\end{align}
and
\begin{align}
	Q_{2}^{k}(T_{k})\DETAILS{&=\sum_{p,q\in L_{k}}\langle e_{p},T_{k}e_{q}\rangle\big(b_{k,p}b_{-k,-q}+b_{-k,-q}^{*}b_{k,p}^{*}\big)\nonumber\\}
	&=\sum_{p\in L_{k}}\big(b_{k}(T_{k}e_{p})b_{-k,-p}+b_{-k,-p}^{*}b_{k}^{*}(T_{k}e_{p})\big).
\end{align}

In the next lemma, we provide how these operators behave under quasi-Bogoliubov transformation:
\begin{lemma}[\cite{CHN-23}, Proposition A.5]\label{lem:exp-Bog-Q12}
	For any $k\in\Z_{*}^{3}$ and symmetric operators $T_{\pm k}:\ell^{2}(L_{\pm k})\rightarrow \ell^{2}(L_{\pm k})$ such that $I_{k}T_{k}=T_{-k}I_{k}$, it holds that 
	\begin{align}\label{exp-Bog-Q1}
			&e^{t\cK}(2Q_{1}^{k}(T_{k})+2Q_{1}^{-k}(T_{-k}))e^{-t\cK}=\Tr(T_{k}^{1}(t)-T_{k})+2Q_{1}^{k}(T_{k}^{1}(t))+Q_{2}^{k}(T_{k}^{2}(t))\\
			&\quad+\int_{0}^{t}e^{(t-\tau)\cK}\Big(\varepsilon_{k}(\{K_{k},T_{k}^{2}(\tau)\})+2\Re(\cE_{1,k}(T_{k}^{1}(\tau)))+2\Re(\cE_{2,k}(T_{k}^{2}(\tau)))\Big)e^{-(t-\tau)\cK}d\tau\nonumber\\
			&\quad\quad\quad\quad\quad+(k\rightarrow -k),\nonumber
			\end{align}
		and
	\begin{align}\label{exp-Bog-Q2}
		&e^{t\cK}(Q_{2}^{k}(T_{k})+Q_{2}^{-k}(T_{-k}))e^{-t\cK}=\Tr(T_{k}^{2}(t))+Q_{1}^{k}(T_{k}^{2}(t))+Q_{2}^{k}(T_{k}^{1}(t))\\
		&\quad+\int_{0}^{t}e^{(t-\tau)\cK}\Big(\varepsilon_{k}(\{K_{k},T_{k}^{1}(\tau)\})+2\Re(\cE_{1,k}(T_{k}^{2}(\tau)))+2\Re(\cE_{2,k}(T_{k}^{1}(\tau)))\Big)e^{-(t-\tau)\cK}d\tau\nonumber\\
		&\quad\quad\quad\quad\quad+(k\rightarrow -k),\nonumber
	\end{align}
	where, for each symmetric operator $A_{k}$ on $\ell^{2}(L_{k})$, we define
	\begin{align*}
		\varepsilon_{k}(A_{k})&=-\sum_{p\in L_{k}}\langle e_{p},A_{k}e_{p}\rangle\big(a_{p}^{*}a_{p}+a_{p-k}a_{p-k}^{*}\big),\\
		\cE_{1,k}(A_{k})&=\sum_{l\in\Z_{*}^{3}}\sum_{p\in L_{k}}\sum_{q\in L_{l}}b_{k}^{*}(A_{k}e_{p})\{\varepsilon_{k,l}(e_{p};e_{q}),b_{-l}^{*}(K_{-l}e_{-q})\},\\
		\cE_{2,k}(A_{k})&=\frac{1}{2}\sum_{l\in\Z_{*}^{3}}\sum_{p\in L_{k}}\sum_{q\in L_{l}}\{b_{k}(A_{k}e_{p}),\{\varepsilon_{-k,-l}(e_{-p};e_{-q}),b_{l}^{*}(K_{l}e_{q})\}\},
	\end{align*}
	and, for $\tau\in [0,t]$, $T_{k}^{1}(\tau)$ and $T_{k}^{2}(\tau)$ are symmetric operators on $\ell^{2}(L_{k})$ given by
	\begin{align*}
		T_{k}^{1}(\tau)&=\frac{1}{2}\Big(e^{\tau K_{k}}T_{k}e^{\tau K_{k}}+e^{-\tau K_{k}}T_{k}e^{-\tau K_{k}}\Big),\\ T_{k}^{2}(\tau)&=\frac{1}{2}\Big(e^{\tau K_{k}}T_{k}e^{\tau K_{k}}-e^{-\tau K_{k}}T_{k}e^{-\tau K_{k}}\Big).
	\end{align*}
\end{lemma}
\begin{proof}
	The proof is the same as in the proof of \cite[Proposition A.5]{CHN-23}, with an additional change of variable.  \DETAILS{We fix $t\in [0,1]$ and define $A_{k}(\tau)=T_{k}^{1}(t\tau)$ and $B_{k}(\tau)=T_{k}^{2}(t\tau)$, satisfying
		$$	A_{k}'(\tau)=t\{K_{t},B_{k}(\tau)\},\quad B_{k}'(\tau)=t\{K_{t},A_{k}(\tau)\},$$
		$$A_{k}(0)=T_{k},\quad A_{k}(1)=T_{k}^{1}(t),\quad B_{k}(0)=0,\quad B_{k}(1)=T_{k}^{2}(t).$$
		By Lemma \ref{lem:Bog-Q1} and these relation, we obtain
		\begin{align*}
			&\frac{d}{d\tau}e^{-\tau t\cK}\Big(2Q_{1}^{k}(A_{k}(\tau))+Q_{2}^{k}(B_{k}(\tau))\Big)e^{\tau t\cK}+(k\rightarrow -k)\nonumber\\
			&=te^{-\tau t\cK}\Big(2Q_{1}^{k}(A_{k}'(\tau))+Q_{2}^{k}(B_{k}'(\tau))-[\cK,2Q_{1}^{k}(A_{k}(\tau))+Q_{2}^{k}(B_{k}(\tau))]\Big)e^{\tau t\cK}+(k\rightarrow -k)\nonumber\\
			&=-t\Big[\Tr(\{K_{k},B_{k}(\tau)\})-e^{\tau t\cK}\Big(\varepsilon_{k}(\{K_{k},B_{t}(\tau)\})+2\Re(\cE_{1,k}(A_{k}(\tau)))+2\Re(\cE_{2,k}(B_{k}(\tau)))\Big)e^{\tau t\cK}\Big]\nonumber\\
			&\quad\quad\quad\quad\quad+(k\rightarrow -k).
		\end{align*}
		We then conclude by the fundamental theorem of calculus and change of variables,
		\begin{align*}
			&e^{t\cK}(2Q_{1}^{k}(T_{k})+2Q_{1}^{-k}(T_{-k}))e^{-t\cK}=2Q_{1}^{k}(A_{k}(1))+Q_{2}^{k}(B_{k}(1))+t\int_{0}^{1}\Tr(\{K_{k},B_{k}(\tau)\})d\tau\nonumber\\
			&\quad+t\int_{0}^{1}e^{(1-\tau)t\cK}\Big(\varepsilon_{k}(\{K_{k},B_{k}(\tau)\})+2\Re(\cE_{1,k}(A_{k}(\tau)))+2\Re(\cE_{2,k}(B_{k}(\tau)))\Big)e^{(1-\tau)t\cK}d\tau\\
			&\quad\quad\quad\quad+(k\rightarrow -k)\\
			&=\Tr(T_{k}^{1}(t)-T_{k})+2Q_{1}^{k}(T_{k}^{1}(t))+Q_{2}^{k}(T_{k}^{2}(t))\\
			&\quad+\int_{0}^{t}e^{(t-\tau)\cK}\Big(\varepsilon_{k}(\{K_{k},T_{k}^{2}(\tau)\})+2\Re(\cE_{1,k}(T_{k}^{1}(\tau)))+2\Re(\cE_{2,k}(T_{k}^{2}(\tau)))\Big)e^{-(t-\tau)\cK}d\tau\\
			&\quad\quad\quad\quad+(k\rightarrow -k).
		\end{align*}
		where we have used the identity
		\begin{align*}
			\int_{0}^{t}\Tr(\{K_{k},T_{k}^{2}(\tau)\})d\tau&=\int_{0}^{t}\Tr((T_{k}^{1})'(\tau))d\tau=\Tr(T_{k}^{1}(t)-T_{k}).
		\end{align*}
		This proves \eqref{exp-Bog-Q1}.  The relation \eqref{exp-Bog-Q2} can be proven by the same procedure.}
\end{proof}

\subsection{Useful identities}\label{sec:useful}
In this subsection, we collect some useful identities from \cite{CHN-23} on the family $(K_{k})_{k\in\Z_{*}^{3}}$ we used to construct our trial state in Subsection \ref{sec:trial-states}.  These identities are crucial when we estimate error $\cE$ in Theorem \ref{thm:main}.  For most of them, we omit proofs for simplicities.

First, for any $|\tau|\leq 1$, we introduce the operator families
\begin{align}\label{CSkt}
	C_{k}(\tau)&=\cosh(-\tau K_{k})-1,\quad S_{k}(\tau)=\sinh(-\tau K_{k}),
\end{align}
and they satisfy the following elementary estimates:
\begin{lemma}[\cite{CHN-23}, Propositions 3.4 and 3.5]\label{lem:matrix-expK}
	For each $k\in\Z_{*}^{3}$, $0\leq \tau\leq 1$ and $p,q\in\ell^{2}(L_{k})$, it holds that
	\begin{align}\label{matrix-K}
		\frac{1}{1+2\langle v_{k}h_{k}^{-1}v_{k}\rangle}\frac{\langle e_{p},v_{k}\rangle\langle v_{k},e_{q}\rangle}{\lambda_{k,p}+\lambda_{k,q}}&\leq \langle e_{p},(-K_{k})e_{q}\rangle\leq \frac{\langle e_{p},v_{k}\rangle\langle v_{k},e_{q}\rangle}{\lambda_{k,p}+\lambda_{k,q}}.
	\end{align}
	\begin{align}
		\frac{1}{1+2\langle v_{k},h_{k}^{-1}v_{k}\rangle}\frac{\langle e_{p},v_{k}\rangle\langle v_{k},e_{q}\rangle}{\lambda_{k,p}+\lambda_{k,q}}\tau\leq \langle e_{p},S_{k}(\tau)e_{q}\rangle&\leq \frac{\langle e_{p},v_{k}\rangle\langle v_{k},e_{q}\rangle}{\lambda_{k,p}+\lambda_{k,q}}\tau,
	\end{align}
	and
	\begin{align}\label{Cktau-ubdd}
		0\leq \langle e_{p},C_{k}(\tau)e_{q}\rangle&\leq \frac{\langle v_{k},h_{k}^{-1}v_{k}\rangle}{1+2\langle v_{k},h_{k}^{-1}v_{k}\rangle}\frac{\langle e_{p},v_{k}\rangle\langle v_{k},e_{q}\rangle}{\lambda_{k,p}+\lambda_{k,q}}.
	\end{align}
	Consequently, $\|K_{k}\|_{\rm HS}\leq C\hat{V}_{k}$ for some constant $C>0$.
\end{lemma}
We observe the operators $T_{k}^{1}(\tau)$ and $T_{k}^{2}(\tau)$ in Lemma \ref{lem:exp-Bog-Q12} can be decomposed as
\begin{align}\label{T1-decomp}
	T_{k}^{1}(\tau)&=T_{k}+\{T_{k},C_{k}(\tau)\}+C_{k}(\tau)T_{k}C_{k}(\tau)+S_{k}(\tau)T_{k}S_{k}(\tau),\\
	\label{T2-decomp}T_{k}^{2}(\tau)&=-\{T_{k},S_{k}(\tau)\}-S_{k}(\tau)T_{k}C_{k}(\tau)-S_{k}(\tau)T_{k}C_{k}(\tau),
\end{align}

Next, we provide the following important lemma on lattice estimates:
\begin{lemma}[\cite{CHN-24}, Proposition A.1]\label{lem:estimate-lambkp}
	For any $k\in\Z_{*}^{3}$ and $\beta\in [-1,0]$, it holds that
	\begin{align}\label{estimate-lambkp}
		\sum_{p\in L_{k}}\lambda_{k,p}^{\beta}&\leq C_{\beta}\begin{cases}
			k_{F}^{2+\beta}|k|^{1+\beta}\quad&\text{ for }|k|\leq 2k_{F},\\
			k_{F}^{3}|k|^{2\beta}&\text{ for }|k|>2k_{F},
		\end{cases}
	\end{align} 
	for some constant $C_{\beta}$ independent of $k_{F}$ and $k$.  In particular, it holds that
	\begin{align}
		\sum_{p\in L_{k}}\lambda_{k,p}^{-1}&\leq C k_{F}\min\{1,k_{F}^{2}|k|^{-2}\}.
	\end{align}
\end{lemma}
By combining Lemmas \ref{lem:matrix-expK} and \ref{lem:estimate-lambkp}, we obtain the following proposition:
\begin{prop}\label{prop:prod-Kk-esti}
Let $A_{k}^{(m)}(\tau)$ be any $m$-fold product of operators from the set $\{C_{k}(\tau),S_{k}(\tau),K_{k}\}$.  Then, for each $k\in\Z_{*}^{3}$, $0\leq t\leq 1$ and $p,q\in \ell^{2}(L_{k})$, it holds that
\begin{align}\label{prod-Kk-esti}
	\big|\langle e_{p},A_{k}^{(m)}(t)e_{q}\rangle\big|&\leq \frac{C_{V,m}k_{F}^{-1}\hat{V}_{k}^{j_{1}}}{\lambda_{k,p}+\lambda_{k,q}}\min\{1,k_{F}^{2}|k|^{-2}\}
\end{align}
for some constant $C_{m,V}$ depends only on $m$ and $\|\hat{V}\|_{\ell^{2}}$.
\end{prop}
\begin{proof}
We prove \eqref{prod-Kk-esti} by induction.  The case $m=1$ is done with $C_{V,1}=\max\{1,\|\hat{V}\|_{\infty}\}$ thanks to Lemma \ref{lem:matrix-expK}.  For $m>1$, we suppose the estimate \eqref{prod-Kk-esti} holds for $m=n$ with $n\geq 1$, and we prove for $m=n+1$.  We observe that
\begin{align}
	A_{k}^{(m)}&=A_{k}^{(m-1)}A_{k}^{(1)}.
\end{align}
Then, by inductive hypothesis, triangle inequality and Lemma \ref{lem:estimate-lambkp}, we obtain for any $p,q\in\ell^{2}(L_{k})$ that
\begin{align*}
	\big|\langle e_{p},A_{k}^{(m)}e_{q}\rangle\big|&=\big|\langle e_{p},A_{k}^{(m-1)}A_{k}^{(1)}e_{q}\rangle\big|\leq \sum_{q\in L_{k}}\big|\langle e_{p},A_{k}^{(m-1)}e_{r}\rangle\big|\big|\langle e_{r},A_{k}^{(1)}e_{q}\rangle\big|\nonumber\\
	&\leq\sum_{q\in L_{k}}\frac{C_{V,m-1}C_{V,1} k_{F}^{-2}\hat{V}_{k}^{j_{1}}}{(\lambda_{k,p}+\lambda_{k,r})(\lambda_{k,q}+\lambda_{k,r})}\leq \sum_{r\in L_{k}}\frac{C_{V,m}k_{F}^{-2}\hat{V}_{k}^{m}}{(\lambda_{k,p}+\lambda_{k,q})\lambda_{k,r}}\nonumber\\
	&\leq C_{V,m}C \frac{k_{F}^{-1}\hat{V}_{k}^{m}}{\lambda_{k,p}+\lambda_{k,q}}\min\{1,k_{F}^{2}|k|^{-2}\}.
\end{align*}
where $C_{V,m}=C_{V,m-1}C_{V,1}$ and $C$ comes from \eqref{prod-Kk-esti}, which completes the proof.
\end{proof}

To conclude this section, we denote
\begin{align}
	\cN_{k}:=\sum_{p\in L_{k}}b_{k,p}^{*}b_{k,p},\quad\cN_{E}:=\sum_{p\in \Z^{3}}\tilde{a}_{p}^{*}\tilde{a}_{p},
\end{align}
where
\begin{align}\label{particle-hole}
	\tilde{a}_{p}&=\chi_{B_{F}}(p)a_{p}^{*}+\chi_{B_{F}^{c}}(p)a_{p}.
\end{align}
According to \cite[Proposition 4.4]{CHN-23}, for any $\Psi\in\cH_{N}$ and $\varphi\in\ell^{2}(L_{k})$, we have
\begin{align}\label{excitation-bdd}
	\|b_{k}(\varphi)\Psi\|&\leq \|\varphi\|\|\cN_{k}^{1/2}\Psi\|,\quad\|b_{k}^{*}(\varphi)\Psi\|\leq \|\varphi\|\|(\cN_{k}+1)^{1/2}\Psi\|.
\end{align}
Note that our definition of $\cN_{E}$ is twice of the one defined in \cite{CHN-23} on $\cH_{N}$ (due to particle-hole symmetry).  Nevertheless, we can still obtain similar operator inequalities as in \cite{CHN-23}:
\begin{align}\label{NkNE-esti}
	\cN_{k}\leq \cN_{E}\text{ for each }k\in\Z_{*}^{3},\quad \sum_{k\in\Z_{*}^{3}}\cN_{k}\leq \cN_{E}^{2}.
\end{align}
The first relation in \eqref{NkNE-esti} is straightforward to see.  To show the second relation, we derive the following \textit{pull-through formula}:
\begin{lemma}\label{lem:pull-through}
	For any $p\in \Z^{3}$, it holds that 
	\begin{align}\label{pull-through}
		\cN_{E}\tilde{a}_{p}&=\tilde{a}_{p}(\cN_{E}-1),\quad\tilde{a}_{p}^{*}\cN_{E}=(\cN_{E}-1)\tilde{a}_{p}^{*}.
	\end{align}
\end{lemma}
\begin{proof}
By CAR, we obtain
\begin{align*}
	\cN_{E}\tilde{a}_{p}&=\sum_{q\in\Z^{3}}\tilde{a}_{q}^{*}\tilde{a}_{q}\tilde{a}_{p}=-\sum_{q\in\Z^{3}}\tilde{a}_{q}^{*}\tilde{a}_{p}\tilde{a}_{q}=\tilde{a}_{p}\cN_{E}-\sum_{q\in\Z^{3}}\delta_{p,q}\tilde{a}_{q}=\tilde{a}_{p}(\cN_{E}-1),
\end{align*}
which proves the first relation in \eqref{pull-through}.  Taking adjoint yields the second relation.  This completes the proof.
\end{proof}
\begin{cor}\label{cor:sumNk-NE}
It holds that
\begin{align}\label{sumNk-NE}
	\sum_{k\in\Z_{*}^{3}}\cN_{k}\leq \cN_{E}^{2}.
\end{align}
\end{cor}
\begin{proof}
By rearranging of summation as in \cite[Eq. (4.18)]{CHN-23} and using the fact that $[\tilde{a}_{p}^{*}\tilde{a}_{p},\cN_{E}]=0$ for any $p\in\Z^{3}$, we obtain
\begin{align}
	\sum_{k\in\Z_{*}^{3}}\cN_{k}&=\sum_{k\in\Z_{*}^{3}}\sum_{p\in L_{k}}\tilde{a}_{p}^{*}\tilde{a}_{p-k}^{*}\tilde{a}_{p-k}\tilde{a}_{p}=\sum_{p\in B_{F}^{c}}\tilde{a}_{p}^{*}\tilde{a}_{p}\sum_{k\in (B_{F}+p)}\tilde{a}_{p-k}^{*}\tilde{a}_{p-k}\nonumber\\
	&\leq \sum_{p\in B_{F}^{c}}\tilde{a}_{p}^{*}\tilde{a}_{p}\sum_{k\in(\Z^{3}+p)}\tilde{a}_{p-k}^{*}\tilde{a}_{p-k}=\sum_{p\in B_{F}^{c}}\tilde{a}_{p}^{*}\tilde{a}_{p}\cN_{E}\nonumber\\
	&\leq \cN_{E}^{1/2}\Big(\sum_{p\in B_{F}^{c}}\tilde{a}_{p}^{*}\tilde{a}_{p}\Big)\cN_{E}^{1/2}\leq \cN_{E}^{1/2}\Big(\sum_{p\in\Z^{3}}\tilde{a}_{p}^{*}\tilde{a}_{p}\Big)\cN_{E}^{1/2}=\cN_{E}^{2},
\end{align}
which completes the proof.
\end{proof}

\begin{cor}\label{cor:pull-through}
For each $p\in\Z^{3}$ and $\Psi\in\cH_{N}$, it holds that 
\begin{align}\label{excitation-bdd2}
	\|\cN_{k}^{1/2}\tilde{a}_{p}\Psi\|&\leq \|\tilde{a}_{p}\cN_{k}^{1/2}\Psi\|\leq \|\tilde{a}_{p}\cN_{E}^{1/2}\Psi\|,\quad \|\cN_{E}^{1/2}\tilde{a}_{p}\Psi\|\leq \|\tilde{a}_{p}\cN_{E}^{1/2}\Psi\|,\\
	\label{excitation-bdd5}\|\cN_{E}\tilde{a}_{p}\Psi\|&\leq \|\tilde{a}_{p}\cN_{E}\Psi\|,\quad\|\cN_{E}^{3/2}\tilde{a}_{p}\Psi\|\leq \|\tilde{a}_{p}\cN_{E}^{3/2}\Psi\|+\|\tilde{a}_{p}\cN_{E}^{1/2}\Psi\|.
\end{align}
\end{cor}
\begin{proof}
The relations in \eqref{excitation-bdd2} is proved in \cite{CHN-23}.  For the first relation in \eqref{excitation-bdd5}, by Lemma \ref{lem:pull-through}, we have for each $p\in\Z^{3}$ that
\begin{align}\label{Na-aN-esti}
	\cN_{E}^{2}\tilde{a}_{p}&=\cN_{E}\tilde{a}_{p}(\cN_{E}-1)=\tilde{a}_{p}\cN_{E}^{2}-\tilde{a}_{p}\cN_{E}-\cN_{E}\tilde{a}_{p}.
\end{align}
Since $[\tilde{a}_{p}^{*}\tilde{a}_{p},\cN_{E}]=0$ for any $p\in\Z^{3}$, Eq. \eqref{Na-aN-esti} implies that 
\begin{align}
	\tilde{a}_{p}^{*}\cN_{E}^{2}\tilde{a}_{p}&\leq \tilde{a}_{p}^{*}\tilde{a}_{p}\cN_{E}^{2}=\cN_{E}\tilde{a}_{p}^{*}\tilde{a}_{p}\cN_{E},
\end{align}
which in turns implies that 
\begin{align}
	\|\cN_{E}\tilde{a}_{p}\Psi\|^{2}&=\langle\Psi,\tilde{a}_{p}^{*}\cN_{E}\tilde{a}_{p}\Psi\rangle\leq\langle\Psi,\cN_{E}\tilde{a}_{p}^{*}\tilde{a}_{p}\cN_{E}\Psi\rangle=\|\tilde{a}_{p}\cN_{E}\Psi\|^{2}.
\end{align}
This gives the first relation in \eqref{excitation-bdd5}.  For the second relation in \eqref{excitation-bdd5}, we again apply Lemma \ref{lem:pull-through} to obtain
\begin{align}
	\cN_{E}^{3}\tilde{a}_{p}&=\cN_{E}\tilde{a}_{p}\cN_{E}^{2}-\cN_{E}^{2}\tilde{a}_{p}-\cN_{E}\tilde{a}_{p}\cN_{E}=\tilde{a}_{p}\cN_{E}(\cN_{E}^{2}+1)-\tilde{a}_{p}\cN_{E}^{2}-\cN_{E}^{2}\tilde{a}_{p}.
\end{align}
Hence, by the relation $[\tilde{a}_{p}^{*}\tilde{a}_{p},\cN_{E}]=0$ again, we have
\begin{align}
	\tilde{a}_{p}^{*}\cN_{E}^{3}\tilde{a}_{p}&=\tilde{a}_{p}^{*}\tilde{a}_{p}\cN_{E}^{3}+\tilde{a}_{p}^{*}\tilde{a}_{p}\cN_{E}-\tilde{a}_{p}^{*}\tilde{a}_{p}\cN_{E}^{2}-\tilde{a}_{p}^{*}\cN_{E}^{2}\tilde{a}_{p}\nonumber\\
	&=\cN_{E}^{3/2}\tilde{a}_{p}^{*}\tilde{a}_{p}\cN_{E}^{3/2}+\cN_{E}^{1/2}\tilde{a}_{p}^{*}\tilde{a}_{p}\cN_{E}^{1/2}-\cN_{E}\tilde{a}_{p}^{*}\tilde{a}_{p}\cN_{E}-\tilde{a}_{p}^{*}\cN_{E}^{2}\tilde{a}_{p}\nonumber\\
	&\leq \cN_{E}^{3/2}\tilde{a}_{p}^{*}\tilde{a}_{p}\cN_{E}^{3/2}+\cN_{E}^{1/2}\tilde{a}_{p}^{*}\tilde{a}_{p}\cN_{E}^{1/2}.
\end{align}
This gives the second relation in \eqref{excitation-bdd5}, which completes the proof.
\end{proof}

Finally, we prove the following important estimates:
\begin{lemma}\label{lem:gap-sum}
For any $\xi\in\Z^{3}$, it holds for any $\delta>0$ that 
\begin{align}\label{gap-sum}
	\sum_{k\in\Z_{*}^{3}}\frac{\chi_{L_{k}}(\xi)}{\lambda_{k,\xi}}&\leq C_{\delta}k_{F}^{1+\delta},\quad
	\sum_{k\in\Z_{*}^{3}}\frac{\chi_{L_{k}'}(\xi)}{\lambda_{k,k+\xi}^{2}}\leq C_{\delta}k_{F}^{1+\delta}m(\xi),
\end{align}
for some constant $C_{\delta}$ depends only on $\delta$ and is independent of $\xi$.
\end{lemma}
\begin{proof}
First, for $\xi\in B_{F}^{c}$, the second relation in \eqref{gap-sum} vanishes so we only need to consider the first one.  We denote the region $D_{\xi}:=\{k\in\Z_{*}^{3}\mid \xi\in L_{k}\}$.  We observe that $D_{\xi}\subseteq \overline{B(\xi,k_{F})}\cap \Z^{3}$, where $B_{R}(a)$ denotes the open ball centered at $a$ with radius $R>0$.  Clearly, $|D_{\xi}|\leq |\overline{B}(0,2k_{F})\cap\Z^{3}|$.  Hence, by \cite[Lemma 3.2]{CHN-24-1} and lower bound \eqref{gap}, for any $\delta>0$, there is some $C_{\delta}>0$ such that
\begin{align*}
	\sum_{k\in \Z_{*}^{3}}\frac{\chi_{L_{k}}(\xi)}{\lambda_{k,\xi}}&\leq \sum_{k\in D_{\xi}}\frac{2}{\big||\xi-k|^{2}-\kappa\big|}\leq \sum_{p\in B_{k_{F}}(0)\cap\Z^{3}}\frac{2}{\big||p|^{2}-\kappa\big|}\leq C_{\delta}k_{F}^{1+\delta},
\end{align*}
for some constant $C_{\delta}$ depends only on $\delta$ and is independent of $\xi$.

Next, for $\xi\in B_{F}$, again, the first relation in \eqref{gap-sum} vanishes so we only need to consider the second one.  In this case, we define the region $D_{\xi}':=\{k\in\Z_{*}^{3}\mid \xi\in L_{k}'\}$.  We observe that $|k+\xi|>k_{F}$ if $k\in D_{\xi}'$, and decompose $D_{\xi}':=D_{\xi,1}'\cup D_{\xi,2}'$ as
\begin{align*}
	D_{\xi,1}'&:=\{k\in D_{\xi}'\mid |k+\xi|\leq 2k_{F}\},\quad D_{\xi,2}':=D_{\xi}'\setminus D_{\xi,1}'.
\end{align*}
We observe that $D_{\xi,1}'\subseteq\overline{B(\xi,2k_{F})}\cap B(\xi,k_{F})^{c}\cap \Z^{3}$ so that
\begin{align}
	|D_{\xi,1}'|&\leq C|\overline{B(0,2k_{F})}\cap\Z^{3}|,
\end{align}
for some universal constant $C>0$.  We further decompose $D_{\xi,1}'=\cup_{\ell=1}^{L}D_{\xi,1,n}'$ for some integer $C\leq L\leq C+1$ such that $|D_{\xi,1,\ell}'|\leq |\overline{B(0,2k_{F})}\cap \Z^{3}|$ for each $n$ so that, by \cite[Lemma 3.2]{CHN-24-1} and lower bound \eqref{gap},
\begin{align*}
	\sum_{k\in D_{\xi,1}'}\frac{1}{\lambda_{k,k+\xi}^{2}}&\leq m(\xi) \sum_{\ell=1}^{L}\sum_{k\in D_{\xi,1,\ell}'}\frac{2}{\big||k+\xi|^{2}-\kappa\big|}\nonumber\\
	&\leq m(\xi)\sum_{p\in B_{k_{F}}(0)\cap\Z^{3}}\frac{2L}{\big||p|^{2}-\kappa\big|}\leq (C+1)C_{\delta}k_{F}^{1+\delta}m(\xi).
\end{align*}
For region $D_{\xi,2}'$, since $\kappa=k_{F}^{2}(1+o(1))$ and $|k+\xi|^{2}\geq 4k_{F}^{2}$ for any $k\in D_{\xi,2}'$, we have
\begin{align*}
	\big||k+\xi|^{2}-\kappa\big|&\geq \frac{1}{2}|k+\xi|^{2}-o(1).
\end{align*}
Moreover, we observe that, for each$\xi\in B_{F}$, there is some universal constant $C'>0$ such that
\begin{align}
	m(\xi)^{-1}=\big||\xi|^{2}-\kappa\big|\leq |\xi|^{2}+\kappa\leq C'k_{F}^{2}(1+o(1)).
\end{align} 
Then we bound the sum over region $D_{\xi,2}'$ with the corresponding integral (up to some constant $C''>0$) to obtain
\begin{align}\label{ex-esti-regionD2}
	\sum_{k\in D_{\xi,2}'}\frac{1}{\lambda_{k,k+\xi}^{2}}&\leq\sum_{k\in D_{\xi,2}'}\frac{1}{\big||k+\xi|^{2}-\kappa\big|^{2}}\leq C''\int_{k_{F}}^{\infty}\frac{p^{2}dp}{(p^{2}-o(1))^{2}}\nonumber\\
	&\leq C'' k_{F}^{-1}\leq C'C''m(\xi)\big(m(\xi)^{-1}k_{F}^{-1}\big)\leq C'C''k_{F}(1+o(1))m(\xi).
\end{align}
Consequently, we obtain for any $\delta>0$ that
\begin{align}\label{ex-esti2}
	\sum_{k\in\Z_{*}^{3}}\frac{\chi_{L_{k}'}(\xi)}{\lambda_{k,k+\xi}^{2}}\leq C_{\delta}k_{F}^{1+\delta}m(\xi),
\end{align}
for some constant $C_{\delta}$ again depends only on $\delta$.  This completes the proof.
\end{proof}

\section{Momentum distribution: Bosonization and exchange contributions}\label{sec:momentum-boson-ex}
In this section, we compute the momentum distribution $n(\xi)$ of our trial state and extract the contribution due to bosonization and exchange correlations.  First, instead of only $n(\xi)$, we consider the family
\begin{align}
	n_{t}(\xi)&=\langle \Phi_{t},\rho_{\xi}\Phi_{t}\rangle,\quad \Phi_{t}:=e^{-t\cK}\Psi_{FS}\text{ for }0\leq t\leq 1,
\end{align}
so that $n(\xi)\equiv n_{1}(\xi)$.  Moreover, due to the reflection symmetry of trial state $\Phi_{t}$ for each $0\leq t\leq 1$ in the momentum space, 
\begin{align}\label{ref-aveg}
	n_{t}(\xi)&=\langle\Phi_{t},\rho_{\xi}\Phi_{t}\rangle=\frac{1}{2}\langle\Phi_{t},(\rho_{\xi}+\rho_{-\xi})\Phi_{t}\rangle.
\end{align}
This symmetry condition is important when we compute momentum distribution (see Remark \ref{rem:symmetric-reason} for explanation).

Next, we compute the occupation number of a single excitation mode:
\begin{lemma}\label{lem:single-mode}
	Let $k\in\Z_{*}^{3}$ and $p\in L_{k}$.  Then we have
	\begin{align}\label{single-mode}
		&\frac{1}{2}[b_{k,p},\rho_{\xi}+\rho_{-\xi}]=(g_{\xi,k}(p)+g_{-\xi,k}(p))b_{k,p},
	\end{align}
	where
	\begin{align}
		g_{\xi,k}(p)&=\frac{1}{2}\big(\chi_{L_{k}}(\xi)\delta_{p,q}+\chi_{L_{k}'}(\xi)\delta_{p,k+\xi}\big).
	\end{align}
\end{lemma}
\begin{proof}
	This is proven by direct application of CAR.
	\DETAILS{We will only show the first relation in \eqref{single-mode} since the other one can be shown by the same computation.  For each $k\in\Z_{*}^{3}$ and $p\in L_{k}$, by CAR, we obtain
		\begin{align*}
			[a_{\pm\xi}^{*}a_{\pm\xi},a_{p}]=0\xi\in B_{F}^{c},\quad[a_{\pm\xi}a_{\pm\xi}^{*},a_{p-k}^{*}]\xi\in B_{F},
		\end{align*}
		so that
		\begin{align}\label{single-mode2}
			2[b_{k,p},\rho_{\xi}]&=\chi_{B_{F}^{c}}(\xi)[a_{p-k}^{*}a_{p},a_{\xi}^{*}a_{\xi}+a_{-\xi}^{*}a_{-\xi}]+\chi_{B_{F}}(\xi)[a_{p-k}^{*}a_{p},a_{\xi}a_{\xi}^{*}+a_{-\xi}a_{-\xi}^{*}]\nonumber\\
			&=\chi_{B_{F}^{c}}(\xi)a_{p-k}^{*}[a_{p},a_{\xi}^{*}a_{\xi}+a_{-\xi}^{*}a_{-\xi}]+\chi_{B_{F}}(\xi)[a_{p-k}^{*},a_{\xi}a_{\xi}^{*}+a_{-\xi}a_{-\xi}^{*}]a_{p}.
		\end{align}
		By CAR, for $\xi\in B_{F}^{c}$, we have
		\begin{align*}
			[a_{p},a_{\pm\xi}^{*}a_{\pm\xi}]&=a_{p}a_{\pm\xi}^{*}a_{\pm\xi}-a_{\pm\xi}^{*}a_{\pm\xi}a_{p}=\delta_{p,\pm\xi}a_{\pm\xi},
		\end{align*}
		and, for $\xi\in B_{F}$, we have
		\begin{align*}
			[a_{p-k}^{*},a_{\pm\xi}a_{\pm\xi}^{*}]&=a_{p-k}^{*}a_{\pm\xi}a_{\pm\xi}^{*}-a_{\pm\xi}a_{\pm\xi}^{*}a_{p-k}^{*}=\delta_{p-k,\pm\xi}a_{\pm\xi}^{*}.
		\end{align*}
		Together with the fact that $p\in L_{k}$, substituting them into \eqref{single-mode2} yields
		\begin{align*}
			2[b_{k,p},\rho_{\xi}]&=\chi_{B_{F}^{c}}(\xi)a_{p-k}^{*}\big(\delta_{p,\xi}a_{\xi}+\delta_{p,-\xi}a_{-\xi}\big)+\chi_{B_{F}}(\xi)\big(\delta_{p-k,\xi}a_{\xi}^{*}+\delta_{p-k,-\xi}a_{\xi}^{*}\big)a_{p}\nonumber\\
			&=\Big[\chi_{B_{F}^{c}\cap L_{k}}(\xi)\big(\delta_{p,\xi}+\delta_{p,-\xi}\big)+\chi_{B_{F}\cap L_{k}'}(\xi)\big(\delta_{p,k+\xi}+\delta_{p,k-\xi}\big)\Big]b_{k,p}\nonumber\\
			&=\Big[\chi_{L_{k}}(\xi)\big(\delta_{p,\xi}+\delta_{p,-\xi}\big)+\chi_{L_{k}'}(\xi)\big(\delta_{p,k+\xi}+\delta_{p,k-\xi}\big)\Big]b_{k,p},
		\end{align*}
		which completes the proof.}
\end{proof}
This lemma allows us to compute the commutator $[\cK,\rho_{\xi}]$:
\begin{prop}\label{prop:1st-com}
	For each $\xi\in\Z^{3}$, we obtain
	\begin{align}\label{1st-com}
		\frac{1}{2}[\cK,\rho_{\xi}+\rho_{-\xi}]\DETAILS{&=\frac{1}{2}\sum_{k\in\Z_{*}^{3}}Q_{2}^{k}(\Theta_{\xi;k})\nonumber\\
			&=\frac{1}{2}\sum_{k\in\Z_{*}^{3}}\sum_{p,q\in L_{k}}\Theta_{\xi;k}(r,s)\Big(b_{-k,-q}^{*}b_{k,p}^{*}+b_{k,p}b_{-k,-q}\Big)\nonumber\\
			&=\frac{1}{2}\sum_{k\in\Z_{*}^{3}}\sum_{p\in L_{k}}\Big(b_{-k,-p}^{*}b_{k}^{*}(\Theta_{\xi;k}e_{p})+b_{k}(\Theta_{\xi;k}e_{p})b_{-k,-p}\Big)}=\frac{1}{2}\sum_{k\in\Z_{*}^{3}}Q_{2}^{k}(\Theta_{\xi;k}),\quad \Theta_{\xi;k}:=\frac{1}{2}\sum_{\zeta\in\cD_{k,\xi}}\chi_{L_{k}}(\zeta)\{P_{\zeta;k},K_{k}\}
	\end{align}
	where $\cD_{k,\xi}:=\{\pm\xi,k\pm\xi\}$ and $P_{\zeta;k}:=\ket{e_{\zeta}}\bra{e_{\zeta}}$ denotes the rank-1 projection onto the vector $e_{\zeta}$ on $\ell^{2}(L_{k})$, provided $\zeta\in L_{k}$.  Moreover, $\Theta_{\xi;k}$ is a symmetric operator on $\ell^{2}(L_{k})$ and satisfies the property
	\begin{align}\label{unitary-tileKqk}
		I_{k}\Theta_{\xi;k}&=\Theta_{\xi;-k}I_{k}\quad\text{ for each }k\in \Z_{*}^{3}.
	\end{align}
\end{prop}
\begin{proof}
	Using Lemma \ref{lem:single-mode}, we compute
	\begin{align*}
		&\frac{1}{2}[\cK,\rho_{\xi}+\rho_{-\xi}]=\frac{1}{4}\sum_{k\in\Z_{*}^{3}}\sum_{p,q\in L_{k}}\langle e_{p},K_{k}e_{q}\rangle\Big([b_{k,p}b_{-k,-q},\rho_{\xi}+\rho_{-\xi}]-[b_{-k,-q}^{*}b_{k,p}^{*},\rho_{\xi}+\rho_{-\xi}]\Big)\nonumber\\
		&=\Re\Big(\sum_{k\in\Z_{*}^{3}}\Big(\chi_{L_{k}}(\xi)b_{-k}^{*}(K_{-k}e_{-\xi})b_{k,q}^{*}+\chi_{L_{k}'}(\xi)b_{k}^{*}(K_{k}e_{k+\xi})b_{-k,-(k+\xi)}^{*}+(\xi\rightarrow -\xi)\Big)\nonumber\\
		&=\Re\sum_{k\in\Z_{*}^{3}}\Big(\chi_{L_{k}}(\xi)b_{-k}^{*}(I_{k}K_{k}e_{\xi})b_{k}^{*}(e_{\xi})+\chi_{L_{k}'}(\xi)b_{k}^{*}(K_{k}e_{k+\xi})b_{-k}^{*}(I_{k}e_{k+\xi})\Big)+(\xi\rightarrow -\xi)\\
		&=:I_{\xi}+II_{\xi}+(\xi\rightarrow -\xi).
	\end{align*}

	Next, we rewrite the above expression in a more convenient form using the following elementary lemma (which can be proven by orthonormal basis expansion):
	\begin{lemma}[\cite{CHN-23}, Lemma A.1]\label{lem:bilinear}
		Let $(V,q)$ be a $n$-dimenisonal Hilbert space and let $q:V\times V\rightarrow W$ be a sesquilinear form into a vector space $W$.  Let $(e_{i})_{i=1}^{n}$ be an orthonormal basis for $V$.  Then, for any linear operators $S,T:V\rightarrow V$, we have
		\begin{align}\label{bilinear}
			\sum_{i=1}^{n}q(Se_{i},Te_{i})&=\sum_{i=1}^{n}q(ST^{*}e_{i},e_{i}).
		\end{align}
	\end{lemma}
	By this lemma and using the symmetry $k\rightarrow -k$, we see that
	\begin{align*}
		I_{\xi}&=\Re\Big(\sum_{k\in\Z_{*}^{3}}\sum_{p\in L_{k}}b_{-k}^{*}(I_{k}K_{k}P_{\xi;k}e_{p})b_{k}^{*}(e_{p})\Big)=\Re\Big(\sum_{k\in\Z_{*}^{3}}\sum_{p\in L_{k}}b_{-k}^{*}(e_{-p})b_{k}^{*}(P_{\xi;k}K_{k}e_{p})\Big)\\
		&=\frac{1}{2}\sum_{k\in\Z_{*}^{3}}\Re\Big(\sum_{p\in L_{k}}b_{-k}^{*}(e_{-p})b_{k}^{*}(P_{\xi;k}K_{k}e_{p})+\sum_{p\in L_{k}}b_{k}^{*}(e_{p})b_{-k}^{*}(I_{k}P_{-\xi;k}K_{k}e_{p})\Big)\\
		&=\frac{1}{2}\sum_{k\in\Z_{*}^{3}}\Re\Big(\sum_{p\in L_{k}}b_{-k}^{*}(e_{-p})b_{k}^{*}(P_{\xi;k}K_{k}e_{p})+\sum_{p\in L_{k}}b_{k}^{*}(K_{k}P_{-\xi;k}e_{p})b_{-k}^{*}(e_{-p})\Big),
	\end{align*}
	where we have used the elementary identities
	\begin{align*}
		P_{-\xi;k}&=P_{\xi;-k},\quad I_{k}P_{\xi;k}=\chi_{L_{k}}(\xi)\ket{e_{-\xi}}\bra{e_{\xi}}=P_{-\xi;-k}I_{k}.
	\end{align*}
	Summing over $\xi$ and $-\xi$ yields
	\begin{align*}
		I_{\xi}+I_{-\xi}&=\frac{1}{2}\sum_{k\in\Z_{*}^{3}}\Re\Big(\sum_{p\in L_{k}}b_{-k}^{*}(e_{-p})b_{k}^{*}(\{P_{\xi;k}+P_{-\xi;k},K_{k}\}e_{p})\Big)\\
		&=\frac{1}{2}\sum_{k\in\Z_{*}^{3}}Q_{2}^{k}(\{P_{\xi;k}+P_{-\xi;k},K_{k}\}).
	\end{align*}
	The same argument above yields for term $II_{\xi}$ that
	\DETAILS{Similarly, for the term $II_{q}$, we have
		\begin{align*}
			II_{q}&=\Re\Big(\sum_{k\in\Z_{*}^{3}}\sum_{p\in L_{k}}b_{k}^{*}(K_{k}P_{k-\xi;k}e_{p})b_{-k}^{*}(e_{-p})\Big)\\
			&=\frac{1}{2}\sum_{k\in\Z_{*}^{3}}\Re\Big(\sum_{p\in L_{k}}b_{k}^{*}(K_{k}P_{k-\xi;k}e_{p})b_{-k}^{*}(e_{-p})+\sum_{p\in L_{k}}b_{-k}^{*}(e_{-p})b_{k}^{*}(P_{k+\xi;k}K_{k}e_{p})\Big).
		\end{align*}
		Again, summing over $\xi$ and $-\xi$ yields}
	\begin{align*}
		II_{\xi}+II_{-\xi}\DETAILS{&=\frac{1}{2}\sum_{k\in\Z_{*}^{3}}\Re\Big(\sum_{p\in L_{k}}b_{-k}^{*}(e_{-p})b_{k}^{*}(\{P_{k-\xi;k}+P_{k+\xi;k},K_{k}\}e_{p})\Big)\\}
		&=\frac{1}{2}\sum_{k\in\Z_{*}^{3}}Q_{2}^{k}(\{P_{k-\xi;k}+P_{k+\xi;k},K_{k}\}),
	\end{align*}
	thus proving \eqref{1st-com}.
	
	To finish the proof, we denote $P_{\xi;k}':=P_{\xi;k}+P_{-\xi;k}$ and $P_{\xi;k}'':=P_{k-\xi;k}+P_{k+\xi;k}$ and then compute
	\begin{align*}
		I_{k}P_{\xi;k}'&=I_{k}(P_{\xi;k}+P_{-\xi;k})=(P_{-\xi;-k}+P_{\xi,-k})I_{k}=P_{\xi,-k}'I_{k},
	\end{align*}
	and
	\begin{align*}
		I_{k}P_{\xi;k}''&=I_{k}(P_{k-\xi;k}+P_{k+\xi;k})=(P_{-k+\xi,-k}+P_{-k-\xi,-k})I_{k}=P_{\xi,-k}''I_{k},
	\end{align*}
	which implies the relation \eqref{unitary-tileKqk}.  This completes the proof.
\end{proof}
\begin{remark}\label{rem:symmetric-reason}
	If our trial state is not invariant under reflection in momentum space, then we cannot average over $\pm \xi$ in \eqref{ref-aveg}, and the above computation would output 
	\begin{align*}
		\begin{cases}
			Q_{2}^{k}(P_{\xi;k}K_{k}+K_{k}P_{-\xi;k})\quad&\xi\in B_{F}^{c},\\
			Q_{2}^{k}(P_{k+\xi;k}K_{k}+K_{k}P_{k-\xi;k})\quad&\xi \in B_{F},
		\end{cases}
	\end{align*}
	for each $k\in\Z_{*}^{3}$.  The operators $P_{\xi;k}K_{k}+K_{k}P_{-\xi;k}$ and $P_{k+\xi;k}K_{k}+K_{k}P_{k-\xi;k}$ are not symmetric and thus prevent us to apply techniques from \cite{CHN-23}.
\end{remark}

Recall that $\Phi_{t}=e^{-t\cK}\Psi_{N}$ for $0\leq t\leq 1$ so that $\Psi_{FS}=\Phi_{0}$.  By the fundamental theorem of calculus and by using the fact that $\langle\Psi_{FS},\rho_{\xi}\Psi_{FS}\rangle=0$ for any $\xi\in\Z^{3}$, we obtain
\begin{align}\label{FTC}
	n_{t}(\xi)=\langle \Phi_{t},\rho_{\xi}\Phi_{t}\rangle&=\langle\Psi_{FS},\rho_{\xi}\Psi_{FS}\rangle+\int_{0}^{t}\langle \Phi_{\tau},[\cK,\rho_{\xi}]\Phi_{\tau}\rangle d\tau\nonumber\\
	&=\frac{1}{2}\sum_{k\in\Z_{*}^{3}}\int_{0}^{t}\langle \Phi_{\tau},Q_{2}^{k}(\Theta_{\xi;k})\Phi_{\tau}\rangle d\tau.
\end{align}
Using Lemma \ref{lem:exp-Bog-Q12} and the fact that $\langle\Psi_{FS}, Q_{1}^{k}(T)\Psi_{FS}\rangle=\langle\Psi_{FS}, Q_{2}^{k}(T)\Psi_{FS}\rangle=0$ for any $k\in\Z_{*}^{3}$ and operators $T$ on $\ell^{2}(L_{k})$, we obtain from \eqref{FTC} that
\begin{align}\label{trial-momentum}
	n_{t}(\xi)&=n_{{\rm b},t}(\xi)+n_{{\rm ex},t}(\xi)+\cE_{1,t}(\xi)+\cE_{2,t}(\xi)+\cE_{3,t}(\xi),
\end{align}
where 
\begin{align*}
	n_{{\rm b},t}(\xi)&:=\frac{1}{2}\sum_{k\in\Z_{*}^{3}}\int_{0}^{t}\Tr\big(\Theta_{\xi;k}^{2}(\tau)\big)d\tau,\\
	n_{{\rm ex},t}(\xi)&:=\sum_{k\in\Z_{*}^{3}}\Re\int_{0}^{t}\int_{0}^{\tau}\big\langle\Psi_{FS},\cE_{2,k}(\Theta_{\xi;k}^{1}(\tau_{1}))\Psi_{FS}\big\rangle d\tau_{1} d\tau,\\
	\cE_{1,t}(\xi)&:=\sum_{k\in\Z_{*}^{3}}\Re\int_{0}^{t}\int_{0}^{\tau}\Big\langle e^{-(\tau-\tau_{1})\cK}\Psi_{FS},\widetilde{\cE}_{2,k}(\Theta_{\xi;k}^{1}(\tau_{1}))e^{-(\tau-\tau_{1})\cK}\Psi_{FS}\big\rangle d\tau_{1}d\tau,\\
	\cE_{2,t}(\xi)&:=\sum_{k\in\Z_{*}^{3}}\Re\int_{0}^{t}\int_{0}^{\tau}\Big\langle e^{-(\tau-\tau_{1})\cK}\Psi_{FS},\cE_{1,k}(\Theta_{\xi;k}^{2}(\tau))e^{-(\tau-\tau_{1})\cK}\Psi_{FS}\big\rangle d\tau_{1} d\tau,\\
	\cE_{3,t}(\xi)&:=\frac{1}{2}\sum_{k\in\Z_{*}^{3}}\int_{0}^{1}\int_{0}^{t}\Big\langle e^{-(\tau-\tau_{1})\cK}\Psi_{FS},\varepsilon_{k}(\{K_{k},\Theta_{\xi;k}^{1}(\tau)\})e^{-(\tau-\tau_{1})\cK}\Psi_{FS}\big\rangle d\tau_{1} d\tau.
\end{align*}
with $\widetilde{\cE}_{2,k}(T_{k}):=\cE_{2,k}(T_{k})-\langle\Psi_{FS},\cE_{2,k}(T_{k})\Psi_{FS}\rangle$.  

\subsection{Momentum distribution from bosonization}\label{sec:boson-contrib}
This subsection is devoted to prove the following proposition (c.f. \cite[Theorem 1.1]{BeneLill-25}):
\begin{prop}\label{prop:bosonized-momentum}
	It holds that
	\begin{align}\label{bosonized-momentum}
		n_{{\rm b},t}(\xi)=\frac{1}{2}\sum_{k\in\Z_{*}^{3}}\sum_{\zeta\in\cD_{k,\xi}}\langle e_{\zeta},(\cosh(-2tK_{k})-1)e_{\zeta}\rangle.
	\end{align} 
	Moreover, $0\leq n_{{\rm b},t}(\xi)\leq C_{V}k_{F}^{-1}m(\xi)$ uniformly in $0\leq t\leq 1$ and $\xi\in\Z^{3}$.  
\end{prop}
\begin{proof}
	Using the cyclicity of trace, we have
	\begin{align*}
		\Tr\big(\Theta_{\xi;k}^{2}(t)\big)&=\sum_{\zeta\in\cD_{k,\xi}}\Tr\big(S_{k}(t)(-K_{k})P_{\zeta;k}\big)=\sum_{\zeta\in\cD_{k,\xi}}\chi_{L_{k}}(\zeta)\langle e_{\zeta},(-K_{k})S_{k}(t)e_{\zeta}\rangle,
	\end{align*}
	and, by integrating over $t$ on $[0,1]$ and the symmetry property of $K_{k}$ under $k\mapsto -k$, we obtain
	\begin{align*}
		\sum_{k\in\Z_{*}^{3}}\int_{0}^{1}\Tr(\Theta_{\xi;k}^{2}(t))dt&=\sum_{k\in\Z_{*}^{3}}\sum_{\zeta\in\cD_{k,\xi}}\langle e_{\zeta},C_{k}(2t)e_{\zeta}\rangle=:n_{{\rm b},t}(\xi),
	\end{align*}
	where we recall $S_{k}(t)$ and $C_{k}(t)$ from \eqref{CSkt}.  Note that $C_{k}(t)$ is an increasing, positive operator-valued function on $0\leq t\leq 1$ so that, by spectral theory, \cite[Corollary 3.3]{CHN-23} and the lower bound \eqref{gap},
	\begin{align*}
		0\leq n_{{\rm b},t}(\xi)&\leq \sum_{k\in\Z_{*}^{3}}\sum_{\zeta\in\cD_{k,\xi}}\langle e_{\zeta},\big(\cosh(-2K_{k})-1\big)e_{\zeta}\rangle\nonumber\\
		&\leq C k_{F}^{-1}\sum_{k\in\Z_{*}^{3}}\hat{V}_{k}^{2}\Big(\frac{\chi_{L_{k}}(\xi)}{\lambda_{k,\xi}}+\frac{\chi_{L_{k}'}(\xi)}{\lambda_{k,k+\xi}}\Big)\leq C_{V}k_{F}^{-1}m(\xi).
	\end{align*}
	This completes the proof.
\end{proof}

%
%
%

\subsection{Momentum distribution from exchange correlation}\label{sec:ex-correct}
Now, we extract information from exchange correlation by proving the following proposition:
\begin{prop}\label{prop:ex-corr-momentum}
	For each $\xi\in\Z^{3}$ and $0\leq t\leq 1$, it holds that
	\begin{align}
		\Big|n_{{\rm ex},t}(\xi)-t^{2}n_{\rm ex}(\xi)\Big|\leq C_{V}(k_{F}^{-3/2}m(\xi)+k_{F}^{-2}m(\xi)),
	\end{align}
	where
	\begin{align}
		n_{{\rm ex}}(\xi)&=-\frac{k_{F}^{-2}}{8(2\pi)^{6}}\sum_{k\in\Z_{*}^{3}}\sum_{\zeta\in\cD_{k,\xi}}\sum_{p\in L_{k}}\frac{\hat{V}_{k}\hat{V}_{p+\zeta-k}}{(\lambda_{k,p}+\lambda_{k,\zeta})^{2}}.
	\end{align}
	Moreover, for each $\xi\in\Z^{3}$, $|n_{\rm ex}(\xi)|\leq C_{\delta,V}k_{F}^{-1+\delta}m(\xi)$ for any $\delta>0$.
\end{prop}
To prove Proposition \ref{prop:ex-corr-momentum}, we need to first establish some technical results.  Recall the definition, for any symmetric operator $T_{k}$ on $\ell^{2}(L_{k})$,
\begin{align*}
	\cE_{2,k}(T_{k})&=\frac{1}{2}\sum_{l\in\Z_{*}^{3}}\sum_{p\in L_{k}}\sum_{q\in L_{l}}\{b_{k}(T_{k}e_{p}),\{\varepsilon_{-k,-l}(e_{-p};e_{-q}),b_{l}^{*}(K_{l}e_{q})\}\}
\end{align*}
Since $b_{k}(T_{k})\Psi_{FS}=\varepsilon_{-k,-l}(e_{-p};e_{-q})\Psi_{FS}=0$ for each $k,l\in\Z_{*}^{3}$ and $p\in L_{k},q\in L_{l}$, for each $0\leq \tau\leq 1$, we have
\begin{align}\label{matrix-Ek2}
	-2\big\langle\Psi_{FS},&\cE_{2,k}(T_{k})\Psi_{FS}\big\rangle=-\sum_{l\in\Z_{*}^{3}}\sum_{p\in L_{k}}\sum_{q\in L_{l}}\big\langle\Psi_{FS}, b_{k}(T_{k}e_{p})\varepsilon_{-k,-l}(e_{-p};e_{-q})b_{l}^{*}(K_{l}e_{q})\Psi_{FS}\big\rangle\nonumber\\
	&=\sum_{l\in \Z_{*}^{3}}\sum_{p\in L_{k}\cap L_{l}}\big\langle\Psi_{FS}, b_{k}(T_{k}e_{p})\tilde{a}_{-p+l}^{*}\tilde{a}_{-p+k}b_{l}^{*}(K_{l}e_{p})\Psi_{FS}\big\rangle\\
	&\quad\quad+\sum_{l\in\Z_{*}^{3}}\sum_{p\in L_{k}'\cap L_{l}'}\big\langle\Psi_{FS}, b_{k}(T_{k}e_{p+k})\tilde{a}_{-p-l}^{*}\tilde{a}_{-p-k}b_{l}^{*}(K_{l}e_{p+l})\Psi_{FS}\big\rangle,\nonumber
\end{align}
where we recall $\tilde{a}_{p}^{*}$ from \eqref{particle-hole}.  It follows that $\tilde{a}_{p}\Psi_{FS}=0$ for each $p\in\Z^{3}$ under our assumption for $\Psi_{FS}$.  Hence, in evaluating $\langle \Psi_{FS},A\Psi_{FS}\rangle$, it suffices to rewrite $A$ in terms of $\tilde{a}_{p}^{\#}$'s, put it into normal order and drop all the terms with $\tilde{a}^{*}$ on the most left and $\tilde{a}$ on the most right.  

For the first term in \eqref{matrix-Ek2}, since $p\in L_{l}$, we obtain
\begin{align}\label{matrix-Ek2-1}
	&\sum_{l\in\Z_{*}^{3}}\sum_{p\in L_{k}\cap L_{l}}\big\langle\Psi_{FS}, b_{k}(T_{k}e_{p})\tilde{a}_{-p+l}^{*}\tilde{a}_{-p+k}b_{l}^{*}(K_{l}e_{p})\Psi_{FS}\big\rangle\nonumber\\
	&=\sum_{l\in\Z_{*}^{3}}\sum_{p\in L_{k}\cap L_{l}}\big\langle\Psi_{FS},  \big[b_{k}(T_{k}e_{p}),\tilde{a}_{-p+l}^{*}\big]\big[b_{l}(K_{l}e_{p}),\tilde{a}_{-p+k}^{*}\big]^{*}\Psi_{FS}\big\rangle\nonumber\\
	&=\sum_{l\in\Z_{*}^{3}}\sum_{p\in L_{k}\cap L_{l}}\sum_{q\in L_{k}}\sum_{q'\in L_{l}}\delta_{-p+l,q-k}\delta_{-p+k,q'-l}\langle T_{k}e_{p},e_{q}\rangle\langle e_{q'},K_{l}e_{p}\rangle\big\langle\Psi_{FS}, \tilde{a}_{q}\tilde{a}_{q'}^{*}\Psi_{FS}\big\rangle\nonumber\\
	&=\sum_{l\in\Z_{*}^{3}}\sum_{p,q\in L_{k}\cap L_{l}}\delta_{k+l,p+q}\langle T_{k}e_{p},e_{q}\rangle\langle e_{q},K_{l}e_{p}\rangle,
\end{align}
where we use the following identity (which is easily followed from CAR; see also \cite[Eq. (4.22)]{CHN-23})
\begin{align}\label{ba-comm}
	[b_{l}(\psi),\tilde{a}_{p}^{*}]&=\begin{cases}
		-\chi_{L_{l}'}(p)\langle\psi,e_{p+l}\rangle\tilde{a}_{p+l},\quad&p\in B_{F},\\
		\chi_{L_{l}}(p)\langle \psi,e_{p}\rangle\tilde{a}_{p-l},\quad&p\in B_{F}^{c}.
	\end{cases}
\end{align}
Similarly, the second term in \eqref{matrix-Ek2-1} gives 
\begin{align}\label{matrix-Ek2-2}
	&\sum_{l\in\Z_{*}^{3}}\sum_{p\in L_{k}'\cap L_{l}'}\Big\langle b_{k}(T_{k}e_{p+k})\tilde{a}_{-p-l}^{*}\tilde{a}_{-p-k}b_{l}^{*}(K_{l}e_{p+l})\Psi_{FS}\big\rangle\nonumber\\
	&=\sum_{l\in\Z_{*}^{3}}\sum_{p,q\in L_{k}'\cap L_{l}'}\delta_{k+l,-(p+q)}\langle T_{k}e_{p+k},e_{q+k}\rangle\langle e_{q+k},K_{l}e_{p+k}\rangle.
\end{align}
By the same argument as in \cite[Subsection 4.3]{CHN-23}, it turns out these two terms are equal so that
\begin{align}\label{matrix-Ek2-3}
	\big\langle\cE_{2,k}(T_{k})\Psi_{FS}\big\rangle&=\sum_{l\in\Z_{*}^{3}}\sum_{p,q\in L_{k}\cap L_{l}}\delta_{k+l,p+q}\langle T_{k}e_{p},e_{q}\rangle\langle e_{q},(-K_{l})e_{p}\rangle.
\end{align}

Next, we note that, since $\cE_{2,k}(T_{k})$ depends linearly on $T_{k}$ and
\begin{align*}
	\Theta_{\xi;k}^{1}(s)&=\frac{1}{2}\sum_{k\in\Z_{*}^{3}}\sum_{\cD_{k,\xi}}\frac{d}{ds}P_{\zeta;k}^{2}(s).
\end{align*}
we can rewrite
\begin{align}\label{momentum2}
	n_{{\rm ex},t}(\xi)&=\sum_{(k,\zeta)\in\cC_{\xi}}\big\langle\Psi_{FS},\cE_{2,k}(G_{\zeta;k}(t))\Psi_{FS}\big\rangle.
\end{align}
where 
\begin{align}\label{Gzeta}
	G_{\zeta;k}(t)&:=\frac{1}{2}\int_{0}^{t}P_{\zeta;k}^{2}(\tau)d\tau.
\end{align}
and we have dropped the ``Re" symbol in $n_{{\rm ex},t}(\xi)$ since the quantity in \eqref{momentum2} is real for each $t$.  To extract contribution from exchange correction due to the second term in \eqref{trial-momentum}, we need to first approximate the integral $G_{\zeta;k}(t)$:
\begin{lemma}\label{lem:approx-Gzeta}
	For each $k\in\Z_{*}^{3}$, $p,q\in L_{k}$, $\zeta\in\cD_{k,\xi}$ and $0\leq t\leq 1$, it holds that 
	\begin{align}\label{approx-Gzeta}
		\Big|\langle e_{p},\Big(G_{\zeta;k}(t)&+\frac{t^{2}}{2}\frac{\{P_{\zeta;k},P_{v_{k}}\}}{\lambda_{k,p}+\lambda_{k,q})}\Big)e_{q}\rangle\Big|\nonumber\\
		&\leq Ck_{F}^{-1}\hat{V}_{k}^{2} \chi_{L_{k}}(\zeta)\Big(\frac{\delta_{p,\zeta}+\delta_{q,\zeta}}{\lambda_{k,p}+\lambda_{k,q}}+\frac{k_{F}^{-1}\hat{V}_{k}}{(\lambda_{k,p}+\lambda_{k,\zeta})(\lambda_{k,\zeta}+\lambda_{k,q})}\Big).
	\end{align}
\end{lemma}
\begin{proof}
	First, we rewrite, for each fixed $p,q\in L_{k}$ and $0\leq \tau\leq t$,
	\begin{align}\label{Pzetak2}
		P_{\zeta;k}^{2}(\tau)=-\tau\frac{\{P_{\zeta;k},P_{v_{k}}\}}{\lambda_{k,p}+\lambda_{k,q}}+F_{k,\zeta}(\tau)
	\end{align}
	where
	\begin{align*}
		F_{k,\zeta}(\tau)&=\tau\frac{\{P_{\zeta;k},P_{v_{k}}\}}{\lambda_{k,p}+\lambda_{k,q}}-\{P_{\zeta;k},S_{k}(\tau)\}-C_{k}(\tau)P_{\zeta;k}S_{k}(\tau)-S_{k}(\tau)P_{\zeta;k}C_{k}(\tau).
	\end{align*}
	
	Next, we estimate the matrix element of $F_{k,\zeta}(\tau)$.  By Lemma \ref{lem:matrix-expK}, it is straightforward to obtain
	\begin{align*}
		|\langle e_{p},S_{k}(\tau)P_{\zeta;k}C_{k}(\tau)e_{q}\rangle|\leq C\frac{k_{F}^{-2}\hat{V}_{k}^{3}}{(\lambda_{k,p}+\lambda_{k,\zeta})(\lambda_{k,\zeta}+\lambda_{k,q})}.
	\end{align*}
	We can obtain the same bound for $\langle e_{p},C_{k}(\tau)P_{\zeta;k}S_{k}(\tau)e_{q}\rangle$.  For remaining term, we again use Lemma \ref{lem:matrix-expK} and the lower bound \eqref{gap} to obtain
	\begin{align*}
		&\Big|\langle e_{p},\Big(\tau\frac{\{P_{\zeta;k},P_{v_{k}}\} }{\lambda_{k,p}+\lambda_{k,q}}-\{P_{\zeta;k},S_{k}(\tau)\}\Big)e_{q}\rangle\Big|\nonumber\\
		&\leq\frac{\tau}{\lambda_{k,p}+\lambda_{k,q}}\Big(1-\frac{1}{1+2\langle v_{k},h_{k}^{-1}v_{k}\rangle}\Big)\big(\delta_{p,\zeta}+\delta_{q,\zeta}\big)\langle e_{p},v_{k}\rangle\langle v_{k},e_{q}\rangle\nonumber\\
		&\leq \frac{2\langle v_{k},h_{k}^{-1}v_{k}\rangle}{1+2\langle v_{k},h_{k}^{-1}v_{k}\rangle}\big(\delta_{p,\zeta}+\delta_{q,\zeta}\big)k_{F}^{-1}\hat{V}_{k}\leq C\big(\delta_{p,\zeta}+\delta_{q,\zeta}\big)\frac{k_{F}^{-1}\hat{V}_{k}^{2}}{\lambda_{k,p}+\lambda_{k,q}}.
	\end{align*}
	Then integrating $\tau$ over $[0,t]$ yields \eqref{approx-Gzeta}.
\end{proof}

\begin{proof}[Proof of Proposition \ref{prop:ex-corr-momentum}]
	First, we note that, when $k+l=p+q$, we have the relation
	\begin{align}\label{symmetry-k+l=p+q}
		\lambda_{k,p}+\lambda_{k,q}&=\frac{|p|^{2}-|p-k|^{2}}{2}+\frac{|q|^{2}-|q-k|^{2}}{2}\nonumber\\
		&=\frac{|p|^{2}-|p-l|^{2}}{2}+\frac{|q|^{2}-|q-l|^{2}}{2}=\lambda_{l,p}+\lambda_{l,q},
	\end{align}
	so that, together with symmetry of summation in $n_{\rm ex}(\xi)$ under the transformation $k\rightarrow -k$, we can write
	\begin{align}\label{ex-distribution}
		n_{{\rm ex}}(\xi)&=-\frac{k_{F}^{-2}}{8(2\pi)^{6}}\sum_{k,l\in\Z_{*}^{3}}\sum_{p,q\in L_{k}\cap L_{l}}\sum_{\zeta\in\cD_{k,\xi}}\delta_{k+l,p+q}\chi_{L_{k}}(\zeta)\frac{(\delta_{p,\zeta}+\delta_{q,\zeta})\hat{V}_{k}\hat{V}_{l}}{(\lambda_{k,p}+\lambda_{k,q})(\lambda_{l,p}+\lambda_{l,q})}\nonumber\\
		&=-\frac{1}{2}\sum_{k,l\in\Z_{*}^{3}}\sum_{p,q\in L_{k}\cap L_{l}}\sum_{\zeta\in\cD_{k,\xi}}\delta_{k+l,p+q}\frac{k_{F}^{-1}\hat{V}_{l}}{2(2\pi)^{3}}\frac{\langle\{P_{\zeta;k},P_{v_{k}}\}e_{p},e_{q}\rangle}{(\lambda_{k,p}+\lambda_{k,q})(\lambda_{l,p}+\lambda_{l,q})}.
	\end{align}
	Hence, by \eqref{momentum2}, we can write
	\begin{align}\label{nex-t}
		n_{{\rm ex},t}(\xi)\DETAILS{&=\sum_{k,l\in \Z_{*}^{3}}\sum_{p,q\in L_{k}\cap L_{l}}\sum_{\zeta\in\cD_{k,\xi}}\delta_{k+l,p+q}\chi_{L_{k}}(\zeta)\langle G_{\zeta;k}(t)e_{p},e_{q}\rangle\langle e_{q},(-K_{l})e_{p}\rangle\nonumber\\}
		&=t^{2}n_{\rm ex}(\xi)+\sum_{k,l\in\Z_{*}^{3}}\sum_{p,q\in L_{k}\cap L_{l}}\sum_{\zeta\in\cD_{k,\xi}}\delta_{k+l,p+q}\chi_{L_{k}}(\zeta)\langle e_{q},(-K_{l})e_{p}\rangle\nonumber\\
		&\quad\quad\quad\quad\quad\quad\quad\quad\quad\quad\quad\quad\quad\times\Big(\langle G_{\zeta;k}(t)+\frac{t^{2}\{P_{\zeta;k},P_{v_{k}}\}}{2(\lambda_{k,p}+\lambda_{k,q})}e_{p},e_{q}\rangle\Big)\nonumber\\
		&\quad-t^{2}\sum_{k,l\in\Z_{*}^{3}}\sum_{p,q\in L_{k}\cap L_{l}}\sum_{\zeta\in\cD_{k,\xi}}\delta_{k+l,p+q}\chi_{L_{k}}(\zeta)\frac{\langle e_{p},\{P_{\zeta;k},P_{v_{k}}\}e_{q}\rangle}{2(\lambda_{k,p}+\lambda_{k,q})}\nonumber\\
		&\quad\quad\quad\quad\quad\quad\quad\quad\quad\quad\quad\quad\times\Big(\langle e_{q},(-K_{l})e_{p}\rangle-\frac{\langle e_{q},v_{l}\rangle\langle v_{l},e_{p}\rangle}{\lambda_{l,p}+\lambda_{l,q}}\Big).
	\end{align}
	We denote the last two sums in \eqref{nex-t} by $A(t)$ and $t^{2}B$, respectively.

	Next, we estimate the term $A(t)$ and $B(t)$.  By Lemmas \ref{lem:matrix-expK}, \ref{lem:approx-Gzeta} and relation \ref{symmetry-k+l=p+q}, we obtain
	\begin{align}\label{Fkt-error}
		|A(t)|\DETAILS{&\leq Ck_{F}^{-1}\sum_{k\in\Z_{*}^{3}}\sum_{\zeta\in\cD_{k,\xi}}\sum_{p,q\in L_{k}}\delta_{k+l,p+q}\hat{V}_{k}^{2}\frac{\langle e_{p},v_{l}\rangle\langle v_{l},e_{q}\rangle}{\lambda_{l,p}+\lambda_{l,q}}\nonumber\\
		&\quad\quad\quad\quad\quad\quad\quad\quad\quad\quad\quad\times\Big(\frac{\delta_{p,\zeta}+\delta_{q,\zeta}}{\lambda_{l,p}+\lambda_{l,q}}+\frac{k_{F}^{-1}\hat{V}_{k}}{(\lambda_{k,p}+\lambda_{k,\zeta})(\lambda_{k,\zeta}+\lambda_{k,q})}\Big)\nonumber\\}
		&\leq Ck_{F}^{-2}\sum_{k\in\Z_{*}^{3}}\sum_{\zeta\in\cD_{k,\xi}}\hat{V}_{k}^{2}\sum_{p,q\in L_{k}}\Big(\delta_{p,\zeta}+\delta_{q,\zeta}+\frac{k_{F}^{-1}\hat{V}_{k}}{\lambda_{k,\zeta}}\Big)\frac{\hat{V}_{p+q-k}}{(\lambda_{k,p}+\lambda_{k,q})^{2}}\nonumber\\
		&\leq Ck_{F}^{-2}\sum_{k\in\Z_{*}^{3}}\sum_{\zeta\in\cD_{k,\xi}}\hat{V}_{k}^{2}\sum_{p\in L_{k}}\frac{\hat{V}_{p+\zeta-k}}{(\lambda_{k,p}+\lambda_{k,\zeta})^{2}}\nonumber\\
		&\quad\quad+Ck_{F}^{-3}\sum_{k\in\Z_{*}^{3}}\sum_{\zeta\in\cD_{k,\zeta}}\hat{V}_{k}^{3}\frac{\chi_{L_{k}}(\zeta)}{\lambda_{k,\zeta}}\sum_{p,q\in L_{k}}\frac{\hat{V}_{p+q-k}}{(\lambda_{k,p}+\lambda_{k,q})^{2}}=:A_{1}+A_{2}.
	\end{align}
	Since $\lambda_{k,\zeta}\geq m(\xi)^{-1}$ for any $\zeta\in\cD_{k,\xi}$, the terms $A_{1}$ and $A_{2}$ can be estimated by Lemmas \ref{lem:estimate-lambkp}, \ref{lem:gap-sum}, Cauchy-Schwarz inequality and the lower bound \eqref{gap} to be
	\begin{align*}
		A_{1}&\leq C k_{F}^{-2}\sum_{k\in\Z_{*}^{3}}\sum_{\zeta\in\cD_{k,\xi}}\hat{V}_{k}^{2}\sum_{p\in L_{k}}\frac{\hat{V}_{p+\zeta-k}}{\lambda_{k,p}\lambda_{k,\zeta}}\nonumber\\
		&\leq Ck_{F}^{-2}m(\xi)\sum_{k\in\Z_{*}^{3}}\hat{V}_{k}^{2}\sqrt{\sum_{l\in\Z_{*}^{3}}\hat{V}_{l}^{2}}\sqrt{\sum_{p\in L_{k}}\frac{1}{\lambda_{k,p}}}\leq C_{V}k_{F}^{-3/2}m(\xi),
	\end{align*}
	and, similarly,
	\begin{align*}
		A_{2}&\leq C k_{F}^{-3}\sum_{k\in\Z_{*}^{3}}\sum_{\zeta\in\cD_{k,\xi}}\hat{V}_{k}^{3}\frac{\chi_{L_{k}}(\zeta)}{\lambda_{k,\zeta}}\sum_{p,q\in L_{k}}\frac{\hat{V}_{p+q-k}}{(\lambda_{k,p}+\lambda_{k,q})^{2}}\nonumber\\
		&\leq C k_{F}^{-3}m(\xi)\|\hat{V}\|_{\infty}\sum_{k\in\Z_{*}^{3}}\hat{V}_{k}^{2}\Big(\sum_{p\in L_{k}}\frac{1}{\lambda_{k,p}}\Big)\sqrt{\sum_{l\in\Z_{*}^{3}}\hat{V}_{l}^{2}}\sqrt{\sum_{q\in L_{k}}\frac{1}{\lambda_{k,q}}}\leq C_{V}k_{F}^{-3/2}m(\xi),
	\end{align*}
	where we have used the inequality $\|\hat{V}\|_{\infty}\leq \|\hat{V}\|_{\ell^{2}}$.  For the term $B$, by a similar argument, we obtain
	\begin{align*}
		|B|&\leq Ck_{F}^{-2}\sum_{k\in\Z_{*}^{3}}\sum_{p\in L_{k}}\sum_{\zeta\in\cD_{k,\xi}}\chi_{L_{k}}(\zeta)\frac{\hat{V}_{k}^{2}\hat{V}_{\zeta+p-k}}{(\lambda_{k,p}+\lambda_{k,\zeta})^{2}}\nonumber\\
		&\leq Ck_{F}^{-2}\sum_{k\in\Z_{*}^{3}}\frac{\hat{V}_{k}^{2}}{\lambda_{k,\zeta}}\sqrt{\sum_{l\in\Z_{*}^{3}}\hat{V}_{l}}\sqrt{\sum_{p\in L_{k}}\frac{1}{\lambda_{k,p}}}\leq C_{V}k_{F}^{-3/2}m(\xi).
	\end{align*}
	
	Now, we estimate $n_{\rm ex}(\xi)$.  To this end, we consider $\xi\in B_{F}^{c}$ and $\xi\in B_{F}$ separately.  For $\xi\in B_{F}^{c}$, by \eqref{ex-distribution}, Lemmas \ref{lem:matrix-expK}, \ref{lem:gap-sum}, \ref{lem:approx-Gzeta}, and relation \ref{symmetry-k+l=p+q}, we obtain for each $\delta>0$ that
	\begin{align}\label{ex-esti}
		|n_{\rm ex}(\xi)|&\leq C k_{F}^{-2}\sum_{k\in\Z_{*}^{3}}\sum_{p\in L_{k}}\chi_{L_{k}}(\xi)\frac{\hat{V}_{k}\hat{V}_{p+\xi-k}}{\lambda_{k,\xi}(\lambda_{k,\xi}+\lambda_{k,p})}\nonumber\\
		&\leq Ck_{F}^{-2}m(\xi)\sum_{k\in\Z_{*}^{3}}\frac{\hat{V}_{k}\chi_{L_{k}}(\xi)}{\sqrt{\lambda_{k,\xi}}}\sqrt{\sum_{l\in\Z_{*}^{3}}\hat{V}_{l}^{2}}\sqrt{\sum_{p\in L_{k}}\frac{1}{\lambda_{k,p}}}\nonumber\\
		&\leq C_{V}k_{F}^{-3/2}m(\xi)\sqrt{\sum_{k\in\Z_{*}^{3}}\hat{V}_{k}^{2}}\sqrt{\sum_{k\in\Z_{*}^{3}}\frac{\chi_{L_{k}}(\xi)}{\lambda_{k,\xi}}}\leq C_{\delta, V}k_{F}^{-1+\delta}m(\xi).
	\end{align}
	For $\xi\in B_{F}$, by a similar argument, we obtain for each $\delta>0$ that
	\begin{align}\label{ex-esti-1}
		|n_{\rm ex}(\xi)|&\leq C k_{F}^{-2}\sum_{k\in\Z_{*}^{3}}\sum_{p\in L_{k}}\chi_{L_{k}'}(\xi)\frac{\hat{V}_{k}\hat{V}_{p+\xi}}{\lambda_{k,k+\xi}(\lambda_{k,k+\xi}+\lambda_{k,p})}\nonumber\\
		&\leq Ck_{F}^{-2}m(\xi)^{1/2}\sum_{k\in\Z_{*}^{3}}\frac{\hat{V}_{k}\chi_{L_{k}}(\xi)}{\lambda_{k,k+\xi}}\sqrt{\sum_{l\in\Z_{*}^{3}}\hat{V}_{l}^{2}}\sqrt{\sum_{p\in L_{k}}\frac{1}{\lambda_{k,p}}}\nonumber\\
		&\leq C_{V}k_{F}^{-3/2}m(\xi)^{1/2}\sqrt{\sum_{k\in\Z_{*}^{3}}\hat{V}_{k}^{2}}\sqrt{\sum_{k\in\Z_{*}^{3}}\frac{\chi_{L_{k}'}(\xi)}{\lambda_{k,k+\xi}^{2}}}\leq C_{\delta,V}k_{F}^{-1+\delta}m(\xi).
	\end{align}
	
	This completes the proof.
\end{proof}

\section{Generic structure of error terms}\label{sec:class-errors}
First, we rewrite the error terms $\cE_{j}(t)$, $j=1,2,3$, into generic forms.  By decomposing $P_{\zeta;k}^{1}(\tau)$ and $P_{\zeta;k}^{2}(\tau)$ as in \eqref{T1-decomp}--\eqref{T2-decomp} and then by expanding anticommutators, we have
\begin{align*}
	\{P_{\zeta;k}^{1}(\tau),K_{k}\}\DETAILS{&=\{P_{\zeta;k}+\{P_{\zeta;k},C_{k}(\tau)\}+C_{k}(\tau)P_{\zeta;k}C_{k}(\tau)+S_{k}(\tau)P_{\zeta;k}S_{k}(\tau),K_{k}\}\nonumber\\}
	&=\{P_{\zeta;k},K_{k}\}+\Lambda_{\zeta;k}^{1}(\tau)+\Gamma_{\zeta;k}^{1}(\tau),\quad\{P_{\zeta;k}^{2}(\tau),K_{k}\}\DETAILS{&=-\{\{P_{\zeta;k},S_{k}(\tau)\}+S_{k}(\tau)P_{\zeta;k}C_{k}(\tau)+C_{k}(\tau)P_{\zeta;k}S_{k}(\tau),K_{k}\}\nonumber\\}
	=-\Lambda_{\zeta;k}^{2}(\tau)-\Gamma_{\zeta;k}^{2}(\tau),
\end{align*}
\begin{align*}
	\{\{P_{\zeta;k}^{1}(\tau),K_{k}\},K_{k}\}\DETAILS{&=\{\{P_{\zeta;k},K_{k}\}+K_{k}C_{k}(\tau)P_{\zeta;k}+P_{\zeta;k}C_{k}(\tau)K_{k},K_{k}\}+\{\Gamma_{\zeta;k}^{1}(\tau),K_{k}\}\nonumber\\}
	&=\Lambda_{\zeta;k}^{3}(\tau)+\Gamma_{\zeta;k}^{3}(\tau),
\end{align*}
where
\begin{align*}
	\Lambda_{\zeta;k}^{1}(\tau)&=K_{k}C_{k}(\tau)P_{\zeta;k}+P_{\zeta;k}C_{k}(\tau)K_{k},\\
	\Lambda_{\zeta;k}^{2}(\tau)&=P_{\zeta;k}S_{k}(\tau)K_{k}+K_{k}S_{k}(\tau)P_{\zeta;k},\\
	\Lambda_{\zeta;k}^{3}(\tau)&=P_{\zeta;k}(K_{k})^{2}+(K_{k})^{2}P_{\zeta;k}+(K_{k})^{2}C_{k}(\tau)P_{\zeta;k}+P_{\zeta;k}C_{k}(\tau)(K_{k})^{2},
\end{align*}
and
\begin{align*}
	\Gamma_{\zeta;k}^{1}(\tau)&=C_{k}(\tau)P_{\zeta;k}K_{k}+K_{k}P_{\zeta;k}C_{k}(\tau)+\{C_{k}(\tau)P_{\zeta;k}C_{k}(\tau)+S_{k}(\tau)P_{\zeta;k}S_{k}(\tau),K_{k}\},\\
	\Gamma_{\zeta;k}^{2}(\tau)&=S_{k}(\tau)P_{\zeta;k}K_{k}+K_{k}P_{\zeta;k}S_{k}(\tau)+\{C_{k}(\tau)P_{\zeta;k}S_{k}(\tau)+S_{k}(\tau)P_{\zeta;k}C_{k}(\tau),K_{k}\},\\
	\Gamma_{\zeta;k}^{3}(\tau)&=2K_{k}P_{\zeta;k}K_{k}+K_{k}C_{k}(\tau)P_{\zeta;k}K_{k}+K_{k}P_{\zeta;k}C_{k}(\tau)K_{k}+\{\Gamma_{\zeta;k}^{1}(\tau),K_{k}\}.
\end{align*}

We observe that $\Gamma_{\zeta;k}^{i}(\tau)$, $i=1,2,3$, are all finite linear combinations of operators in the form
\begin{align*}
	A_{k,\tau}^{(j_{1})}P_{\zeta;k}A_{k,\tau}^{(j_{2})},
\end{align*}
for $1\leq j_{1},j_{2}\leq 3$.  Similarly, $\Lambda_{\zeta;k}^{i}(\tau)$, $i=1,2,3$, are all finite linear combinations of operators in the form
\begin{align*}
	A_{k,\tau}^{(j_{1})}P_{\zeta;k}+P_{\zeta;k}A_{k,\tau}^{(j_{2})},
\end{align*}
for $2\leq j_{1},j_{2}\leq 3$.  Here we denote $A_{k}^{(m)}(\tau)$ denotes any generic $m$-fold products of operators from the set $\{C_{k}(\tau),S_{k}(\tau),K_{k}\}$.  Hence, it suffices to estimate terms in the following sets for each fixed $0\leq s_{1},s_{2}\leq 1$ and for each $(j_{1},j_{2})\in\Sigma_{*}$:
\begin{align}\label{error-class}
	\cW_{s_{1},s_{2}}^{(j_{1},j_{2})}(\cG_{k})&:=\Big\{\sum_{k\in\Z_{*}^{3}}\sum_{\zeta\in\cD_{k,\xi}}\big\langle \Phi_{s_{1}},\cG_{k}\big(A_{k,s_{2}}^{(j_{1})}P_{\zeta;k}A_{k,s_{2}}^{(j_{2})}\big)\Phi_{s_{1}}\big\rangle\Big\},
\end{align}
where
$$\Sigma_{*}:=\{(j_{1},j_{2})\in\Z^{2}\setminus\{0\}\mid 0\leq m,n\leq 3\}$$
and $\cG_{k}$ denotes either $\cE_{1,k},\widetilde{\cE}_{2,k}$ or $\varepsilon_{k}$.  We decompose the index set $\Sigma_{*}=\Sigma_{*}^{1}\cup\Sigma_{*}^{2}$ as 
\begin{align}
	\Sigma_{*}^{1}&:=\{(j_{1},j_{2})\in\Sigma_{*}\mid 1\leq m,n\leq 3\},\quad\Sigma_{*}^{2}\setminus\Sigma_{*}^{1}.
\end{align}
Moreover, to simplify notation, we introduce the following measure for each $\xi\in\Z^{3}$ and function $f$ depending on $k$ and $\zeta$:
\begin{align}\label{measure}
	\sum_{(k,\zeta)\in\cC_{\xi}}f(k,\zeta)&:=\sum_{k\in\Z_{*}^{3}}\sum_{\zeta\in\cD_{k,\xi}}\chi_{L_{k}}(\zeta)f(k,\zeta),
\end{align}
where $\cC_{\xi}:=\{(k,\zeta)\in\Z_{*}^{3}\times\Z_{*}^{3}\mid \zeta\in\cD_{k,\xi}\}$, and we observe from \eqref{gap} that
\begin{align}\label{gap2}
	\lambda_{k,\zeta}&\geq m(\xi)^{-1}\quad\text{ for each }(k,\zeta)\in\cC_{\xi}.
\end{align}

\bigskip

\section{Analysis on error terms}\label{sec:principal-errors}
In this section, we consider the error terms in \eqref{trial-momentum}.  Since the structures of operator-valued kernels $\widetilde{\cE}_{2,k},\cE_{1,k}$ and $\varepsilon_{k}$ are very different, we must analyze them separately.

\subsection{Preliminary estimates}\label{sec:preliminary-esti}
Before proceeding, we need some preliminary estimates that will be useful when we estimate our major errors.  First, we generalize \cite[Proposition 5.8]{CHN-23} to all powers of $(\cN_{E}+1)$ by the same argument in \cite{CHN-23} with Gr\"onwall's lemma (we omit its proof for simplicity).
\begin{lemma}\label{lem:Gronwall}
For any $\Psi\in\cH_{N}$, $m\in\mathbb{N}$ and $|\tau|\leq 1$, it holds that
\begin{align}\label{Gronwall}
	\big\langle e^{-\tau\cK}\Psi,(\cN_{E}+1)^{m}e^{-\tau\cK}\Psi\rangle&\leq C_{m,V}\langle\Psi,(\cN_{E}+1)^{m}\Psi\rangle,
\end{align}
for some constant $C_{m}>0$.
\end{lemma}

\begin{lemma}\label{lem:excitation-Gronwall}
For any $m\geq 1$, $|\tau|\leq 1$ and $\delta>0$, it holds that
\begin{align}\label{excitation-Gronwall2}
	\|\tilde{a}_{p}(\cN_{E}+1)^{m}\Phi_{\tau}\|&\leq C_{\delta,m,V}\|\tilde{a}_{p}\Phi_{\tau}\|^{1-\delta}.
\end{align}
\end{lemma}
\begin{proof}
Since $[\tilde{a}_{p}^{*}\tilde{a}_{p},\cN_{E}]=0$, by Cauchy-Schwarz inequality, we have
\begin{align*}
	\|\tilde{a}_{p}(\cN_{E}+1)^{m}e^{-\tau\cK}\Psi\|^{2}&=\langle e^{-\tau\cK}\Psi,(\cN_{E}+1)^{m}\tilde{a}_{p}^{*}\tilde{a}_{p}(\cN_{E}+1)^{m}e^{-\tau\cK}\Psi\rangle\nonumber\\
	&=\langle \tilde{a}_{p}e^{-\tau\cK}\Psi,\tilde{a}_{p}(\cN_{E}+1)^{2m}e^{-\tau\cK}\Psi\rangle\nonumber\\
	&\leq \|\tilde{a}_{p}e^{-\tau\cK}\Psi\|\|\tilde{a}_{p}(\cN_{E}+1)^{2m}e^{-\tau\cK}\Psi\|.
\end{align*}
Together with operator norm $\|\tilde{a}_{p}\|=1$, by iterating this estimate $n$ times, we obtain
\begin{align}
	\|\tilde{a}_{p}(\cN_{E}+1)^{m}\Phi_{\tau}\|^{2}&\leq \|\tilde{a}_{p}\Phi_{\tau}\|^{2(1-2^{-n})}\|\tilde{a}_{p}(\cN_{E}+1)^{2^{n}m}\Phi_{\tau}\|^{2^{-n+1}}\nonumber\\
	&\leq C_{m,n,V}\|\tilde{a}_{p}\Phi_{\tau}\|^{2(1-2^{-n})}\|(\cN_{E}+1)^{2^{n}m}\Psi_{FS}\|^{2^{-n+1}}.
\end{align}
Since $n$ is arbitrary, together with Lemma \ref{lem:Gronwall}, this yields \eqref{excitation-Gronwall2}.
\end{proof}
To estimate error terms in \eqref{error-class}, we introduce the following quantity:
\begin{align}\label{bootstrap}
	\cQ&:=\sup_{\xi\in\Z^{3}}\sup_{0\leq \tau\leq 1}\|\tilde{a}_{\xi}\Phi_{\tau}\|.
\end{align}
We will show below each of the error terms are bounded by products of $k_{F}^{-1}$ and $\cQ$.  

\subsection{Analysis on errors in $\cW_{s_{1},s_{2}}^{(j_{1},j_{2})}(\varepsilon_{k})$}\label{sec:esti-E3134}
Recall the definition for any symmetric operator $T_{k}$ on $\ell^{2}(L_{k})$ satisfying $\langle e_{r},T_{k}e_{s}\rangle=\langle e_{-r},T_{-k}e_{-s}\rangle$ for each $r,s\in L_{k}$ that
\begin{align*}
	-\varepsilon_{k}(T_{k})&=\sum_{p\in L_{k}}\langle e_{p},T_{k}e_{p}\rangle\big(\tilde{a}_{p}^{*}\tilde{a}_{p}+\tilde{a}_{p-k}^{*}\tilde{a}_{p-k}\big),
\end{align*}
since the variable $k$ in the summand always belongs to $B_{F}^{c}$, whereas $p-k$ belongs to $B_{F}$, which implies that, for any $0\leq s_{1}\leq 1$, 
\begin{align}
	\big|\big\langle\Phi_{s_{1}},\varepsilon_{k}(T_{k})\Phi_{s_{1}}\big\rangle\big|&\leq C_{V}\cQ^{2}\sum_{p\in L_{k}}\big|\langle e_{p},T_{k}e_{p}\rangle\big|.
\end{align}
Then we obtain the following theorem:
\begin{thm}\label{thm:esti-varepk}
	For each $0\leq s_{1},s_{2}\leq 1$, it holds that each term in $\cW_{s_{1},s_{2}}^{(j_{1},j_{2})}(\varepsilon_{k})$ such that either $(j_{1},j_{2})\in\Sigma_{*}^{1}$, or $(j_{1},0),(0,j_{2})\in\Sigma_{*}^{2}$ with $j_{1},j_{2}\geq 2$, is bounded by $C_{V}k_{F}^{-1}\cQ^{2}m(\xi)$.
\end{thm}
\begin{proof}
Recall that the terms in $\cW_{s_{1},s_{2}}^{(j_{1},j_{2})}(\varepsilon_{k})$ are
\begin{align}\label{varepk-esti1}
	\sum_{(k,\zeta)\in\cC_{\xi}}\big|\big\langle\Phi_{s_{1}},\varepsilon_{k}\big(A_{k,s_{2}}^{(j_{1})}P_{\zeta;k}A_{k,s_{2}}^{(j_{2})}\big)\Phi_{s_{1}}\big\rangle\big|.
\end{align}
For $(j_{1},j_{2})\in\Sigma_{*}^{1}$, by Lemma \ref{lem:estimate-lambkp} and the lower bound \eqref{gap}, we have
\begin{align*}
	&\sum_{(k,\zeta)\in\cC_{\xi}}\big|\big\langle\Phi_{s_{1}},\varepsilon_{k}\big(A_{k,s_{2}}^{(j_{1})}P_{\zeta;k}A_{k,s_{2}}^{(j_{2})}\big)\Phi_{s_{1}}\big\rangle\big|\leq C_{V}\cQ^{2}\sum_{(k,\zeta)\in\cC_{\xi}}\sum_{p\in L_{k}}\big|\langle e_{p},A_{k,s_{2}}^{(j_{1})}e_{\zeta}\rangle\big|\big|\langle e_{p},A_{k,s_{2}}^{(j_{2})}e_{\zeta}\rangle\big|\nonumber\\
	&\leq C_{V}k_{F}^{-2}\cQ^{2}\sum_{(k,\zeta)\in\cC_{\xi}}\frac{\hat{V}_{k}^{j_{1}+j_{2}}}{\lambda_{k,\zeta}}\sum_{p\in L_{k}}\frac{1}{\lambda_{k,p}}\leq C_{V}k_{F}^{-1}\cQ^{2}m(\xi)\sum_{k\in\Z_{*}^{3}}\hat{V}_{k}^{j_{1}+j_{2}}\leq C_{V}k_{F}^{-1}\cQ^{2}m(\xi).
\end{align*} 

Similarly, for $(j_{1},0)\in\Sigma_{*}^{2}$ with $j_{1}\geq 2$, we have
\begin{align}\label{varepk-esti2}
	&\sum_{(k,\zeta)\in\cC_{\xi}}\big|\big\langle\Phi_{s_{1}},\varepsilon_{k}\big(A_{k,s_{2}}^{(j_{1})}P_{\zeta;k}\big)\Phi_{s_{1}}\big\rangle\big|\leq C_{V}\cQ^{2}\sum_{(k,\zeta)\in\cC_{\xi}}\big|\langle e_{\zeta},A_{k,s_{2}}^{(j_{1})}e_{\zeta}\rangle\big|\nonumber\\
	&\leq C_{V}k_{F}^{-1}\cQ^{2}m(\xi)\sum_{k\in\Z_{*}^{3}}\hat{V}_{k}^{j_{1}}\leq C_{V}k_{F}^{-1}\cQ^{2}m(\xi).
\end{align}
By the same computation, we obtain the same bound for terms in \eqref{varepk-esti1} with $(0,j_{2})\in\Sigma_{*}^{2}$ for $j_{2}\geq 2$.
\end{proof}

\subsection{Analysis on $\cW_{s_{1},s_{2}}^{(j_{1},j_{2})}(\cG_{k})$ for $(j_{1},j_{2})\in\Sigma_{*}^{1}$}\label{sec:esti-hoe}
In this section, we consider the error terms in $\cW_{s_{1},s_{2}}^{(j_{1},j_{2})}(\cG_{k})$ for $(j_{1},j_{2})\in\Sigma_{*}^{1}$, with $\cG_{k}$ denotes either $\cE_{1,k}$ or $\widetilde{\cE}_{2,k}$.  These terms can be readily estimated using results in \cite{CHN-23} and is given in the following proposition:
\begin{prop}\label{prop:esti-A5A6}
	For each $0\leq s_{1},s_{2}\leq 1$ and $(j_{1},j_{2})\in\Sigma_{*}^{1}$, each terms in $\cW_{s_{1},s_{2}}^{(j_{1},j_{2})}(\cG_{k})$ is bounded by $C_{V}k_{F}^{-3/2}m(\xi)$, where $\cG_{k}$ denotes either $\cE_{1,k}$ or $\widetilde{\cE}_{2,k}$.
\end{prop}
\begin{proof}
	By \cite[Propositions 4.6--4.11]{CHN-23} and Lemma \ref{lem:Gronwall}, each term in $\cW_{s_{1},s_{2}}^{(j_{1},j_{2})}(\cG_{k})$ is bounded by
	\begin{align*}
		C_{V}\Big(\sqrt{\sum_{k\in\Z_{*}^{3}}\sum_{p\in L_{k}}\max_{q\in L_{k}}\big|\langle e_{q},T_{k,s_{2}}^{(j_{1},j_{2})}e_{p}\rangle\big|^{2}}+k_{F}^{-1/2}\sqrt{\sum_{k\in\Z_{*}^{3}}\|T_{k,s_{2}}^{(j_{1},j_{2})}h_{k}^{-1/2}\|_{\rm HS}^{2}}\Big).
	\end{align*}
	Hence, it suffices to estimate these quantities.  By Proposition \ref{prop:prod-Kk-esti}, Lemma \ref{lem:estimate-lambkp} and the lower bound \eqref{gap}, we obtain
	\begin{align}\label{E1-mn-bdd1}
		&\sum_{k\in\Z_{*}^{3}}\sum_{p\in L_{k}}\max_{q\in L_{k}}\big|\langle e_{q},T_{k,s_{2}}^{(j_{1},j_{2})}e_{p}\rangle\big|^{2}=\sum_{(k,\zeta)\in\cC_{\xi}}\sum_{p\in L_{k}}\max_{q\in L_{k}}\big|\langle e_{\zeta},A_{k,s_{2}}^{(j_{1})}e_{q}\rangle\big|^{2}\big|\langle e_{\zeta},A_{k,s_{2}}^{(j_{2})}e_{p}\rangle\big|^{2}\nonumber\\
		&\leq k_{F}^{-4}\sum_{(k,\zeta)\in\cC_{\xi}}\sum_{p\in L_{k}}\frac{\hat{V}_{k}^{2(j_{1}+j_{2})}}{\lambda_{k,\zeta}^{2}(\lambda_{k,\zeta}+\lambda_{k,p})^{2}}\leq k_{F}^{-4}m(\xi)^{2}\sum_{k\in\Z_{*}^{3}}\hat{V}_{k}^{2(j_{1}+j_{2})}\sum_{p\in L_{k}}\frac{1}{\lambda_{k,p}}\nonumber\\
		&\leq C_{V}k_{F}^{-3}m(\xi)^{2},
	\end{align}
	and, similarly,
	\begin{align}\label{E1-mn-bdd2}
		\sum_{k\in\Z_{*}^{3}}\|T_{k,s_{2}}^{(j_{1},j_{2})}h_{k}^{-1/2}\|_{\rm HS}^{2}&=\sum_{(k,\zeta)\in\cC_{\xi}}\sum_{p\in L_{k}}\frac{\big|\langle e_{\zeta},A_{k,s_{2}}^{(j_{2})}e_{p}\rangle\big|^{2}}{\lambda_{k,p}}\|A_{k,s_{2}}^{(j_{1})}e_{\zeta}\|^{2}\nonumber\\
		&\leq Ck_{F}^{-3}\sum_{(k,\zeta)\in\cC_{\xi}}\frac{\hat{V}_{k}^{2(j_{1}+j_{2})}}{\lambda_{k,\zeta}^{2}}\sum_{p\in L_{k}}\frac{1}{\lambda_{k,p}}\leq Ck_{F}^{-2}m(\xi)^{2}.
	\end{align}
	This completes the proof.
\end{proof}

\subsection{Analysis on $\cW_{s_{1},s_{2}}^{(j_{1},j_{2})}(\cE_{1,k})$ for $(j_{1},j_{2})\in\Sigma_{*}^{2}$}\label{sec:esti-errEk1}
Recall for symmetric operator $T_{k}$ on $\ell^{2}(L_{k})$ that
\begin{align*}
	\cE_{1,k}(T_{k})&=\sum_{l\in\Z_{*}^{3}}\sum_{p\in L_{k}}\sum_{q\in L_{l}}b_{k}^{*}(T_{k}e_{p})\{\varepsilon_{k,l}(e_{p};e_{q}),b_{-l}^{*}(K_{-l}e_{-q})\},
\end{align*}
and
\begin{align*}
	\varepsilon_{k,l}(e_{p};e_{q})&=-(\delta_{p,q}a_{q-l}a_{p-k}^{*}+\delta_{p-k,q-l}a_{q}^{*}a_{p}).
\end{align*}
We see that $\cE_{1,k}(T_{k})$ splits into two sums:
\begin{align}
	-\cE_{1,k}(T_{k})&=\sum_{l\in\Z_{*}^{3}}\sum_{p\in L_{k}\cap L_{l}}b_{k}^{*}(T_{k}e_{p})\{a_{p-l}a_{p-k}^{*},b_{-l}^{*}(K_{-l}e_{-p})\}\nonumber\\
	&\quad\quad\quad+\sum_{l\in\Z_{*}^{3}}\sum_{p\in L_{k}'\cap L_{l}'}b_{k}^{*}(T_{k}e_{p+k})\{a_{p+l}^{*}a_{p+k},b_{-l}^{*}(K_{-l}e_{-p-l})\}.
\end{align}
Following the argument in \cite{CHN-23}, both of these sums can be written in the following schematic form
\begin{align}\label{E1-schematic}
	&\sum_{l\in\Z_{*}^{3}}\sum_{p\in S_{k}\cap S_{l}}b_{k}^{*}(T_{k}e_{p_{1}})\{\tilde{a}_{p_{2}}^{*}\tilde{a}_{p_{3}},b_{-l}^{*}(K_{-l}e_{p_{4}})\}\nonumber\\
	&=2\sum_{l\in\Z_{*}^{3}}\sum_{p\in S_{k}\cap S_{l}}b_{k}^{*}(T_{k}e_{p_{1}})\tilde{a}_{p_{2}}^{*}b_{-l}(K_{-l}e_{p_{4}})\tilde{a}_{p_{3}}\nonumber\\
	&\quad\quad\quad\quad\quad\quad+\sum_{l\in\Z_{*}^{3}}\sum_{p\in S_{k}\cap S_{l}}b_{k}^{*}(T_{k}e_{p_{1}})\tilde{a}_{p_{2}}^{*}[b_{-l}^{*}(K_{-l}e_{p_{4}}),\tilde{a}_{p_{3}}],
\end{align}
where $S_{k}$ denotes either $L_{k}$ or $L_{k}'$ and
\begin{align}\label{E1-gen-p}
	(p_{1},p_{2},p_{3},p_{4})&=\begin{cases}
		(p,p-l,p-k,-p)\quad&S_{k}=L_{k},\\
		(p+k,p+l,p+k,-p-l)&S_{k}=L_{k}'.
	\end{cases}
\end{align}
The terms in $\cW_{s_{1},s_{2}}^{(m)}(\cE_{1,k})$ are the matrix elements $\big\langle\Phi_{s_{1}},\cE_{1,k}(T_{k,s_{2}}^{(j_{1},j_{2})})\Phi_{s_{1}}\big\rangle$ for $(j_{1},j_{2})\in\Sigma_{*}$, where
\begin{align}\label{T1T2T3}
	T_{k,s_{2}}^{(j_{1},j_{2})}&:=\sum_{\zeta\in\cD_{k,\xi}}A_{k,s_{2}}^{(j_{1})}P_{\zeta;k}A_{k,s_{2}}^{(j_{2})}.
\end{align}
Hence, it estimate to consider the following terms for $(j_{1},j_{2})\in\Sigma_{*}^{2}$:
\begin{align}
	\label{Ek1-1}&\sum_{k,l\in\Z_{*}^{3}}\sum_{p\in S_{k}\cap S_{l}}\big\langle\Phi_{s_{1}},b_{k}^{*}\big(T_{k,s_{2}}^{(j_{1},j_{2})}e_{p_{1}}\big)\tilde{a}_{p_{2}}^{*}b_{-l}(K_{-l}e_{p_{4}})\tilde{a}_{p_{3}}\Phi_{s_{1}}\big\rangle,\\
	\label{Ek1-2}&\sum_{k,l\in\Z_{*}^{3}}\sum_{p\in S_{k}\cap S_{l}}\big\langle\Phi_{s_{1}},b_{k}^{*}\big(T_{k,s_{2}}^{(j_{1},j_{2})}e_{p_{1}}\big)\tilde{a}_{p_{2}}^{*}[b_{-l}^{*}(K_{-l}e_{p_{4}}),\tilde{a}_{p_{3}}]\Phi_{s_{1}}\big\rangle.
\end{align}

\begin{prop}\label{prop:esti-A1A2}
For each $0\leq s_{1},s_{2}\leq 1$ and $1\leq j_{1}\leq 3$, it holds that
\begin{align}
	\label{esti-A1}&\sum_{k,l\in\Z_{*}^{3}}\sum_{p\in S_{k}\cap S_{l}}\big|\big\langle\Phi_{s_{1}},b_{k}^{*}\big(T_{k,s_{2}}^{(j_{1},0)}e_{p_{1}}\big)\tilde{a}_{p_{2}}^{*}b_{-l}(K_{-l}e_{p_{4}})\tilde{a}_{p_{3}}\Phi_{s_{1}}\big\rangle\big|,\\
	\label{esti-A2}&\sum_{k,l\in\Z_{*}^{3}}\sum_{p\in S_{k}\cap S_{l}}\big|\big\langle\Phi_{s_{1}},b_{k}^{*}\big(T_{k,s_{2}}^{(j_{1},0)}e_{p_{1}}\big)\tilde{a}_{p_{2}}^{*}[b_{-l}^{*}(K_{-l}e_{p_{4}}),\tilde{a}_{p_{3}}]\Phi_{s_{1}}\big\rangle\big|,
\end{align}
are bounded by
\begin{align*}
	C_{V}k_{F}^{-1}\big(\cQ m(\xi)+m(\xi)^{1/2}\cdot\sup_{0\leq \tau\leq 1}\|\tilde{a}_{\xi}\Phi_{\tau}\|\big).
\end{align*}  
\end{prop}
\begin{proof}
For \eqref{esti-A1}, by Proposition \ref{prop:prod-Kk-esti}, Cauchy-Schwarz inequality, pull-through formula \eqref{pull-through}, the bounds \eqref{excitation-bdd} and \eqref{gap}, we obtain for $S_{k}=L_{k}$ that 
\begin{align}\label{A1m-esti}
	&\sum_{k,l\in\Z_{*}^{3}}\sum_{p\in L_{k}\cap L_{l}}\big|\big\langle\Phi_{s_{1}},b_{k}^{*}\big(T_{k,s_{2}}^{(j_{1},0)}e_{p_{1}}\big)\tilde{a}_{p_{2}}^{*}b_{-l}(K_{-l}e_{p_{4}})\tilde{a}_{p_{3}}\Phi_{s_{1}}\big\rangle\big|\nonumber\\
	\displaybreak
	&\leq \sum_{(k,\zeta)\in\cC_{\xi}}\sum_{l\in\Z_{*}^{3}}\sum_{p\in L_{k}\cap L_{l}}\delta_{p,\zeta}\|A_{k,s_{2}}^{(j_{1})}e_{\zeta}\|\|(\cN_{E}+1)^{-1}\cN_{-l}^{1/2}\tilde{a}_{p-k}\Phi_{s_{1}}\|\nonumber\\
	&\quad\quad\quad\quad\quad\quad\quad\quad\quad\quad\quad\quad\quad\quad\times\|K_{l}e_{p}\|\|(\cN_{E}+1)\cN_{k}^{1/2}\tilde{a}_{p-l}\Phi_{s_{1}}\|\nonumber\\
	&\leq k_{F}^{-1} \sum_{k,l\in\Z_{*}^{3}}\sum_{\zeta\in\cD_{k,\xi}}\frac{\hat{V}_{k}\hat{V}_{l}}{\sqrt{\lambda_{k,\zeta}\lambda_{l,\zeta}}}\|\tilde{a}_{\zeta-k}\Phi_{s_{1}}\|\|(\cN_{E}+1)\cN_{k}^{1/2}\tilde{a}_{\zeta-l}\Phi_{s_{1}}\|\nonumber\\
	&\leq k_{F}^{-1}m(\xi)^{1/2}\sum_{k,l\in\Z_{*}^{3}}\hat{V}_{k}\hat{V}_{l}\Big(\chi_{L_{k}\cap L_{l}}(\xi)m(\xi)^{1/2}\|\tilde{a}_{\xi-k}\Phi_{s_{1}}\|\|\cN_{k}^{1/2}(\cN_{E}+1)\tilde{a}_{\xi-l}\Phi_{s_{1}}\|\nonumber\\
	&\quad\quad\quad\quad\quad\quad\quad\quad\quad\quad+\chi_{L_{k}\cap L_{l}}(\xi+k)\|\tilde{a}_{\xi}\Phi_{s_{1}}\|\|\cN_{k}^{1/2}(\cN_{E}+1)\tilde{a}_{\xi+k-l}\Phi_{s_{1}}\|\Big)\nonumber\\
	&\leq k_{F}^{-1}m(\xi)^{1/2}\Big(\sum_{l\in\Z_{*}^{3}}\hat{V}_{l}\chi_{L_{l}}(\xi)m(\xi)^{1/2}\cQ\sqrt{\sum_{k\in\Z_{*}^{3}}\hat{V}_{k}^{2}}\sqrt{\sum_{k\in\Z_{*}^{3}}\|\cN_{k}^{1/2}\tilde{a}_{\xi-l}\cN_{E}\Phi_{s_{1}}\|^{2}}\nonumber\\
	&\quad\quad+\sum_{k\in\Z_{*}^{3}}\hat{V}_{k}\chi_{L_{k}'}(\xi)\|\tilde{a}_{\xi}\Phi_{s_{1}}\|\sqrt{\sum_{l\in\Z_{*}^{3}}\hat{V}_{l}^{2}}\sqrt{\sum_{l\in\Z_{*}^{3}}\chi_{L_{l}}(\xi+k)\|\tilde{a}_{\xi+k-l}\cN_{k}^{1/2}\cN_{E}\Phi_{s_{1}}\|^{2}}\Big)\nonumber\\
	&\leq C_{V}k_{F}^{-1}m(\xi)^{1/2}\Big(m(\xi)^{1/2}\cQ\sqrt{\sum_{l\in\Z_{*}^{3}}\hat{V}_{l}^{2}}\sqrt{\sum_{l\in\Z_{*}^{3}}\|\tilde{a}_{\xi-l}\cN_{E}^{3/2}\Phi_{s_{1}}\|^{2}}\nonumber\\
	&\quad\quad\quad\quad\quad\quad\quad\quad\quad\quad\quad\quad+\|\tilde{a}_{\xi}\Phi_{s_{1}}\|\sqrt{\sum_{k\in\Z_{*}^{3}}\hat{V}_{k}^{2}}\sqrt{\sum_{k\in\Z_{*}^{3}}\|\cN_{k}^{1/2}\cN_{E}^{3/2}\Phi_{s_{1}}\|^{2}}\Big)\nonumber\\
	&\leq C_{V}k_{F}^{-1}\big(\chi_{B_{F}^{c}}(\xi)m(\xi)\cQ+\chi_{B_{F}}(\xi)m(\xi)^{1/2}\|\tilde{a}_{\xi}\Phi_{s_{1}}\|\big).
\end{align}
By the same computation, the case $S_{k}=L_{k}'$ yields the bound
\begin{align}
	&\sum_{k,l\in\Z_{*}^{3}}\sum_{p\in L_{k}'\cap L_{l}'}\big|\big\langle\Phi_{s_{1}},b_{k}^{*}\big(T_{k,s_{2}}^{(j_{1},0)}e_{p_{1}}\big)\tilde{a}_{p_{2}}^{*}b_{-l}(K_{-l}e_{p_{4}})\tilde{a}_{p_{3}}\Phi_{s_{1}}\big\rangle\big|\nonumber\\
	&\leq C_{V}k_{F}^{-1}\big(\chi_{B_{F}}(\xi)m(\xi)\cQ+\chi_{B_{F}^{c}}(\xi)m(\xi)^{1/2}\|\tilde{a}_{\xi}\Phi_{s_{1}}\|\big).
\end{align}
In summary, we obtain
\begin{align}
	&\sum_{k,l\in\Z_{*}^{3}}\sum_{p\in S_{k}\cap S_{l}}\big|\big\langle\Phi_{s_{1}},b_{k}^{*}\big(T_{k,s_{2}}^{(j_{1},0)}e_{p_{1}}\big)\tilde{a}_{p_{2}}^{*}b_{-l}(K_{-l}e_{p_{4}})\tilde{a}_{p_{3}}\Phi_{s_{1}}\big\rangle\big|\nonumber\\
	&\leq C_{V}k_{F}^{-1}\big(m(\xi)\cQ+m(\xi)^{1/2}\|\tilde{a}_{\xi}\Phi_{s_{1}}\|\big).
\end{align}

Next, for \eqref{esti-A2}, we first compute the commutator
\begin{align}\label{E1-com}
	[b_{-l}(K_{-l}e_{p_{4}}),\tilde{c}_{p_{3}}^{*}]&=\begin{cases}
		-\chi_{L_{-l}'}(p_{3})\langle K_{-l}e_{p_{4}},e_{p_{3}-l}\rangle\tilde{a}_{p_{3}-l}\quad&S_{k}=L_{k},\\
		\chi_{L_{-l}}(p_{3})\langle K_{-l}e_{p_{4}},e_{p_{3}}\rangle\tilde{a}_{p_{3}+l}&S_{k}=L_{k}',
	\end{cases}.
\end{align}
which, by \cite[Eq. (4.28)]{CHN-23}, satifies
\begin{align}
	\begin{cases}\label{E1-com2}
		\big|\chi_{L_{-l}}(p_{3}-l)\langle K_{-l}e_{p_{4}},e_{p_{3}-l}\rangle\big|\quad&S_{k}=L_{k}\\
		\big|\chi_{L_{-l}}(p_{3})\langle K_{-l}e_{p_{4}},e_{p_{3}}\rangle\big|&S_{k}=L_{k}'
	\end{cases}\leq C\frac{k_{F}^{-1}\hat{V}_{-l}}{\sqrt{\lambda_{k,p_{1}}\lambda_{-l,p_{4}}}}.
\end{align}
Then, using \eqref{excitation-bdd}--\eqref{excitation-bdd5} and \eqref{gap2}, we obtain
\begin{align}\label{A2m-esti}
	&\sum_{k,l\in\Z_{*}^{3}}\sum_{p\in S_{k}\cap S_{l}}\big|\big\langle\Phi_{s_{1}},b_{k}^{*}\big(T_{k,s_{2}}^{(j_{1},0)}e_{p_{1}}\big)\tilde{a}_{p_{2}}^{*}[b_{-l}^{*}(K_{-l}e_{p_{4}}),\tilde{a}_{p_{3}}]\Phi_{s_{1}}\big\rangle\big|\nonumber\\
	&\leq \sum_{k,l\in\Z_{*}^{3}}\sum_{p\in S_{k}\cap S_{l}}\|b_{k}(T_{k,s_{2}}^{(j_{1},0)}e_{p_{1}})\tilde{a}_{p_{2}}[b_{-l}(K_{-l}e_{p_{4}}),\tilde{a}_{p_{3}}^{*}]^{*}\Phi_{s_{1}}\|\nonumber\\
	&\leq Ck_{F}^{-1}\sum_{(k,\zeta)\in\cC_{\xi}}\sum_{l\in\Z_{*}^{3}}\sum_{p\in S_{k}\cap S_{l}}\delta_{p_{1},\zeta}\frac{\hat{V}_{-l}}{\sqrt{\lambda_{k,p_{1}}\lambda_{-l,p_{4}}}}\|A_{k,s_{2}}^{(j_{1})}e_{\zeta}\|\|\cN_{k}^{1/2}\tilde{a}_{p_{2}}\tilde{a}_{p_{3}\mp l}\Phi_{s_{1}}\|\nonumber\\
	&\leq Ck_{F}^{-3/2}\sum_{(k,\zeta)\in\cC_{\xi}}\sum_{l\in\Z_{*}^{3}}\sum_{p\in S_{k}\cap S_{l}}\delta_{p_{1},\zeta}\frac{\hat{V}_{-l}}{\sqrt{\lambda_{k,p_{1}}\lambda_{-l,p_{4}}}}\frac{\hat{V}_{k}^{j_{1}}}{\sqrt{\lambda_{k,\zeta}}}\|\cN_{k}^{1/2}\tilde{a}_{p_{2}}\tilde{a}_{p_{3}\mp l}\Phi_{s_{1}}\|\nonumber\\
	&\leq Ck_{F}^{-3/2}m(\xi)\sum_{(k,\zeta)\in\cC_{\xi}}\hat{V}_{k}^{j_{1}}\sqrt{\sum_{l\in\Z_{*}^{3}}\hat{V}_{l}^{2}}\sqrt{\sum_{l\in\Z_{*}^{3}}\sum_{p\in S_{k}\cap S_{l}}\frac{\delta_{p_{1},\zeta}}{\lambda_{-l,p_{4}}}\|\cN_{k}^{1/2}\tilde{a}_{p_{2}}\tilde{a}_{p_{3}\mp l}\Phi_{s_{1}}\|^{2}}.
\end{align}
By \eqref{excitation-bdd}--\eqref{excitation-bdd5} and operator inequality $\|\tilde{a}_{q}\|=1$, the last factor in \eqref{A2m-esti} gives
\begin{align}
	\sum_{l\in\Z_{*}^{3}}\sum_{p\in S_{k}\cap S_{l}}\frac{\delta_{p_{1},\zeta}}{\lambda_{-l,p_{4}}}\|\cN_{k}^{1/2}\tilde{a}_{p_{2}}\tilde{a}_{p_{3}\mp l}\Phi_{s_{1}}\|^{2}&=\sum_{l\in\Z_{*}^{3}}\begin{cases}
		\chi_{L_{l}}(\zeta)\|\cN_{k}^{1/2}\tilde{a}_{\zeta+l}\Phi_{s_{1}}\|^{2}\quad&S_{k}=L_{k}\\
		\chi_{L_{l}'}(\zeta-k)\|\cN_{k}^{1/2}\tilde{a}_{\zeta-k-l}\Phi_{s_{1}}\|^{2}&S_{k}=L_{k}'
	\end{cases}\nonumber\\
	&\leq\|\cN_{k}^{1/2}\cN_{E}^{1/2}\Phi_{s_{1}}\|^{2},
\end{align}
where we recall the values of $p_{2}$ from \eqref{E1-gen-p}.  Consequently, \eqref{A2m-esti} becomes
\begin{align}
	&\sum_{k,l\in\Z_{*}^{3}}\sum_{p\in S_{k}\cap S_{l}}\big|\big\langle\Phi_{s_{1}},b_{k}^{*}\big(T_{k,s_{2}}^{(j_{1},0)}e_{p_{1}}\big)\tilde{a}_{p_{2}}^{*}[b_{-l}^{*}(K_{-l}e_{p_{4}}),\tilde{a}_{p_{3}}]\Phi_{s_{1}}\big\rangle\big|\nonumber\\
	&\leq C_{V}k_{F}^{-3/2}m(\xi)\sum_{(k,\zeta)\in\cC_{\xi}}\hat{V}_{k}^{j_{1}}\|\cN_{k}^{1/2}\cN_{E}\Phi_{s_{1}}\|\nonumber\\
	&\leq C_{V}k_{F}^{-3/2}m(\xi)\sqrt{\sum_{k\in\Z_{*}^{3}}\hat{V}_{k}^{2j_{1}}}\sqrt{\sum_{k\in\Z_{*}^{3}}\|\cN_{k}^{1/2}\cN_{E}^{1/2}\Phi_{s_{1}}\|^{2}}\leq C_{V}k_{F}^{-3/2}m(\xi).
\end{align}
This completes the proof.
\end{proof}
\begin{prop}\label{prop:esti-A3A4}
For each $0\leq s_{1},s_{2}\leq 1$ and $1\leq j_{2}\leq 3$, it holds for any $\delta>0$ that
\begin{align}\label{esti-A3}
	&\sum_{k,l\in\Z_{*}^{3}}\sum_{p\in S_{k}\cap S_{l}}\big|\big\langle\Phi_{s_{1}},b_{k}^{*}\big(T_{k,s_{2}}^{(0,j_{2})}e_{p_{1}}\big)\tilde{a}_{p_{2}}^{*}b_{-l}(K_{-l}e_{p_{4}})\tilde{a}_{p_{3}}\Phi_{s_{1}}\big\rangle\big|,\\
	\label{esti-A4}&\sum_{k,l\in\Z_{*}^{3}}\sum_{p\in S_{k}\cap S_{l}}\big|\big\langle\Phi_{s_{1}},b_{k}^{*}\big(T_{k,s_{2}}^{(0,j_{2})}e_{p_{1}}\big)\tilde{a}_{p_{2}}^{*}[b_{-l}^{*}(K_{-l}e_{p_{4}}),\tilde{a}_{p_{3}}]\Phi_{s_{1}}\big\rangle\big|,
\end{align}
are bounded by $C_{\delta,V}k_{F}^{-1}\cQ^{1-\delta}m(\xi)$.  
\end{prop}
\begin{proof}
For \eqref{esti-A3}, by Cauchy-Schwarz inequality, Lemma \ref{lem:excitation-Gronwall} and the lower bound \eqref{gap}, we obtain for each $\delta>0$ that
\begin{align}\label{A3m-esti}
	&\sum_{k,l\in\Z_{*}^{3}}\sum_{p\in S_{k}\cap S_{l}}\big|\big\langle\Phi_{s_{1}},b_{k}^{*}\big(T_{k,s_{2}}^{(0,j_{2})}e_{p_{1}}\big)\tilde{a}_{p_{2}}^{*}b_{-l}(K_{-l}e_{p_{4}})\tilde{a}_{p_{3}}\Phi_{s_{1}}\big\rangle\big|\nonumber\\
	\displaybreak
	&\leq C\sum_{(k,\zeta)\in\cC_{\xi}}\sum_{l\in\Z_{*}^{3}}\sum_{p\in S_{k}\cap S_{l}}\big|\langle e_{\zeta},A_{k,s_{2}}^{(j_{2})}e_{p_{1}}\rangle\big|\|K_{-l}e_{p_{4}}\|\|\tilde{a}_{\zeta}\tilde{a}_{\zeta-k}\tilde{a}_{p_{2}}\Phi_{s_{1}}\|\|\cN_{-l}^{1/2}\tilde{a}_{p_{3}}\Phi_{s_{1}}\|\nonumber\\
	&\leq Ck_{F}^{-1}\sum_{(k,\zeta)\in\cC_{\xi}}\sum_{l\in\Z_{*}^{3}}\sum_{p\in S_{k}\cap S_{l}}\frac{\hat{V}_{k}^{j_{2}}}{\lambda_{k,\zeta}+\lambda_{k,p_{1}}}\|K_{-l}e_{p_{4}}\|\|\tilde{a}_{\zeta}\tilde{a}_{\zeta-k}\tilde{a}_{p_{2}}\Phi_{s_{1}}\|\|\cN_{-l}^{1/2}\tilde{a}_{p_{3}}\Phi_{s_{1}}\|\nonumber\\
	&\leq Cm(\xi)k_{F}^{-1}\sum_{(k,\zeta)\in\cC_{\xi}}\sum_{l\in\Z_{*}^{3}}\hat{V}_{k}^{j_{2}}\|\tilde{a}_{\zeta}\tilde{a}_{\zeta-k}\Phi_{s_{1}}\|\sqrt{\sum_{p\in S_{k}\cap S_{l}}\|K_{-l}e_{p_{4}}\|^{2}}\sqrt{\sum_{p\in S_{k}\cap S_{l}}\|\tilde{a}_{p_{3}}\cN_{-l}^{1/2}\Phi_{s_{1}}\|^{2}}\nonumber\\
	&\leq Cm(\xi)k_{F}^{-1}\sum_{(k,\zeta)\in\cC_{\xi}}\hat{V}_{k}^{j_{2}}\|\tilde{a}_{\zeta}\tilde{a}_{\zeta-k}\Phi_{s_{1}}\|\sqrt{\sum_{l\in\Z_{*}^{3}}\|K_{l}\|_{\rm HS}^{2}}\sqrt{\sum_{l\in\Z_{*}^{3}}\|\cN_{-l}^{1/2}\cN_{E}^{1/2}\Phi_{s_{1}}\|^{2}}\nonumber\\
	&\leq C_{V}m(\xi)k_{F}^{-1}\sqrt{\sum_{k\in\Z_{*}^{3}}\hat{V}_{k}^{2j_{2}}}\sqrt{\sum_{(k,\zeta)\in\cC_{\xi}}\|\tilde{a}_{\zeta}\tilde{a}_{\zeta-k}\Phi_{s_{1}}\|^{2}}\nonumber\\
	&\leq C_{V}m(\xi)k_{F}^{-1}\sqrt{\sum_{k\in\Z_{*}^{3}}\Big(\chi_{L_{k}}(\xi)\|\tilde{a}_{\xi-k}\tilde{a}_{\xi}\Phi_{s_{1}}\|^{2}+\chi_{L_{k}'}(\xi)\|\tilde{a}_{k+\xi}\tilde{a}_{\xi}\Phi_{s_{1}}\|^{2}\Big)}\nonumber\\
	&\leq C_{V}m(\xi)k_{F}^{-1}\|\tilde{a}_{\xi}\cN_{E}^{1/2}\Phi_{s_{1}}\|\leq C_{\delta,V}k_{F}^{-1}\cQ^{1-\delta}m(\xi),
\end{align}
where we have used the estimates $\|K_{l}\|_{\rm HS}\leq C\hat{V}_{l}$ and the lower bound \eqref{gap}.  

Next, by relations \eqref{E1-com}--\eqref{E1-com2}, Lemmas \ref{lem:estimate-lambkp}, \ref{lem:excitation-Gronwall}, and the operator norm $\|\tilde{a}_{q}\|=1$, we obtain for each $\delta>0$ that
\begin{align*}
	&\sum_{k,l\in\Z_{*}^{3}}\sum_{p\in S_{k}\cap S_{l}}\big|\big\langle\Phi_{s_{1}},b_{k}^{*}\big(T_{k,s_{2}}^{(0,j_{2})}e_{p_{1}}\big)\tilde{a}_{p_{2}}^{*}[b_{-l}^{*}(K_{-l}e_{p_{4}}),\tilde{a}_{p_{3}}]\Phi_{s_{1}}\big\rangle\big|\nonumber\\
	&\leq C\sum_{l\in\Z_{*}^{3}}\sum_{(k,\zeta)\in\cC_{\xi}}\sum_{p\in S_{k}\cap S_{l}}\big|\langle e_{\zeta},A_{k,s_{2}}^{(j_{2})}e_{p_{1}}\rangle\big|\|[b_{-l}(K_{-l}e_{p_{4}}),\tilde{a}_{p_{3}}^{*}]\tilde{a}_{p_{2}}b_{k,\zeta}\Phi_{s_{1}}\|\nonumber\\
	&\leq C k_{F}^{-2}\sum_{l\in\Z_{*}^{3}}\sum_{(k,\zeta)\in\cC_{\xi}}\sum_{p\in S_{k}\cap S_{l}}\frac{\hat{V}_{-l}}{\sqrt{\lambda_{k,p_{1}}\lambda_{-l,p_{4}}}}\frac{\hat{V}_{k}^{j_{2}}}{\lambda_{k,\zeta}+\lambda_{k,p_{1}}}\|\tilde{a}_{p_{3}\mp l}\tilde{a}_{p_{2}}\tilde{a}_{\zeta}\tilde{a}_{\zeta-k}\Phi_{s_{1}}\|\nonumber\\
	&\leq Cm(\xi)k_{F}^{-2}\sum_{k,l\in\Z_{*}^{3}}\sum_{p\in S_{k}\cap S_{l}}\frac{\hat{V}_{k}^{j_{2}}\hat{V}_{-l}}{\sqrt{\lambda_{k,p_{1}}\lambda_{-l,p_{4}}}}\|\tilde{a}_{p_{3}\mp l}\tilde{a}_{p_{2}}\tilde{a}_{\xi}\Phi_{s_{1}}\|\nonumber\\
	&\leq Cm(\xi)k_{F}^{-2}\sum_{p}\sum_{l\in\Z_{*}^{3}}\frac{\hat{V}_{-l}\chi_{S_{l}}(p)}{\sqrt{\lambda_{-l,p_{4}}}}\sqrt{\sum_{k\in\Z_{*}^{3}}\frac{\hat{V}_{k}^{2j_{1}}\chi_{S_{k}}(p)}{\lambda_{k,p_{1}}}}\sqrt{\sum_{k\in\Z_{*}^{3}}\chi_{S_{k}}(p)\|\tilde{a}_{p_{3}\mp l}\tilde{a}_{p_{2}}\tilde{a}_{\xi}\Phi_{s_{1}}\|^{2}}\nonumber\\
	&\leq Cm(\xi)k_{F}^{-2}\sum_{p}\sqrt{\sum_{l\in\Z_{*}^{3}}\frac{\hat{V}_{-l}^{2}\chi_{S_{l}}(p)}{\lambda_{-l,p_{4}}}}\sqrt{\sum_{k\in\Z_{*}^{3}}\frac{\hat{V}_{k}^{2j_{1}}\chi_{S_{k}}(p)}{\lambda_{k,p_{1}}}}\sqrt{\sum_{l\in\Z_{*}^{3}}\chi_{ S_{l}}(p)\|\tilde{a}_{p_{2}}\cN_{E}^{1/2}\tilde{a}_{\xi}\Phi_{s_{1}}\|^{2}}\nonumber\\
	&\leq Cm(\xi)k_{F}^{-2}\|\tilde{a}_{\xi}\cN_{E}\Phi_{s_{1}}\|\sqrt{\sum_{l\in\Z_{*}^{3}}\hat{V}_{l}^{2}\sum_{p\in L_{l}}\frac{1}{\lambda_{l,p}}}\sqrt{\sum_{k\in\Z_{*}^{3}}\hat{V}_{k}^{2j_{1}}\sum_{p\in L_{k}}\frac{1}{\lambda_{k,p}}}\leq C_{\delta,V}k_{F}^{-1}\cQ^{1-\delta}m(\xi).
\end{align*}
This completes the proof.
\end{proof}
In summary, we obtain the following theorem from Propositions \ref{prop:esti-A1A2}--\ref{prop:esti-A3A4}:
\begin{thm}\label{thm:errorE1k}
For each $0\leq s_{1},s_{2}\leq 1$ and $(j_{1},j_{2})\in\Sigma_{*}^{2}$, each term in $\cW_{s_{1},s_{2}}^{(j_{1},j_{2})}(\cE_{1,k})$ is bounded by
\begin{align*}
	C_{V}k_{F}^{-1}m(\xi)^{1/2}\cdot\sup_{0\leq \tau\leq 1}\|\tilde{a}_{\xi}\Phi_{\tau}\|+C_{\delta,V}k_{F}^{-1}\cQ^{1-\delta}m(\xi)
\end{align*}
for each $\delta>0$.
\end{thm}

\subsection{Analysis on $\cW_{s_{1},s_{2}}^{(j_{1},j_{2})}(\widetilde{\cE}_{2,k})$ for $(j_{1},j_{2})\in\Sigma_{*}^{2}$}\label{sec:esti-errEk2}
Again, recall for symmetric operator $T_{k}$ on $\ell^{2}(L_{k})$ that
\begin{align*}
	\cE_{2,k}(T_{k})&=\frac{1}{2}\sum_{l\in\Z_{*}^{3}}\sum_{p\in L_{k}}\sum_{q\in L_{l}}\{b_{k}(T_{k}e_{p}),\{\varepsilon_{-k,-l}(e_{-p};e_{-q}),b_{l}^{*}(K_{l}e_{q})\}\},
\end{align*}
where
\begin{align*}
	\varepsilon_{-k,-l}(e_{-p};e_{-q})&=-\big(\delta_{p,q}\tilde{a}_{-q+l}^{*}\tilde{a}_{-p+k}+\delta_{p-k,q-l}\tilde{a}_{-q}^{*}\tilde{a}_{-p}\big).
\end{align*}
Just as for $\cE_{1,k}(T_{k})$, we can split $\cE_{2,k}(T_{k})$ into two sums
\begin{align}
	-2\cE_{2,k}(T_{k})&=\sum_{l\in\Z_{*}^{3}}\sum_{p\in L_{k}\cap L_{l}}\{b_{k}(T_{k}e_{p}),\{\tilde{a}_{-p+l}^{*}\tilde{a}_{-p+k},b_{l}^{*}(K_{l}e_{p})\}\}\nonumber\\
	&\quad\quad+\sum_{l\in\Z_{*}^{3}}\sum_{p\in L_{k}'\cap L_{l}'}\{b_{k}(T_{k}e_{p+k}),\{\tilde{a}_{-p-l}^{*}\tilde{a}_{-p-k},b_{l}^{*}(K_{l}e_{p+l})\}\},
\end{align}
and again we can write the summand of these sums into schematic form:
\begin{align}\label{E2k-generic}
	\sum_{l\in\Z_{*}^{3}}\sum_{p\in S_{k}\cap S_{l}}\{b_{k}(T_{k}e_{p_{1}}),\{\tilde{a}_{p_{2}}^{*}\tilde{a}_{p_{3}},b_{l}^{*}(K_{l}e_{p_{4}})\}\},
\end{align}
where
\begin{align*}
	(p_{1},p_{2},p_{3},p_{4})&=\begin{cases}
		(p,-p+l,-p+k,p)\quad&S_{k}=L_{k},\\
		(p+k,-p-l,-p-k,p+l)&S_{k}=L_{k}'.
	\end{cases}
\end{align*}

Next, we put this schematic form into normal order: According to \cite[Eq. (4.40)]{CHN-23}, we have
\begin{align}
	&\{b_{k}(T_{k}e_{p_{1}}),\{\tilde{a}_{p_{2}}^{*}\tilde{a}_{p_{3}},b_{l}^{*}(K_{l}e_{p_{4}})\}\}\nonumber\\
	\label{top}&=4\tilde{a}_{p_{2}}^{*}b_{l}^{*}(K_{l}e_{p_{4}})b_{k}(T_{k}e_{p_{1}})\tilde{a}_{p_{3}}+2\tilde{a}_{p_{2}}^{*}[b_{k}(T_{k}e_{p_{1}}),b_{l}^{*}(K_{l}e_{p_{4}})]\tilde{a}_{p_{3}}\\
	\label{single-com}&\quad+2\tilde{a}_{p_{2}}^{*}[b_{l}(K_{l}e_{p_{4}}),\tilde{a}_{p_{3}}^{*}]^{*}b_{k}(T_{k}e_{p_{1}})+2b_{l}^{*}(K_{l}e_{p_{4}})[b_{k}(T_{k}e_{p_{1}}),\tilde{a}_{p_{2}}^{*}]\tilde{a}_{p_{3}}\\
	\label{double-com1}&\quad+\tilde{a}_{p_{2}}^{*}[b_{k}(T_{k}e_{p_{1}}),[b_{l}(K_{l}e_{p_{4}}),\tilde{a}_{p_{3}}^{*}]^{*}]+2[b_{l}(K_{l}e_{p_{4}}),[b_{k}(T_{k}e_{p_{1}}),\tilde{a}_{p_{2}}^{*}]^{*}]^{*}\tilde{a}_{p_{3}}\\
	\label{double-com2}&\quad-[b_{l}(K_{l}e_{p_{4}}),\tilde{a}_{p_{3}}^{*}]^{*}[b_{k}(T_{k}e_{p_{1}}),\tilde{a}_{p_{2}}^{*}]+\{[b_{k}(T_{k}e_{p_{1}}),\tilde{a}_{p_{2}}^{*}],[b_{l}(K_{l}e_{p_{4}}),\tilde{a}_{p_{3}}^{*}]^{*}\}.
\end{align}
Note that only the last term in \eqref{double-com2} is a constant that does not depend on any creation or annihilation operator.  Since all other terms are annihilated by taking expectation w.r.t. $\Psi_{FS}$, the constant term yields precisely $\langle\Psi_{FS},\cE_{2,k}(T_{k})\Psi_{FS}\rangle$, whence estimating other terms yield a bound for 
\begin{align*}
	\widetilde{\cE}_{2,k}(T_{k})&=\cE_{2,k}(T_{k})-\langle\Psi_{FS},\cE_{2,k}(T_{k})\Psi_{FS}\rangle.
\end{align*}
Again, in what follows, we estimate each term in \eqref{top}--\eqref{double-com2} for $T_{k}$ is given by $T_{k,s_{2}}^{(j_{1},j_{2})}$ defined in \eqref{T1T2T3} for $(j_{1},j_{2})\in\Sigma_{*}$.

\smallskip

\noindent\textbf{Estimation of the Top Terms}.  We begin by estimating the ``top" terms in \eqref{top}:
\begin{align*}
	\sum_{k,l\in\Z_{*}^{3}}\sum_{p\in S_{k}\cap S_{l}}\tilde{a}_{p_{2}}^{*}b_{l}^{*}(K_{l}e_{p_{4}})b_{k}(T_{k}e_{p_{1}})\tilde{a}_{p_{3}},\quad\sum_{k,l\in\Z_{*}^{3}}\sum_{p\in S_{k}\cap S_{l}}\tilde{a}_{p_{2}}^{*}[b_{k}(T_{k}e_{p_{1}}),b_{l}^{*}(K_{l}e_{p_{4}})]\tilde{a}_{p_{3}}.
\end{align*}
By Lemma \ref{lem:CR-excitation}, the commutator term becomes
\begin{align}\label{top-com}
	&\sum_{k,l\in\Z_{*}^{3}}\sum_{p\in S_{k}\cap S_{l}}\tilde{a}_{p_{2}}^{*}[b_{k}(T_{k}e_{p_{1}}),b_{l}^{*}(K_{l}e_{p_{4}})]\tilde{a}_{p_{3}}\\
	&=\sum_{k\in\Z_{*}^{3}}\sum_{p\in S_{k}}\langle T_{k}e_{p_{1}},K_{k}e_{p_{1}}\rangle\tilde{a}_{p_{3}}^{*}\tilde{a}_{p_{3}}+\sum_{k,l\in\Z_{*}^{3}}\sum_{p\in S_{k}\cap S_{l}}\tilde{a}_{p_{2}}^{*}\varepsilon_{k,l}(T_{k}e_{p_{1}};K_{l}e_{p_{4}})\tilde{a}_{p_{3}},\nonumber
\end{align}
since $p_{1}=p_{4}$ and $p_{2}=p_{3}$ when $k=l$.  The exchange correction of the second sum splits as
\begin{align}
	-\varepsilon_{k,l}(T_{k}e_{p_{1}};K_{l}e_{p_{4}})&=\sum_{q\in L_{k}\cap L_{l}}\langle T_{k}e_{p_{1}},e_{q}\rangle\langle e_{q},K_{l}e_{p_{4}}\rangle\tilde{a}_{q-l}^{*}\tilde{a}_{q-k}\nonumber\\
	&\quad\quad+\sum_{q\in L_{k}'\cap L_{l}'}\langle T_{k}e_{p_{1}},e_{q+k}\rangle\langle e_{q+l},K_{l}e_{p_{4}}\rangle\tilde{a}_{q+l}^{*}\tilde{a}_{q+k},
\end{align}
and we see that these sums both take the schematic form
\begin{align}
	\sum_{q\in S_{k}'\cap S_{l}'}\langle T_{k}e_{p_{1}},e_{q_{1}}\rangle\langle e_{q_{4}},K_{l}e_{p_{4}}\rangle\tilde{a}_{q_{2}}^{*}\tilde{a}_{q_{3}}.
\end{align}
Hence, to estimate $\sum_{k,l\in\Z_{*}^{3}}\sum_{p\in S_{k}\cap S_{l}}\tilde{a}_{p_{2}}^{*}\varepsilon_{k,l}(T_{k}e_{p_{1}};K_{l}e_{p_{4}})\tilde{a}_{p_{3}}$, it suffices to estimate
\begin{align}
	\sum_{k,l\in\Z_{*}^{3}}\sum_{p\in S_{k}\cap S_{l}}\sum_{q\in S_{k}'\cap S_{l}'}\langle T_{k}e_{p_{1}},e_{q_{1}}\rangle\langle e_{q_{4}},K_{l}e_{p_{4}}\rangle\tilde{a}_{p_{2}}^{*}\tilde{a}_{q_{2}}^{*}\tilde{a}_{q_{3}}\tilde{a}_{p_{3}}.
\end{align}

Before proceeding, we remark that the first top term in \eqref{top} with $T_{k,s_{2}}^{(1,0)}$ is particularly hard to estimate in comparison to those $T_{k,s_{2}}^{(j_{1},0)}$ with $j_{1}\geq 2$; if we naively estimate them with the same method, then the term that involves $T_{k,s_{2}}^{(1,0)}$ is of order $O(k_{F}^{-1})$, which is the same order as the bosonization contribution.  Hence, we single out this term
\begin{align}\label{top-hard}
	\sum_{k,l\in\Z_{*}^{3}}\sum_{p\in S_{k}\cap S_{l}}\big\langle\Phi_{s_{1}},\tilde{a}_{p_{2}}^{*}b_{l}^{*}(K_{l}e_{p_{4}})b_{k}(T_{k,s_{2}}^{(1,0)}e_{p_{1}})\tilde{a}_{p_{3}}\Phi_{s_{1}}\big\rangle,
\end{align}
and postpone its estimation in the next subsection.  In the next proposition, we estimate those remaining terms:
\begin{prop}\label{prop:esti-top1}
For each $0\leq s_{1},s_{2}\leq 1$ and $2\leq j_{1}\leq 3$, it holds that
\begin{align}
	\label{esti-top1-1}&\sum_{k,l\in\Z_{*}^{3}}\sum_{p\in S_{k}\cap S_{l}}\Big|\big\langle\Phi_{s_{1}},\tilde{a}_{p_{2}}^{*}b_{l}^{*}(K_{l}e_{p_{4}})b_{k}(T_{k,s_{2}}^{(j_{1},0)}e_{p_{1}})\tilde{a}_{p_{3}}\Phi_{s_{1}}\big\rangle\Big|,\\
	\label{esti-top1-2}&\sum_{k,l\in\Z_{*}^{3}}\sum_{p\in S_{k}\cap S_{l}}\Big|\big\langle\Phi_{s_{1}},\tilde{a}_{p_{2}}^{*}[b_{k}(T_{k,s_{2}}^{(j_{1},0)}e_{p_{1}}),b_{l}^{*}(K_{l}e_{p_{4}})]\tilde{a}_{p_{3}}\Phi_{s_{1}}\big\rangle\Big|,
\end{align}
are all bounded by
\begin{align*}
	C_{V}k_{F}^{-1}\big(\cQ m(\xi)+m(\xi)^{1/2}\cdot\sup_{0\leq \tau\leq 1}\|\tilde{a}_{\xi}\Phi_{\tau}\|\big).
\end{align*}  
The same bound also holds for term in \eqref{esti-top1-2} with $j_{1}=1$.
\end{prop}
\begin{proof}
For \eqref{esti-top1-1}, by \eqref{excitation-bdd}--\eqref{excitation-bdd5}, Cauchy-Schwarz inequality, pull-through formula \eqref{pull-through} and Lemma \ref{lem:Gronwall}, we compute for $S_{k}=L_{k}$ that
\begin{align}\label{esti-top1-1-1}
	&\sum_{k,l\in\Z_{*}^{3}}\sum_{p\in L_{k}\cap L_{l}}\Big|\big\langle\Phi_{s_{1}},\tilde{a}_{p_{2}}^{*}b_{l}^{*}(K_{l}e_{p_{4}})b_{k}(T_{k,s_{2}}^{(j_{1},0)}e_{p_{1}})\tilde{a}_{p_{3}}\Phi_{s_{1}}\big\rangle\Big|\nonumber\\
	&\leq\sum_{(k,\zeta)\in\cC_{\xi}}\sum_{l\in\Z_{*}^{3}}\sum_{p\in L_{k}\cap L_{l}}\delta_{p,\zeta}\|A_{k,s_{2}}^{(j_{1})}e_{\zeta}\|\|\cN_{l}^{1/2}(\cN_{E}+1)\tilde{a}_{-p+l}\Phi_{s_{1}}\|\nonumber\\
	&\quad\quad\quad\quad\quad\quad\quad\quad\quad\quad\quad\quad\quad\quad\quad\quad\quad\times\|K_{l}e_{p}\|\|(\cN_{E}+1)^{-1}\cN_{k}^{1/2}\tilde{a}_{-p+k}\Phi_{s_{1}}\|\nonumber\\
	\displaybreak
	&\leq k_{F}^{-1}\sum_{k,l\in\Z_{*}^{3}}\sum_{\zeta\in\cD_{k,\xi}}\frac{\hat{V}_{k}^{j_{1}}\hat{V}_{l}}{\sqrt{\lambda_{k,\zeta}\lambda_{l,\zeta}}}\|\cN_{l}^{1/2}\cN_{E}\Phi_{s_{1}}\|\|\tilde{a}_{-\zeta+k}\Phi_{s_{1}}\|\nonumber\\
	&\leq k_{F}^{-1}\sum_{k,l\in\Z_{*}^{3}}\hat{V}_{k}^{j_{1}}\hat{V}_{l}\|\cN_{l}^{1/2}\cN_{E}\Phi_{s_{1}}\|\Big(\chi_{L_{k}\cap L_{l}}(\xi)m(\xi)\|\tilde{a}_{-\xi+k}\Phi_{s_{1}}\|\nonumber\\
	&\quad\quad\quad\quad\quad\quad\quad\quad\quad\quad\quad\quad\quad\quad\quad+\chi_{L_{k}\cap L_{l}}(\xi+k)m(\xi)^{1/2}\|\tilde{a}_{-\xi}\Phi_{s_{1}}\|\Big)\nonumber\\
	&\leq k_{F}^{-1}\sum_{k\in\Z_{*}^{3}}\hat{V}_{k}^{j_{1}}\sqrt{\sum_{l\in\Z_{*}^{3}}\hat{V}_{l}^{2}}\sqrt{\sum_{l\in\Z_{*}^{3}}\|\cN_{l}^{1/2}\cN_{E}\Phi_{s_{1}}\|^{2}}\Big(\chi_{L_{k}}(\xi)m(\xi)\cQ+\chi_{L_{k}'}(\xi)m(\xi)^{1/2}\|\tilde{a}_{\xi}\Phi_{s_{1}}\|\Big)\nonumber\\
	&\leq C_{V}k_{F}^{-1}\Big(\chi_{B_{F}^{c}}(\xi)m(\xi)\cQ+\chi_{B_{F}}(\xi)m(\xi)^{1/2}\|\tilde{a}_{\xi}\Phi_{s_{1}}\|\Big)
\end{align}
The case $S_{k}=L_{k}'$ can be handled by the same method and is bounded by
\begin{align*}
	C_{V}k_{F}^{-1}\big(\chi_{B_{F}}(\xi)m(\xi)\cQ+\chi_{B_{F}^{c}}(\xi)m(\xi)^{1/2}\|\tilde{a}_{\xi}\Phi_{s_{1}}\|\big).
\end{align*}
In summary, we obtain
\begin{align}
	&\sum_{k,l\in\Z_{*}^{3}}\sum_{p\in S_{k}\cap S_{l}}\Big|\big\langle\Phi_{s_{1}},\tilde{a}_{p_{2}}^{*}b_{l}^{*}(K_{l}e_{p_{4}})b_{k}(T_{k,s_{2}}^{(j_{1},0)}e_{p_{1}})\tilde{a}_{p_{3}}\Phi_{s_{1}}\big\rangle\Big|\nonumber\\
	&\leq C_{V}k_{F}^{-1}\big(m(\xi)\cQ+m(\xi)^{1/2}\|\tilde{a}_{\xi}\Phi_{s_{1}}\|\big).
\end{align}

Next, for \eqref{esti-top1-2}, by \eqref{top-com}, we first use Proposition \ref{prop:prod-Kk-esti} to obtain
\begin{align*}
	&\sum_{k\in\Z_{*}^{3}}\sum_{p\in S_{k}}\big|\langle T_{k,s_{2}}^{(j_{1},0)}e_{p_{1}},K_{k}e_{p_{1}}\rangle\big|\langle \Phi_{s_{1}},\tilde{a}_{p_{3}}^{*}\tilde{a}_{p_{3}}\Phi_{s_{1}}\rangle\leq C_{V}\cQ^{2}\sum_{(k,\zeta)\in\cC_{\xi}}\langle A_{k,s_{2}}^{(j_{1})}e_{\zeta},K_{k}e_{\zeta}\rangle\nonumber\\
	&\leq C_{V}k_{F}^{-1}\cQ^{2}\sum_{(k,\zeta)\in\cC_{\xi}}\hat{V}_{k}^{m+1}\frac{\chi_{L_{k}}(\zeta)}{\lambda_{k,\zeta}}\leq C_{V}k_{F}^{-1}\cQ^{2}m(\xi).
\end{align*}
Finally, using a similar argument, we obtain
\begin{align}\label{esti-top2-1}
	&\sum_{k,l\in\Z_{*}^{3}}\sum_{p\in S_{k}\cap S_{l}}\sum_{q\in S_{k}'\cap S_{l}'}\big|\langle T_{k,s_{2}}^{(j_{1},0)}e_{p_{1}},e_{q_{1}}\rangle\big|\big|\langle e_{q_{4}},K_{l}e_{p_{4}}\rangle\big|\big|\langle \Phi_{s_{1}},\tilde{a}_{p_{2}}^{*}\tilde{a}_{q_{2}}^{*}\tilde{a}_{q_{3}}\tilde{a}_{p_{3}}\Phi_{s_{1}}\rangle\big|\nonumber\\
	&\leq\sum_{(k,\zeta)\in\cC_{\xi}}\sum_{l\in\Z_{*}^{3}}\sum_{p\in S_{k}\cap S_{l}}\sum_{q\in S_{k}'\cap S_{l}'}\delta_{p_{1},\zeta}\big|\langle A_{k,s_{2}}^{(j_{1})}e_{\zeta},e_{q_{1}}\rangle\big|\big|\langle e_{q_{4}},K_{l}e_{p_{4}}\rangle\big|\|\tilde{a}_{q_{2}}\tilde{a}_{p_{2}}\Phi_{s_{1}}\|\|\tilde{a}_{q_{3}}\tilde{a}_{p_{3}}\Phi_{s_{1}}\|\nonumber\\ 
	&\leq\sqrt{\sum_{(k,\zeta)\in\cC_{\xi}}\sum_{l\in\Z_{*}^{3}}\sum_{p\in S_{k}\cap S_{l}}\sum_{q\in S_{k}'\cap S_{l}'}\chi_{L_{k}}(\zeta)\delta_{p_{1},\zeta}\big|\langle A_{k,s_{2}}^{(j_{1})}e_{\zeta},e_{q_{1}}\rangle\big|^{2}\|\tilde{a}_{q_{2}}\tilde{a}_{p_{2}}\Phi_{s_{1}}\|^{2}}\nonumber\\
	&\times\sqrt{\sum_{(k,\zeta)\in\cC_{\xi}}\sum_{l\in\Z_{*}^{3}}\sum_{p\in S_{k}\cap S_{l}}\sum_{q\in S_{k}'\cap S_{l}'}\chi_{L_{k}}(\zeta)\delta_{p_{1},\zeta}\big|\langle e_{q_{4}},K_{l}e_{p_{4}}\rangle\big|^{2}\|\tilde{a}_{q_{3}}\tilde{a}_{p_{3}}\Phi_{s_{1}}\|^{2}}.
\end{align}
We estimate these factors separately.  For the first one, by \eqref{excitation-bdd}--\eqref{excitation-bdd5} and Proposition \ref{prop:prod-Kk-esti}, we obtain for $S_{k}=L_{k}$ that
\begin{align*}
	&\sum_{(k,\zeta)\in\cC_{\xi}}\sum_{l\in\Z_{*}^{3}}\sum_{p\in S_{k}\cap S_{l}}\sum_{q\in S_{k}'\cap S_{l}'}\chi_{L_{k}}(\zeta)\delta_{p_{1},\zeta}\big|\langle A_{k,s_{2}}^{(j_{1})}e_{\zeta},e_{q_{1}}\rangle\big|^{2}\|\tilde{a}_{q_{2}}\tilde{a}_{p_{2}}\Phi_{s_{1}}\|^{2}\nonumber\\
	\displaybreak
	&\leq C_{V}k_{F}^{-2}\sum_{(k,\zeta)\in\cC_{\xi}}\sum_{l\in\Z_{*}^{3}}\sum_{q\in S_{k}'\cap S_{l}'}\frac{\hat{V}_{k}^{2j_{1}}\chi_{L_{k}\cap L_{l}}(\zeta)}{(\lambda_{k,\zeta}+\lambda_{k,q_{1}})^{2}}\|\tilde{a}_{q_{3}}\tilde{a}_{l-\zeta}\Phi_{s_{1}}\|^{2}\nonumber\\
	&\leq C_{V}k_{F}^{-2}\sum_{(k,\zeta)\in\cC_{\xi}}\sum_{l\in\Z_{*}^{3}}\frac{\hat{V}_{k}^{2j_{1}}\chi_{L_{k}\cap L_{l}}(\zeta)}{\lambda_{k,\zeta}^{2}}\|\tilde{a}_{l-\zeta}\cN_{E}^{1/2}\Phi_{s_{1}}\|^{2}\nonumber\\
	&\leq C_{V}k_{F}^{-2}m(\xi)^{2}\sum_{k\in\Z_{*}^{3}}\hat{V}_{k}^{2j_{1}}\sum_{l\in\Z_{*}^{3}}\chi_{L_{l}}(\zeta)\|\tilde{a}_{l-\zeta}\cN_{E}^{1/2}\Phi_{s_{1}}\|^{2}\leq C_{V}k_{F}^{-2}m(\xi)^{2}.
\end{align*}
For the second factor in \eqref{esti-top2-1}, we proceed similarly to obtain for $S_{k}=L_{k}$ that
\begin{align*}
	&\sum_{(k,\zeta)\in\cC_{\xi}}\sum_{l\in\Z_{*}^{3}}\sum_{p\in S_{k}\cap S_{l}}\sum_{q\in S_{k}'\cap S_{l}'}\delta_{p_{1},\zeta}\big|\langle e_{q_{4}},K_{l}e_{p_{4}}\rangle\big|^{2}\|\tilde{a}_{q_{3}}\tilde{a}_{p_{3}}\Phi_{s_{1}}\|^{2}\nonumber\\
	&=\sum_{k,l\in\Z_{*}^{3}}\sum_{q\in S_{k}'\cap S_{l}'}\chi_{L_{k}\cap L_{l}}(\xi)\big|\langle e_{q_{4}},K_{l}e_{\xi}\rangle\big|^{2}\|\tilde{a}_{q_{3}}\tilde{a}_{k-\xi}\Phi_{s_{1}}\|^{2}\nonumber\\
	&\quad\quad\quad\quad\quad+\sum_{k,l\in\Z_{*}^{3}}\sum_{q\in S_{k}'\cap S_{l}'}\chi_{L_{k}\cap L_{l}}(k+\xi)\big|\langle e_{q_{4}},K_{l}e_{k+l+\xi}\rangle\big|^{2}\|\tilde{a}_{q_{3}}\tilde{a}_{\xi}\Phi_{s_{1}}\|^{2}\nonumber\\
	&\leq k_{F}^{-2}\sum_{l\in\Z_{*}^{3}}\sum_{q\in S_{l}'}\frac{\hat{V}_{l}^{2}}{\lambda_{l,q_{4}}}\Big(\sum_{k\in\Z_{*}^{3}}\chi_{S_{k}'}(q)\|\tilde{a}_{q_{3}}\Phi_{s_{1}}\|^{2}+\sum_{k\in\Z_{*}^{3}}\chi_{L_{k}'}(\xi)\chi_{S_{k}'}(q)\|\tilde{a}_{q_{3}}\tilde{a}_{\xi}\Phi_{s_{1}}\|^{2}\Big)\nonumber\\
	&\leq k_{F}^{-2}\Big(\|\cN_{E}^{1/2}\Phi_{s_{1}}\|^{2}+\|\tilde{a}_{\xi}\cN_{E}^{1/2}\Phi_{s_{1}}\|^{2}\Big)\sum_{l\in\Z_{*}^{3}}\hat{V}_{l}^{2}\sum_{q\in S_{l}'}\frac{1}{\lambda_{l,q_{4}}}\leq C_{V}k_{F}^{-1}.
\end{align*}
Substituting them in \eqref{esti-top2-1} gives
\begin{align}\label{esti-top2-2}
	&\sum_{k,l\in\Z_{*}^{3}}\sum_{p\in S_{k}\cap S_{l}}\sum_{q\in S_{k}'\cap S_{l}'}\big|\langle T_{k,s_{2}}^{(j_{1},0)}e_{p_{1}},e_{q_{1}}\rangle\big|\big|\langle e_{q_{4}},K_{l}e_{p_{4}}\rangle\big|\big|\langle \Phi_{s_{1}},\tilde{a}_{p_{2}}^{*}\tilde{a}_{q_{2}}^{*}\tilde{a}_{q_{3}}\tilde{a}_{p_{3}}\Phi_{s_{1}}\rangle\big|\nonumber\\
	&\leq C_{\delta,V}k_{F}^{-3/2}m(\xi).
\end{align}
We obtain the same estimate for $S_{k}=L_{k}'$ and omit it for simplicity.  We observe that, for commutator term, the above estimate extends to the case with $j_{1}=1$.
\end{proof}

\begin{prop}\label{prop:esti-top3}
For each $0\leq s_{1},s_{2}\leq 1$ and $1\leq j_{2}\leq 3$, it holds for each $\delta>0$ that
\begin{align}
	\label{esti-top3-1}&\sum_{k,l\in\Z_{*}^{3}}\sum_{p\in S_{k}\cap S_{l}}\Big|\big\langle\Phi_{s_{1}},\tilde{a}_{p_{2}}^{*}b_{l}^{*}(K_{l}e_{p_{4}})b_{k}(T_{k,s_{2}}^{(0,j_{2})}e_{p_{1}})\tilde{a}_{p_{3}}\Phi_{s_{1}}\big\rangle\Big|,\\
	\label{esti-top3-2}&\sum_{k,l\in\Z_{*}^{3}}\sum_{p\in S_{k}\cap S_{l}}\Big|\big\langle\Phi_{s_{1}},\tilde{a}_{p_{2}}^{*}[b_{k}(T_{k,s_{2}}^{(0,j_{2})}e_{p_{1}}),b_{l}^{*}(K_{l}e_{p_{4}})]\tilde{a}_{p_{3}}\Phi_{s_{1}}\big\rangle\Big|,
\end{align}
are all bounded by $C_{\delta,V}(k_{F}^{-1}\cQ^{1-\delta}+k_{F}^{-3/2})m(\xi)$ for any $\delta>0$.  
\end{prop}
\begin{proof}
For \eqref{esti-top3-1}, by \eqref{excitation-bdd}--\eqref{excitation-bdd5}, Cauchy-Schwarz inequality, Lemmas \ref{lem:Gronwall}, \ref{lem:excitation-Gronwall}, Proposition \ref{prop:prod-Kk-esti} and the lower bound \eqref{gap}, we have
\begin{align}\label{esti-top3-1-1}
	&\sum_{k,l\in\Z_{*}^{3}}\sum_{p\in S_{k}\cap S_{l}}\Big|\big\langle\Phi_{s_{1}},\tilde{a}_{p_{2}}^{*}b_{l}^{*}(K_{l}e_{p_{4}})b_{k}(T_{k,s_{2}}^{(0,j_{2})}e_{p_{1}})\tilde{a}_{p_{3}}\Phi_{s_{1}}\big\rangle\Big|\nonumber\\
	\displaybreak
	&\leq \sum_{(k,\zeta)\in\cC_{\xi}}\sum_{l\in\Z_{*}^{3}}\sum_{p\in S_{k}\cap S_{l}}\big|\langle e_{\zeta}, A_{k,s_{2}}^{(j_{2})}e_{p_{1}}\rangle\big|\|K_{l}e_{p_{4}}\|\|\cN_{l}^{1/2}\tilde{a}_{p_{2}}\Phi_{s_{1}}\|\|\tilde{a}_{\zeta}\tilde{a}_{\zeta-k}\tilde{a}_{p_{3}}\Phi_{s_{1}}\|\nonumber\\
	&\leq k_{F}^{-1}\sum_{(k,\zeta)\in\cC_{\xi}}\frac{\hat{V}_{k}^{j_{2}}}{\lambda_{k,\zeta}}\sqrt{\sum_{l\in\Z_{*}^{3}}\sum_{p\in S_{l}}\|K_{l}e_{p_{4}}\|^{2}}\sqrt{\sum_{l\in\Z_{*}^{3}}\|\cN_{l}^{1/2}\Phi_{s_{1}}\|^{2}}\sqrt{\sum_{p\in S_{k}}\|\tilde{a}_{p_{3}}\tilde{a}_{\zeta}\tilde{a}_{\zeta-k}\Phi_{s_{1}}\|^{2}}\nonumber\\
	&\leq C_{V}k_{F}^{-1}m(\xi)\|\cN_{E}\Phi_{s_{1}}\|\sqrt{\sum_{k\in\Z_{*}^{3}}\hat{V}_{k}^{2j_{2}}}\sqrt{\sum_{(k,\zeta)\in\cC_{\xi}}\|\cN_{E}^{1/2}\tilde{a}_{\zeta}\tilde{a}_{\zeta-k}\Phi_{s_{1}}\|^{2}}.
\end{align}
For the last factor in \eqref{esti-top3-1-1}, we use \eqref{excitation-bdd}--\eqref{excitation-bdd5} to obtain for each $\delta>0$ that
\begin{align}
	\sum_{(k,\zeta)\in\cC_{\xi}}\|\cN_{E}^{1/2}\tilde{a}_{\zeta}\tilde{a}_{\zeta-k}\Phi_{s_{1}}\|^{2}&=\sum_{k\in\Z_{*}^{3}}\chi_{L_{k}}(\xi)\|\tilde{a}_{\xi-k}\cN_{E}^{1/2}\tilde{a}_{\xi}\Phi_{s_{1}}\|^{2}\nonumber\\
	&\quad\quad\quad\quad\quad\quad+\sum_{k\in\Z_{*}^{3}}\chi_{L_{k}'}(\xi)\|\tilde{a}_{k+\xi}\cN_{E}^{1/2}\tilde{a}_{\xi}\Phi_{s_{1}}\|^{2}\nonumber\\
	&\leq \|\cN_{E}\tilde{a}_{\xi}\Phi_{s_{1}}\|^{2}\leq \|\tilde{a}_{\xi}\cN_{E}\Phi_{s_{1}}\|^{2}\leq C_{\delta,V}\cQ^{2-\delta}.
\end{align}
Consequently, \eqref{esti-top3-1-1} becomes
\begin{align}
	&\sum_{k,l\in\Z_{*}^{3}}\sum_{p\in S_{k}\cap S_{l}}\Big|\big\langle\Phi_{s_{1}},\tilde{a}_{p_{2}}^{*}b_{l}^{*}(K_{l}e_{p_{4}})b_{k}(T_{k,s_{2}}^{(0,j_{2})}e_{p_{1}})\tilde{a}_{p_{3}}\Phi_{s_{1}}\big\rangle\Big|\leq C_{\delta,V}k_{F}^{-1}\cQ^{1-\delta}m(\xi).
\end{align}

For \eqref{esti-top3-2}, by \eqref{top-com}, we have 
\begin{align}\label{esti-top3-4}
	&\sum_{k\in\Z_{*}^{3}}\sum_{p\in S_{k}}\big|\langle T_{k,s_{2}}^{(0,j_{2})}e_{p_{1}},K_{k}e_{p_{1}}\rangle\big|\big|\langle \Phi_{s_{1}},\tilde{a}_{p_{3}}^{*}\tilde{a}_{p_{3}}\Phi_{s_{1}}\rangle\big|\nonumber\\
	&\leq\sum_{(k,\zeta)\in\cC_{\xi}}\sum_{p\in S_{k}}\chi_{L_{k}}(\zeta)\big|\langle e_{\zeta},A_{k,s_{2}}^{(j_{2})}e_{p_{1}}\rangle\big|\big|\langle e_{\zeta},K_{k}e_{p_{1}}\rangle\big|\|\tilde{a}_{p_{3}}\Phi_{s_{1}}\|^{2}\nonumber\\
	&\leq C k_{F}^{-2}\cQ^{2}\sum_{(k,\zeta)\in\cC_{\xi}}\sum_{p\in S_{k}}\hat{V}_{k}^{j_{2}+1}\frac{\chi_{L_{k}}(\zeta)}{\lambda_{k,p_{1}}\lambda_{k,\zeta}}\leq C_{V}k_{F}^{-1}\cQ^{2}m(\xi),
\end{align}
and
\begin{align}\label{esti-top3-3}
	&\sum_{k,l\in\Z_{*}^{3}}\sum_{p\in S_{k}\cap S_{l}}\sum_{q\in S_{k}'\cap S_{l}'}\big|\langle T_{k,s_{2}}^{(0,j_{2})}e_{p_{1}},e_{q_{1}}\rangle\big|\big|\langle e_{q_{4}},K_{l}e_{p_{4}}\rangle\big|\big|\big\langle\Phi_{s_{1}},\tilde{a}_{p_{2}}^{*}\tilde{a}_{q_{2}}^{*}\tilde{a}_{q_{3}}\tilde{a}_{p_{3}}\Phi_{s_{1}}\big\rangle\big|\nonumber\\
	&\leq\sum_{(k,\zeta)\in\cC_{\xi}}\sum_{l\in\Z_{*}^{3}}\sum_{p\in S_{k}\cap S_{l}}\sum_{q\in S_{k}'\cap S_{l}'}\chi_{L_{k}}(\zeta)\delta_{q_{1},\zeta}\big|\langle A_{k,s_{2}}^{(j_{1})}e_{p_{1}},e_{\zeta}\rangle\big|\nonumber\\
	&\quad\quad\quad\quad\quad\quad\quad\quad\quad\quad\quad\times\big|\langle e_{q_{4}},K_{l}e_{p_{4}}\rangle\big|\big|\big\langle\Phi_{s_{1}},\tilde{a}_{p_{2}}^{*}\tilde{a}_{2}^{*}\tilde{a}_{q_{3}}\tilde{a}_{p_{3}}\Phi_{s_{1}}\big\rangle\big|.
\end{align}
We see that, due to the symmetry of variables $p$ and $q$, this estimate is exactly the same as for \eqref{esti-top2-1} by switching the roles of $p$ and $q$, which we know from \eqref{esti-top2-2} that \eqref{esti-top3-3} is bounded by $C_{\delta,V}k_{F}^{-1}\cQ^{1-\delta}m(\xi)$. 
\end{proof}

\smallskip

\noindent\textbf{Estimation of the Single Commutator Terms}.  Next, we estimate the singule commutator terms in \eqref{single-com}:
\begin{align*}
	\sum_{k,l\in\Z_{*}^{3}}\sum_{p\in S_{k}\cap S_{l}}\tilde{a}_{p_{2}}^{*}[b_{l}(K_{l}e_{p_{4}}),\tilde{a}_{p_{3}}^{*}]^{*}b_{k}(T_{k}e_{p_{1}}),\quad \sum_{k,l\in\Z_{*}^{3}}\sum_{p\in S_{k}\cap S_{l}}b_{l}^{*}(K_{l}e_{p_{4}})[b_{k}(T_{k}e_{p_{1}}),\tilde{a}_{p_{2}}^{*}]\tilde{a}_{p_{3}}.
\end{align*}
Note that
\begin{align}\label{E2-com-single-esti}
	[b_{l}(K_{l}e_{p_{4}}),\tilde{a}_{p_{3}}^{*}]&=\begin{cases}
		-\chi_{L_{l}}(p_{3}+l)\langle K_{l}e_{p_{4}},e_{p_{3}+l}\rangle\tilde{a}_{p_{3}+l}\quad&S_{k}=L_{k},\\
		\chi_{L_{l}}(p_{3})\langle K_{l}e_{p_{4}},e_{p_{3}}\rangle \tilde{a}_{p_{3}-l}&S_{k}=L_{k}',
	\end{cases},
\end{align}
which satisfies
\begin{align}\label{E2-com-single-esti1}
	\begin{cases}
		|\chi_{L_{l}'}(p_{3})\langle K_{l}e_{p_{4}},e_{p_{3}+l}\rangle|\quad&S_{k}=L_{k}\\
		|\chi_{L_{l}}(p_{3})\langle K_{l}e_{p_{4}},e_{p_{3}}\rangle|&S_{k}=L_{k}'\end{cases}
	&\leq C \frac{k_{F}^{-1}\hat{V}_{l}}{\sqrt{\lambda_{k,p_{1}}\lambda_{l,p_{4}}}}.
\end{align}
\begin{prop}\label{prop:esti-single-com-1}
For each $0\leq s_{1},s_{2}\leq 1$ and $1\leq j_{1}\leq 3$, it holds for each $\delta>0$ that
\begin{align}
	\label{esti-single-com-1}&\sum_{k,l\in\Z_{*}^{3}}\sum_{p\in S_{k}\cap S_{l}}\big|\big\langle\Phi_{s_{1}},\tilde{a}_{p_{2}}^{*}[b_{l}(K_{l}e_{p_{4}}),\tilde{a}_{p_{3}}^{*}]^{*}b_{k}(T_{k,s_{2}}^{(j_{1},0)}e_{p_{1}})\Phi_{s_{1}}\big\rangle\big|,\\
	\label{esti-single-com-2}&\sum_{k,l\in\Z_{*}^{3}}\sum_{p\in S_{k}\cap S_{l}}\big|\big\langle\Phi_{s_{1}},b_{l}^{*}(K_{l}e_{p_{4}})[b_{k}(T_{k,s_{2}}^{(j_{1},0)}e_{p_{1}}),\tilde{a}_{p_{2}}^{*}]\tilde{a}_{p_{3}}\Phi_{s_{1}}\big\rangle\big|,
\end{align}
are bounded by $C_{\delta,V}\big(k_{F}^{-3/2}+k_{F}^{-1+\delta}\cQ\big)m(\xi)$.  
\end{prop}
\begin{proof}
For \eqref{esti-single-com-1}, by \eqref{excitation-bdd}--\eqref{excitation-bdd5}, \eqref{E2-com-single-esti}--\eqref{E2-com-single-esti1}, Cauchy-Schwarz inequality, Lemmas \ref{lem:Gronwall}, \ref{lem:excitation-Gronwall} and Proposition \ref{prop:prod-Kk-esti}, if $S_{k}=L_{k}$, we have
\begin{align}\label{esti-single-com-1-1}
	&\sum_{k,l\in\Z_{*}^{3}}\sum_{p\in S_{k}\cap S_{l}}\big|\big\langle\Phi_{s_{1}},\tilde{a}_{p_{2}}^{*}[b_{l}(K_{l}e_{p_{4}}),\tilde{a}_{p_{3}}^{*}]^{*}b_{k}(T_{k,s_{2}}^{(j_{1},0)}e_{p_{1}})\Phi_{s_{1}}\big\rangle\big|\nonumber\\
	&\leq C k_{F}^{-3/2}\sum_{(k,\zeta)\in\cC_{\xi}}\sum_{l\in\Z_{*}^{3}}\chi_{L_{l}}(\zeta)\frac{\hat{V}_{k}^{j_{1}}\hat{V}_{l}}{\lambda_{k,\zeta}\sqrt{\lambda_{l,\zeta}}}\|\tilde{a}_{k+l-\zeta}\tilde{a}_{l-\zeta}\Phi_{s_{1}}\|\|\cN_{k}^{1/2}\Phi_{s_{1}}\|\nonumber\\
	&\leq Ck_{F}^{-3/2}m(\xi)\sum_{k,l\in\Z_{*}^{3}}\chi_{L_{k}\cap L_{l}}(\xi)\hat{V}_{k}^{j_{1}}\hat{V}_{l}\|\tilde{a}_{k+l-\xi}\tilde{a}_{l-\xi}\Phi_{s_{1}}\|\|\cN_{k}^{1/2}\Phi_{s_{1}}\|\nonumber\\
	&\quad\quad+Ck_{F}^{-3/2}m(\xi)\sum_{k,l\in\Z_{*}^{3}}\chi_{L_{k}\cap L_{l}}(k+\xi)\hat{V}_{k}^{j_{1}}\hat{V}_{l}\|\tilde{a}_{l-\xi}\tilde{a}_{-k-\xi+l}\Phi_{s_{1}}\|\|\cN_{k}^{1/2}\Phi_{s_{1}}\|\nonumber\\
	&\leq Ck_{F}^{-3/2}m(\xi)\sum_{k\in\Z_{*}^{3}}\hat{V}_{k}^{j_{1}}\|\cN_{k}^{1/2}\Phi_{s_{1}}\|\sqrt{\sum_{l\in\Z_{*}^{3}}\hat{V}_{l}^{2}}\sqrt{\sum_{l\in\Z_{*}^{3}}\chi_{L_{l}}(\xi)\|\tilde{a}_{l-\xi}\Phi_{s_{1}}\|^{2}}\nonumber\\
	&\quad\quad+Ck_{F}^{-3/2}m(\xi)\sum_{k\in\Z_{*}^{3}}\hat{V}_{k}^{j_{1}}\|\cN_{k}^{1/2}\Phi_{s_{1}}\|\sqrt{\sum_{l\in\Z_{*}^{3}}\hat{V}_{l}^{2}}\sqrt{\sum_{l\in\Z_{*}^{3}}\chi_{L_{l}}(k+\xi)\|\tilde{a}_{-k-\xi+l}\Phi_{s_{1}}\|^{2}}\nonumber\\
	&\leq C_{V}k_{F}^{-3/2}m(\xi)\|\cN_{E}^{1/2}\Phi_{s_{1}}\|\sqrt{\sum_{k\in\Z_{*}^{3}}\|\cN_{k}^{1/2}\Phi_{s_{1}}\|^{2}}\sqrt{\sum_{k\in\Z_{*}^{3}}\hat{V}_{k}^{2j_{1}}}\leq C_{V}k_{F}^{-3/2}m(\xi).
\end{align}  

Next, for \eqref{esti-single-com-2}, we first compute the commutator using \eqref{ba-comm}:
\begin{align}\label{single-T1-com}
	[b_{k}(T_{k,s_{2}}^{(j_{1},0)}e_{p_{1}}),\tilde{a}_{p_{2}}^{*}]&=\begin{cases}
		-\chi_{L_{k}'}(p_{2})\langle T_{k,s_{2}}^{(j_{1},0)}e_{p_{1}},e_{p_{2}+k}\rangle\tilde{a}_{p_{2}+k}\quad&S_{k}=L_{k},\\
		\chi_{L_{k}}(p_{2})\langle T_{k,s_{2}}^{(j_{1},0)}e_{p_{1}},e_{p_{2}}\rangle\tilde{a}_{p_{2}-k}&S_{k}=L_{k}'.
	\end{cases}
\end{align}
Again, we will only consider \eqref{esti-single-com-2} for $S_{k}=L_{k}$ since the case $S_{k}=L_{k}'$ can be done in the same way.  Together with the lower bound \eqref{gap} and Lemma \ref{lem:gap-sum}, this gives for each $\delta>0$ that 
\begin{align}\label{single-T1-com2}
	&\sum_{k,l\in\Z_{*}^{3}}\sum_{p\in S_{k}\cap S_{l}}\big|\big\langle\Phi_{s_{1}},b_{l}^{*}(K_{l}e_{p_{4}})[b_{k}(T_{k,s_{2}}^{(j_{1},0)}e_{p_{1}}),\tilde{a}_{p_{2}}^{*}]\tilde{a}_{p_{3}}\Phi_{s_{1}}\big\rangle\big|\nonumber\\
	&=\sum_{k,l\in\Z_{*}^{3}}\sum_{p\in L_{k}\cap L_{l}}\chi_{L_{k}'}(-p+l)\big|\langle T_{k,s_{2}}^{(j_{1},0)}e_{p},e_{-p+l+k}\rangle\big|\big|\|b_{l}(K_{l}e_{p})\Phi_{s_{1}}\|\|\tilde{a}_{-p+l+k}\tilde{a}_{-p+k}\Phi_{s_{1}}\|\nonumber\\
	&\leq C\sum_{k,l\in\Z_{*}^{3}}\chi_{L_{k}\cap L_{l}}(\zeta)\chi_{L_{k}'}(l-\zeta)\big|\langle A_{k,s_{2}}^{(j_{1})}e_{\zeta},e_{-\zeta+l+k}\rangle\big|\|K_{l}e_{\zeta}\|\|\cN_{l}^{1/2}\Phi_{s_{1}}\|\|\tilde{a}_{-\zeta+k}\Phi_{s_{1}}\|\nonumber\\
	&\leq C k_{F}^{-3/2}\sum_{(k,\zeta)\in\cC_{\xi}}\sum_{l\in\Z_{*}^{3}}\chi_{L_{k}\cap L_{l}}(\zeta)\chi_{L_{k}'}(l-\zeta)\frac{\hat{V}_{k}^{j_{1}}\hat{V}_{l}}{(\lambda_{k,\zeta}+\lambda_{k,-\zeta+l+k})\sqrt{\lambda_{l,\zeta}}}\|\tilde{a}_{-\zeta+k}\Phi_{s_{1}}\|\nonumber\\
	&\leq Ck_{F}^{-3/2}\Big[\sum_{k,l\in\Z_{*}^{3}}\chi_{L_{k}\cap L_{l}}(\xi)\chi_{L_{k}'}(l-\xi)\frac{\hat{V}_{k}^{j_{1}}\hat{V}_{l}}{(\lambda_{k,\xi}+\lambda_{k,-\xi+l+k})\sqrt{\lambda_{l,\xi}}}\|\cN_{l}^{1/2}\Phi_{s_{1}}\|\|\tilde{a}_{-\xi+k}\Phi_{s_{1}}\|\nonumber\\
	&\quad+\sum_{k,l\in\Z_{*}^{3}}\chi_{L_{k}\cap L_{l}}(k+\xi)\chi_{L_{k}}(l-\xi)\frac{\hat{V}_{k}^{j_{1}}\hat{V}_{l}}{(\lambda_{k,k+\xi}+\lambda_{k,-\xi+l})\sqrt{\lambda_{l,\xi+k}}}\|\cN_{l}^{1/2}\Phi_{s_{1}}\|\|\tilde{a}_{-\xi}\Phi_{s_{1}}\|\Big]\nonumber\\
	&\leq Ck_{F}^{-3/2}m(\xi)\sqrt{\sum_{k\in\Z_{*}^{3}}\hat{V}_{k}^{2j_{1}}}\sum_{l\in\Z_{*}^{3}}\hat{V}_{l}\|\cN_{l}^{1/2}\Phi_{s_{1}}\|\nonumber\\
	&\quad\quad\quad\quad\quad\quad\quad\times\Big(\sqrt{\sum_{k\in\Z_{*}^{3}}\chi_{L_{k}}(\xi)\|\tilde{a}_{-\xi+k}\Phi_{s_{1}}\|^{2}}+\cQ	\sqrt{\sum_{k\in\Z_{*}^{3}}\chi_{L_{k}}(l-\xi)m(l-\xi-k)}\Big)\nonumber\\
	&\leq C_{\delta,V}k_{F}^{-3/2}m(\xi)\Big(\|\cN_{E}^{1/2}\Phi_{s_{1}}\|+k_{F}^{1/2+\delta}\cQ\Big)\leq C_{\delta,V}\big(k_{F}^{-3/2}+k_{F}^{-1+\delta}\cQ\big)m(\xi).
\end{align}
This completes the proof.
\end{proof}

\begin{prop}\label{prop:esti-single-com-2}
For each $0\leq s_{1},s_{2}\leq 1$ and $1\leq j_{2}\leq 3$, it holds that
\begin{align}
	\label{esti-single-com-3}&\sum_{k,l\in\Z_{*}^{3}}\sum_{p\in S_{k}\cap S_{l}}\big|\big\langle\Phi_{s_{1}},\tilde{a}_{p_{2}}^{*}[b_{l}(K_{l}e_{p_{4}}),\tilde{a}_{p_{3}}^{*}]^{*}b_{k}(T_{k,s_{2}}^{(0,j_{2})}e_{p_{1}})\Phi_{s_{1}}\big\rangle\big|,\\
	\label{esti-single-com-4}&\sum_{k,l\in\Z_{*}^{3}}\sum_{p\in S_{k}\cap S_{l}}\big|\big\langle\Phi_{s_{1}},b_{l}^{*}(K_{l}e_{p_{4}})[b_{k}(T_{k,s_{2}}^{(0,j_{2})}e_{p_{1}}),\tilde{a}_{p_{2}}^{*}]\tilde{a}_{p_{3}}\Phi_{s_{1}}\big\rangle\big|,
\end{align}
are bounded by $C_{V}\big(k_{F}^{-1}\cQ+k_{F}^{-3/2}\big)m(\xi)$.  
\end{prop}
\begin{proof}
For \eqref{esti-single-com-3}, by \eqref{excitation-bdd}--\eqref{excitation-bdd5}, \eqref{E2-com-single-esti}--\eqref{E2-com-single-esti1}, Cauchy-Schwarz inequality, Lemmas \ref{lem:Gronwall}, \ref{lem:excitation-Gronwall}, Proposition \ref{prop:prod-Kk-esti} and operator norm $\|\tilde{a}_{q}\|=1$, we have
\begin{align*}
	&\sum_{k,l\in\Z_{*}^{3}}\sum_{p\in S_{k}\cap S_{l}}\big|\big\langle\Phi_{s_{1}},\tilde{a}_{p_{2}}^{*}[b_{l}(K_{l}e_{p_{4}}),\tilde{a}_{p_{3}}^{*}]^{*}b_{k}(T_{k,s_{2}}^{(0,j_{2})}e_{p_{1}})\Phi_{s_{1}}\big\rangle\big|\nonumber\\
	&\leq C k_{F}^{-1}\sum_{(k,\zeta)\in\cC_{\xi}}\sum_{l\in\Z_{*}^{3}}\sum_{p\in S_{k}\cap S_{l}}\frac{\hat{V}_{l}}{\sqrt{\lambda_{k,p_{1}}\lambda_{l,p_{4}}}}\big|\langle e_{\zeta},A_{k,s_{2}}^{(j_{1})}e_{p_{1}}\rangle\big|\|\tilde{a}_{p_{3}\pm l}\tilde{a}_{p_{2}}\Phi_{s_{1}}\|\|\tilde{a}_{\zeta}\tilde{a}_{\zeta-k}\Phi_{s_{1}}\|\nonumber\\
	\displaybreak
	&\leq C k_{F}^{-2}\sum_{(k,\zeta)\in\cC_{\xi}}\sum_{l\in\Z_{*}^{3}}\sum_{p\in S_{k}\cap S_{l}}\frac{\hat{V}_{l}}{\sqrt{\lambda_{k,p_{1}}\lambda_{l,p_{4}}}}\frac{\hat{V}_{k}^{j_{1}}}{\lambda_{k,\zeta}+\lambda_{k,p_{1}}}\|\tilde{a}_{p_{3}\pm l}\tilde{a}_{p_{2}}\Phi_{s_{1}}\|\|\tilde{a}_{\zeta}\Phi_{s_{1}}\|\nonumber\\
	&\leq C k_{F}^{-2}\cQ m(\xi)\sum_{l\in\Z_{*}^{3}}\sum_{p\in S_{l}}\frac{\hat{V}_{l}}{\sqrt{\lambda_{l,p_{4}}}}\sqrt{\sum_{k\in\Z_{*}^{3}}\frac{\hat{V}_{k}^{2j_{1}}\chi_{S_{k}}(p)}{\lambda_{k,p_{1}}}}\sqrt{\sum_{k\in\Z_{*}^{3}}\chi_{S_{k}\cap S_{l}}(p)\|\tilde{a}_{p_{3}\pm l}\tilde{a}_{p_{2}}\Phi_{s_{1}}\|^{2}}\nonumber\\
	&\leq C k_{F}^{-2}\cQ m(\xi)\sum_{p}\sqrt{\sum_{l\in\Z_{*}^{3}}\frac{\hat{V}_{l}^{2}\chi_{S_{l}}(p)}{\lambda_{l,p_{4}}}}\sqrt{\sum_{l\in\Z_{*}^{3}}\chi_{S_{l}}(p)\|\tilde{a}_{p_{2}}\cN_{E}^{1/2}\Phi_{s_{1}}\|^{2}}\sqrt{\sum_{k\in\Z_{*}^{3}}\frac{\hat{V}_{k}^{2j_{1}}\chi_{S_{k}}(p)}{\lambda_{k,p_{1}}}}\nonumber\\
	&\leq C k_{F}^{-2}\cQ m(\xi)\|\cN_{E}^{3/2}\Phi_{s_{1}}\|\sqrt{\sum_{l\in\Z_{*}^{3}}\hat{V}_{l}^{2}\sum_{p\in L_{l}}\frac{1}{\lambda_{l,p}}}\sqrt{\sum_{k\in\Z_{*}^{3}}\hat{V}_{k}^{2j_{1}}\sum_{p\in L_{k}}\frac{1}{\lambda_{k,p}}}\nonumber\\
	&\leq C_{V}k_{F}^{-1}\cQ m(\xi).
\end{align*}

Next, for \eqref{esti-single-com-4}, by the same set of arguments and also \eqref{single-T1-com} with $T_{k,s_{2}}^{(j_{1},0)}$ replaced by $T_{k,s_{2}}^{(0,j_{2})}$, we obtain for $S_{k}=L_{k}$ that 
\begin{align}
	&\sum_{k,l\in\Z_{*}^{3}}\sum_{p\in S_{k}\cap S_{l}}\big|\big\langle\Phi_{s_{1}},b_{l}^{*}(K_{l}e_{p_{4}})[b_{k}(T_{k,s_{2}}^{(0,j_{2})}e_{p_{1}}),\tilde{a}_{p_{2}}^{*}]\tilde{a}_{p_{3}}\Phi_{s_{1}}\big\rangle\big|\nonumber\\
	&\leq \sum_{(k,\zeta)\in\cC_{\xi}}\sum_{l\in\Z_{*}^{3}}\sum_{p\in L_{k}\cap L_{l}}\delta_{p_{2}+k,\zeta}\big|\langle e_{\zeta},A_{k,s_{2}}^{(j_{1})}e_{p_{1}}\rangle\big|\|K_{l}e_{p_{4}}\|\|\cN_{l}^{1/2}\Phi_{s_{1}}\|\|\tilde{a}_{p_{2}+k}\tilde{a}_{p_{3}}\Phi_{s_{1}}\|\nonumber\\
	&\leq\sum_{(k,\zeta)\in\cC_{\xi}}\sum_{l\in\Z_{*}^{3}}\chi_{L_{k}\cap L_{l}}(l-\zeta)\big|\langle e_{\zeta},A_{k,s_{2}}^{(j_{1})}e_{k+l-\zeta}\rangle\big|\|K_{l}e_{l-\zeta}\|\|\cN_{l}^{1/2}\Phi_{s_{1}}\|\|\tilde{a}_{k+l-\zeta}\tilde{a}_{-l+\zeta}\Phi_{s_{1}}\|\nonumber\\
	&\leq C k_{F}^{-3/2}\sum_{(k,\zeta)\in\cC_{\xi}}\sum_{l\in\Z_{*}^{3}}\chi_{L_{k}\cap L_{l}}(l-\zeta)\frac{\hat{V}_{k}^{j_{1}}\hat{V}_{l}}{\lambda_{k,\zeta}+\lambda_{k,l-\zeta}}\|\cN_{l}^{1/2}\Phi_{s_{1}}\|\|\tilde{a}_{k+l-\zeta}\tilde{a}_{-l+\zeta}\Phi_{s_{1}}\|\nonumber\\
	&\leq C k_{F}^{-3/2}m(\xi)\sum_{l\in\Z_{*}^{3}}\hat{V}_{l}\|\cN_{l}^{1/2}\Phi_{s_{1}}\|\sqrt{\sum_{k\in\Z_{*}^{3}}\hat{V}_{k}^{2j_{1}}}\nonumber\\
	&\quad\quad\quad\quad\times\sqrt{\sum_{k\in\Z_{*}^{3}}\chi_{L_{k}}(l-\xi)\|\tilde{a}_{k+l-\xi}\Phi_{s_{1}}\|^{2}+\sum_{k\in\Z_{*}^{3}}\chi_{L_{k}}(l-k-\xi)\|\tilde{a}_{-l+k+\xi}\Phi_{s_{1}}\|^{2}}\nonumber\\
	&\leq C_{V}k_{F}^{-3/2}m(\xi)\|\cN_{E}^{1/2}\Phi_{s_{1}}\|\sqrt{\sum_{l\in\Z_{*}^{3}}\hat{V}_{l}^{2}}\sqrt{\sum_{l\in\Z_{*}^{3}}\|\cN_{l}^{1/2}\Phi_{s_{1}}\|^{2}}\leq C_{V}k_{F}^{-3/2}m(\xi).
\end{align}
The case for $S_{k}=L_{k}'$ can be done similarly.  This completes the proof.
\end{proof}

\smallskip

\noindent\textbf{Estimation of the Double Commutator Terms}.  \DETAILS{Now, we come to the double commutator terms in \eqref{double-com1}--\eqref{double-com2}.  }Before proceeding, we note the following identity from \cite[Eq. (4.57)]{CHN-23} for double commutators: For each $\varphi\in\ell^{2}(L_{k}),\psi\in\ell^{2}(L_{l})$ and $p\in\Z_{*}^{3}$, it holds that 
\begin{equation}\label{double-com}
	-[b_{k}(\varphi),[b_{l}(\psi),\tilde{a}_{p}^{*}]^{*}]=\begin{cases}
		\chi_{L_{k}}(p+l)\chi_{L_{l}'}(p)\langle\varphi,e_{p+l}\rangle\langle e_{p+l},\psi\rangle\tilde{a}_{p+l-k}\quad&p\in B_{F},\\
		\chi_{L_{k}'}(p-l)\chi_{L_{l}}(p)\langle\varphi,e_{p-l+k}\rangle\langle e_{p},\psi\rangle\tilde{a}_{p-l+k}&p\in B_{F}^{c}.
	\end{cases}
\end{equation}
Then we have
\begin{prop}\label{prop:esti-double-com}
For each $0\leq s_{1},s_{2}\leq 1$ and $(j_{1},j_{2})\in\Sigma_{*}^{2}$, it holds for each $\delta>0$ that
\begin{align}
	\label{esti-double-com1}&\sum_{k,l\in\Z_{*}^{3}}\sum_{p\in S_{k}\cap S_{l}}\big|\big\langle\Phi_{s_{1}},\tilde{a}_{p_{2}}^{*}[b_{k}(T_{k,s_{2}}^{(j_{1},j_{2})}e_{p_{1}}),[b_{l}(K_{l}e_{p_{4}}),\tilde{a}_{p_{3}}^{*}]^{*}]\Phi_{s_{1}}\big\rangle\big|,\\
	\label{esti-double-com2}&\sum_{k,l\in\Z_{*}^{3}}\sum_{p\in S_{k}\cap S_{l}}\big|\big\langle\Phi_{s_{1}},[b_{l}(K_{l}e_{p_{4}}),[b_{k}(T_{k,s_{2}}^{(j_{1},j_{2})}e_{p_{1}}),\tilde{a}_{p_{2}}^{*}]^{*}]^{*}\tilde{a}_{p_{3}}\Phi_{s_{1}}\big\rangle\big|,\\
	\label{esti-double-com3}&\sum_{k,l\in\Z_{*}^{3}}\sum_{p\in S_{k}\cap S_{l}}\big|\big\langle\Phi_{s_{1}},[b_{l}(K_{l}e_{p_{4}}),\tilde{a}_{p_{3}}^{*}]^{*}[b_{k}(T_{k,s_{2}}^{(j_{1},j_{2})}e_{p_{1}}),\tilde{a}_{p_{2}}^{*}]\Phi_{s_{1}}\big\rangle\big|,
\end{align}
are bounded by $C_{\delta,V}k_{F}^{-3/2+\delta}\cQ m(\xi)$.  
\end{prop}
\begin{proof}
We will only consider the case $S_{k}=L_{k}$ since the case $S_{k}=L_{k}'$ is done by the same computations.  For \eqref{esti-double-com1} with $j_{2}=0$, by \eqref{double-com}, \eqref{excitation-bdd}--\eqref{excitation-bdd5}, \eqref{E2-com-single-esti}--\eqref{E2-com-single-esti}--\eqref{E2-com-single-esti1}, Cauchy-Schwarz inequality, Lemmas \ref{lem:Gronwall}, \ref{lem:excitation-Gronwall} and Proposition \ref{prop:prod-Kk-esti}, we have
\begin{align}\label{esti-double-com1-1}
	&\sum_{k,l\in\Z_{*}^{3}}\sum_{p\in S_{k}\cap S_{l}}\big|\big\langle\Phi_{s_{1}},\tilde{a}_{p_{2}}^{*}[b_{k}(T_{k,s_{2}}^{(j_{1},0)}e_{p_{1}}),[b_{l}(K_{l}e_{p_{4}}),\tilde{a}_{p_{3}}^{*}]^{*}]\Phi_{s_{1}}\big\rangle\big|\nonumber\\
	&=\sum_{k,l\in\Z_{*}^{3}}\sum_{p\in L_{k}\cap L_{l}}\chi_{L_{k}\cap L_{l}}(k+l-p)\big|\langle T_{k,s_{2}}^{(j_{1},0)}e_{p},e_{k+l-p}\rangle\big|\big|\langle e_{k+l-p},K_{l}e_{p}\rangle\big|\|\tilde{a}_{-p+l}\Phi_{s_{1}}\|^{2}\nonumber\\
	&\leq C k_{F}^{-2}\sum_{(k,\zeta)\in\cC_{\xi}}\sum_{l\in\Z_{*}^{3}}\frac{\hat{V}_{k}^{j_{1}}\hat{V}_{l}\chi_{L_{k}\cap L_{l}}(\zeta)\chi_{L_{k}\cap L_{l}}(k+l-\zeta)}{(\lambda_{k,\zeta}+\lambda_{k,k+l-\zeta})(\lambda_{l,\zeta}+\lambda_{l,k+l-\zeta})}\|\tilde{a}_{l-\zeta}\Phi_{s_{1}}\|^{2}\nonumber\\
	&\leq Ck_{F}^{-2}\cQ\sum_{k,l\in\Z_{*}^{3}}\Big(\frac{\hat{V}_{k}^{j_{1}}\hat{V}_{l}\chi_{L_{k}\cap L_{l}}(\xi)}{\lambda_{k,\xi}\lambda_{l,\xi}}\|\tilde{a}_{l-\xi}\Phi_{s_{1}}\|+\frac{\hat{V}_{k}^{j_{1}}\hat{V}_{l}\chi_{L_{k}\cap L_{l}}(k+\xi)\chi_{L_{l}'}(-\xi)}{\lambda_{k,k+\xi}\lambda_{l,l-\xi}}\|\tilde{a}_{l-k-\xi}\Phi_{s_{1}}\|\Big)\nonumber\\
	&\leq Ck_{F}^{-2}\cQ m(\xi)\sum_{k\in\Z_{*}^{3}}\frac{\hat{V}_{k}^{j_{1}}\chi_{L_{k}}(\xi)}{\lambda_{k,\xi}}\sqrt{\sum_{l\in\Z_{*}^{3}}\hat{V}_{l}^{2}}\sqrt{\sum_{l\in\Z_{*}^{3}}\chi_{L_{l}}(\xi)\|\tilde{a}_{l-\xi}\Phi_{s_{1}}\|^{2}}\nonumber\\
	&\quad+Ck_{F}^{-2}\cQ m(\xi)\sum_{l\in\Z_{*}^{3}}\frac{\hat{V}_{l}\chi_{L_{l}'}(-\xi)}{\lambda_{l,l-\xi}}\sqrt{\sum_{k\in\Z_{*}^{3}}\hat{V}_{k}^{2j_{1}}}\sqrt{\sum_{k\in\Z_{*}^{3}}\chi_{L_{l}}(k+\xi)\tilde{a}_{l-k-\xi}\Phi_{s_{1}}\|^{2}}\nonumber\\
	&\leq C_{V}k_{F}^{-2}\cQ m(\xi)\Big(\sqrt{\sum_{k\in\Z_{*}^{3}}\hat{V}_{k}^{2j_{1}}}\sqrt{\sum_{k\in\Z_{*}^{3}}\frac{\chi_{L_{k}}(\xi)}{\lambda_{k,\xi}}}+\sqrt{\sum_{l\in\Z_{*}^{3}}\hat{V}_{l}^{2}}\sqrt{\sum_{l\in\Z_{*}^{3}}\frac{\chi_{L_{l}'}(-\xi)}{\lambda_{l,l-\xi}^{2}}}\Big)\nonumber\\
	&\leq C_{\delta,V}k_{F}^{-3/2+\delta}\cQ m(\xi).
\end{align}

Next, for \eqref{esti-double-com2}, we have
\begin{align*}
	&\sum_{k,l\in\Z_{*}^{3}}\sum_{p\in S_{k}\cap S_{l}}\big|\big\langle\Phi_{s_{1}},[b_{l}(K_{l}e_{p_{4}}),[b_{k}(T_{k,s_{2}}^{(j_{1},0)}e_{p_{1}}),\tilde{a}_{p_{2}}^{*}]^{*}]^{*}\tilde{a}_{p_{3}}\Phi_{s_{1}}\big\rangle\big|\nonumber\\
	&=\sum_{k,l\in\Z_{*}^{3}}\sum_{p\in L_{k}\cap L_{l}}\chi_{L_{k}\cap L_{l}}(-p+l+k)\big|\langle K_{l}e_{p},e_{-p+l+k}\rangle\big|\big|\langle e_{-p+l+k},T_{k,s_{2}}^{(j_{1},0)}e_{p}\rangle\big|\|\tilde{a}_{-p+k}\Phi_{s_{1}}\|^{2}.
\end{align*}
We observe that, by reversing the roles of $k$ and $l$, this is exactly the same as \eqref{esti-double-com1-1} and thus satisfies the same bound.  For \eqref{esti-double-com3}, by \eqref{ba-comm}, in the case $S_{k}=L_{k}$, we have
\begin{align}
	&[b_{l}(K_{l}e_{p_{4}}),\tilde{a}_{p_{3}}^{*}]^{*}[b_{k}(T_{k,s_{2}}^{(j_{1},0)}e_{p_{1}}),\tilde{a}_{p_{2}}^{*}]\nonumber\\
	&=\chi_{L_{k}}(p_{2}+k)\chi_{L_{l}}(p_{3}+l)\langle T_{k,s_{2}}^{(j_{1},0)}e_{p_{1}},e_{p_{2}+k}\rangle\langle e_{p_{3}+l},K_{l}e_{p_{4}}\rangle\tilde{a}_{p_{3}+l}^{*}\tilde{a}_{p_{2}+k},
\end{align}
so for each $\delta>0$ that 
\begin{align}\label{esti-double-com3-1}
	&\sum_{k,l\in\Z_{*}^{3}}\sum_{p\in S_{k}\cap S_{l}}\big|\big\langle\Phi_{s_{1}},[b_{l}(K_{l}e_{p_{4}}),\tilde{a}_{p_{3}}^{*}]^{*}[b_{k}(T_{k,s_{2}}^{(j_{1},0)}e_{p_{1}}),\tilde{a}_{p_{2}}^{*}]\Phi_{s_{1}}\big\rangle\big|\nonumber\\
	&=\sum_{(k,\zeta)\in\cC_{\xi}}\sum_{l\in\Z_{*}^{3}}\chi_{L_{k}\cap L_{l}}(\zeta)\chi_{L_{k}\cap L_{l}}(k+l-\zeta)\big|\langle A_{k,s_{2}}^{(j_{1})}e_{\zeta},e_{k+l-\zeta}\rangle\big|\big|\langle e_{k+l-\zeta},K_{l}e_{\zeta}\rangle\big|\big|\|\tilde{a}_{k+l-\zeta}\Phi_{s_{1}}\|^{2}\nonumber\\
	&\leq C k_{F}^{-2}\sum_{(k,\zeta)\in\cC_{\xi}}\sum_{l\in\Z_{*}^{3}}\frac{\hat{V}_{k}^{j_{1}}\hat{V}_{l}\chi_{L_{k}\cap L_{l}}(\zeta)\chi_{L_{k}\cap L_{l}}(k+l-\zeta)}{(\lambda_{k,\zeta}+\lambda_{k,k+l-\zeta})(\lambda_{l,\zeta}+\lambda_{l,k+l-\zeta})}\|\tilde{a}_{k+l-\zeta}\Phi_{s_{1}}\|^{2}\nonumber\\
	&\leq Ck_{F}^{-2}\cQ\sum_{k,l\in\Z_{*}^{3}}\hat{V}_{k}^{j_{1}}\hat{V}_{l}\Big(\frac{\chi_{L_{k}}(\xi)\chi_{L_{l}'}(k-\xi)}{\lambda_{k,\xi}\lambda_{l,\xi}}\|\tilde{a}_{k+l-\xi}\Phi_{s_{1}}\|+\frac{\chi_{L_{k}}(\xi+k)\chi_{L_{l}}(l-\xi)}{\lambda_{k,k+\xi}\lambda_{l,l-\xi}}\|\tilde{a}_{l-\xi}\Phi_{s_{1}}\|\Big)\nonumber\\
	&\leq Ck_{F}^{-2}\cQ m(\xi)\sum_{k\in\Z_{*}^{3}}\frac{\hat{V}_{k}^{j_{1}}\chi_{L_{k}}(\xi)}{\lambda_{k,\xi}}\sqrt{\sum_{l\in\Z_{*}^{3}}\hat{V}_{l}^{2}}\sqrt{\sum_{l\in\Z_{*}^{3}}\chi_{L_{l}}(k+l-\xi)\|\tilde{a}_{k+l-\xi}\Phi_{s_{1}}\|^{2}}\nonumber\\
	&\quad\quad+Ck_{F}^{-2}\cQ m(\xi)\sum_{k\in\Z_{*}^{3}}\frac{\hat{V}_{k}^{j_{1}}\chi_{L_{k}'}(\xi)}{\lambda_{k,k+\xi}}\sqrt{\sum_{l\in\Z_{*}^{3}}\hat{V}_{l}^{2}}\sqrt{\sum_{l\in\Z_{*}^{3}}\chi_{L_{l}}(l-\xi)\|\tilde{a}_{l-\xi}\Phi_{s_{1}}\|^{2}}\nonumber\\
	&\leq C_{V}k_{F}^{-2}\cQ m(\xi)\|\cN_{E}^{1/2}\Phi_{s_{1}}\|\sqrt{\sum_{k\in\Z_{*}^{3}}\hat{V}_{k}^{2j_{1}}}\Big(\sqrt{\sum_{k\in\Z_{*}^{3}}\frac{\chi_{L_{k}}(\xi)}{\lambda_{k,\xi}}}+\sqrt{\sum_{k\in\Z_{*}^{3}}\frac{\chi_{L_{k}'}(\xi)}{\lambda_{k,k+\xi}^{2}}}\Big)\nonumber\\
	&\leq C_{\delta,V}k_{F}^{-3/2+\delta}\cQ m(\xi).
\end{align}

Finally, we mention that the quantities \eqref{esti-double-com1}--\eqref{esti-double-com3} with $j_{1}=0$ have same structure as in \eqref{esti-double-com3-1}--\eqref{esti-double-com3} with $j_{2}=0$ and thus they satisfy the same bounds.  
\end{proof}

\subsection{Analysis of \eqref{top-hard}}\label{sec:esti-top-hard}
In this section, we estimate the term \eqref{top-hard}.  As remarked earlier, if we use the same method in the proof of Proposition \ref{prop:esti-top1}, then this term is of order $O(k_{F}^{-1})$, which is not small in comparison with the bosonization contribution.  

In the same spirit as in Subsections \ref{sec:esti-errEk1}--\ref{sec:esti-errEk2}, we perform another Bogoliubov transformation and extract its leading order contribution.  To simplify notations, we note that, from the error terms listed in Subsection \ref{sec:class-errors}, it suffices to consider the case where $T_{k,s_{2}}^{(1,0)}=K_{k}$.  Moreover, we recall the definition of $(p_{1},p_{2},p_{3},p_{4})$ below \eqref{E2k-generic} and compute for $S_{k}=L_{k}$ that
\begin{align}\label{top-hardLk}
	&\sum_{k,l\in\Z_{*}^{3}}\sum_{p\in S_{k}\cap S_{l}}\tilde{a}_{p_{2}}^{*}b_{l}^{*}(K_{l}e_{p_{4}})b_{k}(K_{k}e_{p_{1}})\tilde{a}_{p_{3}}\nonumber\\
	&=2\sum_{k,l\in\Z_{*}^{3}}\chi_{L_{k}\cap L_{l}}(\xi)\tilde{a}_{l-\xi}^{*}b_{l}^{*}(K_{l}e_{\xi})b_{k}(K_{k}e_{\xi})\tilde{a}_{k-\xi}\nonumber\\
	&\quad+2\sum_{k,l\in\Z_{*}^{3}}\chi_{L_{k}\cap L_{l}}(k+\xi)\tilde{a}_{-k+l-\xi}^{*}b_{l}^{*}(K_{l}e_{k+\xi})b_{k}(K_{k}e_{k+\xi})\tilde{a}_{\xi},
\end{align}
and for $S_{k}=L_{k}'$ that
\begin{align}\label{top-hardLk'}
	&\sum_{k,l\in\Z_{*}^{3}}\sum_{p\in S_{k}\cap S_{l}}\tilde{a}_{p_{2}}^{*}b_{l}^{*}(K_{l}e_{p_{4}})b_{k}(K_{k}e_{p_{1}})\tilde{a}_{p_{3}}\nonumber\\
	&=2\sum_{k,l\in\Z_{*}^{3}}\chi_{L_{k}'\cap L_{l}'}(\xi)\tilde{a}_{-l-\xi}^{*}b_{l}^{*}(K_{l}e_{l+\xi})b_{k}(K_{k}e_{k+\xi})\tilde{a}_{-k-\xi}\nonumber\\
	&\quad+2\sum_{k,l\in\Z_{*}^{3}}\chi_{L_{k}'\cap L_{l}'}(-k+\xi)\tilde{a}_{k-l-\xi}^{*}b_{l}^{*}(K_{l}e_{\xi-k+l})b_{k}(K_{k}e_{\xi})\tilde{a}_{-\xi}.
\end{align}
We note that the matrix elements of the second sums in both of \eqref{top-hardLk} and \eqref{top-hardLk'} w.r.t. state $\Phi_{s_{1}}$ can be estimated using the same method in Proposition \ref{prop:esti-A1A2} and yields the upper bound $C_{\delta,V}k_{F}^{-1}\big(\cQ m(\xi)+m(\xi)^{1/2}\cdot\sup_{0\leq \tau\leq 1}\|\tilde{a}_{\xi}\Phi_{\tau}\|\big)$, so we omit details for simplicity.  Hence, it suffices to estimate the first terms in \eqref{top-hardLk}--\eqref{top-hardLk'}.  Furthermore, since the first sums in \eqref{top-hardLk} and \eqref{top-hardLk'} share similar structure, it suffices to only consider the first sum in \eqref{top-hardLk}.

Next, by the fundamental theorem of calculus and symmetry in summation over $k$ and $l$, we obtain
\begin{align}\label{top-hard1}
	&\sum_{k,l\in\Z_{*}^{3}}\chi_{L_{k}\cap L_{l}}(\xi)\big\langle\Phi_{s_{1}},\tilde{a}_{l-\xi}^{*}b_{l}^{*}(K_{l}e_{\xi})b_{k}(K_{k}e_{\xi})\tilde{a}_{k-\xi}\Phi_{s_{1}}\big\rangle\nonumber\\
	&=\Re\sum_{k,l\in\Z_{*}^{3}}\chi_{L_{k}\cap L_{l}}(\xi)\int_{0}^{s_{1}}\big\langle\Phi_{s_{3}},\tilde{a}_{l-\xi}^{*}b_{l}^{*}(K_{l}e_{\xi})[\cK,b_{k}(K_{k}e_{\xi})\tilde{a}_{k-\xi}]\Phi_{s_{3}}\big\rangle ds_{3}.
\end{align}
Using CAR and Proposition \ref{prop:Kcom-excit} and then substituting into \eqref{top-hard1}, we obtain
\begin{align}
	&\sum_{k,l\in\Z_{*}^{3}}\chi_{L_{k}\cap L_{l}}(\xi)\tilde{a}_{l-\xi}^{*}b_{l}^{*}(K_{l}e_{\xi})[\cK,b_{k}(K_{k}e_{\xi})\tilde{a}_{k-\xi}]\nonumber\\
	\label{top-hard2-1}&=\sum_{(k,l)\in\cB_{\xi}}\tilde{a}_{l-\xi}^{*}b_{l}^{*}(K_{l}e_{\xi})b_{-k}^{*}(K_{-k}^{2}e_{-\xi})\tilde{a}_{k-\xi}\\
	\label{top-hard2-2}&\quad-\sum_{(k,l,\ell)\in\cB_{\xi}'}\tilde{a}_{l-\xi}^{*}b_{l}^{*}(K_{l}e_{\xi})\Big(\tilde{a}_{-\eta}^{*}b_{\ell}^{*}(K_{\ell}e_{\eta})b_{k}(K_{k}e_{\xi})+b_{\ell}^{*}(K_{\ell}e_{\eta})[b_{k}(K_{k}e_{\xi}),\tilde{a}_{-\eta}^{*}]\nonumber\\
	&\quad\quad\quad\quad\quad\quad\quad\quad\quad\quad\quad\quad\quad\quad+[b_{k}(K_{k}e_{\xi}),b_{\ell}^{*}(K_{\ell}e_{\eta})]\tilde{a}_{-\eta}^{*}\Big)\\
	\label{top-hard2-3}&\quad+\frac{1}{2}\sum_{(k,l)\in\cB_{\xi}}\sum_{\ell\in\Z_{*}^{3}}\sum_{q\in L_{\ell}}\tilde{a}_{l-\xi}^{*}b_{l}^{*}(K_{l}e_{\xi})\{\varepsilon_{k,\ell}(K_{k}e_{\xi};e_{q}),b_{-\ell}^{*}(K_{-\ell}e_{-q})\}\tilde{a}_{k-\xi},
\end{align}
where we denote $\eta:=\xi-k+\ell$, $\cB_{\xi}:=\{(k,l)\in (\Z_{*}^{3})^{2}\mid \xi\in L_{k}\cap L_{l}\}$ and $\cB_{\xi}':=\{(k,l,\ell)\in(\Z_{*}^{3})^{3}\mid \xi\in L_{k}\cap L_{l},\xi-k\in L_{\ell}'\}$.  We will analyse the terms in \eqref{top-hard2-1}--\eqref{top-hard2-3} separately as they require different strategies.



We begin by estimating \eqref{top-hard2-1} in the next proposition:
\begin{prop}\label{prop:esti-top-hard2-1}
For each $0\leq s_{3}\leq 1$, it holds for each $\delta>0$ that
\begin{align}\label{esti-top-hard2-1}
	\sum_{(k,l)\in\cB_{\xi}}\big|\big\langle\Phi_{s_{3}},\tilde{a}_{l-\xi}^{*}b_{l}^{*}(K_{l}e_{\xi})b_{-k}^{*}(K_{-k}^{2}e_{-\xi})\tilde{a}_{k-\xi}\Phi_{s_{3}}\big\rangle\big|\leq C_{V}k_{F}^{-1}\cQ m(\xi).
\end{align}
\end{prop}
\begin{proof}
By \eqref{excitation-bdd}--\eqref{excitation-bdd5}, \eqref{E2-com-single-esti}--\eqref{E2-com-single-esti1}, Cauchy-Schwarz inequality, Lemmas \ref{lem:Gronwall}, \ref{lem:excitation-Gronwall}, Proposition \ref{prop:prod-Kk-esti}, the relation $[b_{l,p},\cN_{E}]=b_{l,p}$ and the operator norm $\|\tilde{a}_{q}\|=1$, we have
\begin{align}\label{esti-top-hard2-1-comp}
	&\sum_{(k,l)\in\cB_{\xi}}\big|\big\langle\Phi_{s_{3}},\tilde{a}_{l-\xi}^{*}b_{l}^{*}(K_{l}e_{\xi})b_{-k}^{*}(K_{-k}^{2}e_{-\xi})\tilde{a}_{k-\xi}\Phi_{s_{3}}\big\rangle\big|\nonumber\\
	&\leq\sum_{(k,l)\in\cB_{\xi}}\|\tilde{a}_{l-\xi}b_{-k}(K_{-k}^{2}e_{-\xi})b_{l}(K_{l}e_{\xi})\Phi_{s_{3}}\|\|\tilde{a}_{k-\xi}\Phi_{s_{3}}\|\nonumber\\
	&\leq k_{F}^{-1/2}\cQ m(\xi)^{1/2}\sum_{l\in\Z_{*}^{3}}\chi_{L_{l}}(\xi)\sqrt{\sum_{k\in\Z_{*}^{3}}\hat{V}_{k}^{4}}\sqrt{\sum_{k\in\Z_{*}^{3}}\|\cN_{-k}^{1/2}b_{l}(K_{l}e_{\xi})\Phi_{s_{3}}\|^{2}}\nonumber\\
	&\leq C_{V}k_{F}^{-1/2}\cQ m(\xi)^{1/2}\sum_{l\in\Z_{*}^{3}}\chi_{L_{l}}(\xi)\|\cN_{E}b_{l}(K_{l}e_{\xi})\Phi_{s_{3}}\|\nonumber\\
	&=C_{V}k_{F}^{-1/2}\cQ m(\xi)^{1/2}\sum_{l\in\Z_{*}^{3}}\chi_{L_{l}}(\xi)\|b_{l}(K_{l}e_{\xi})(\cN_{E}+1)\Phi_{s_{3}}\|\nonumber\\
	&\leq C_{V}k_{F}^{-1}\cQ m(\xi)\sqrt{\sum_{l\in\Z_{*}^{3}}\hat{V}_{l}^{2}}\sqrt{\sum_{l\in\Z_{*}^{3}}\|\cN_{l}^{1/2}(\cN_{E}+1)\Phi_{s_{3}}\|^{2}}\leq C_{V}k_{F}^{-1}\cQ m(\xi).
\end{align}
This completes the proof.
\end{proof}

Next, we consider the terms in \eqref{top-hard2-2}.  Before proceeding, we remark that the first two sums and the last one require different strategies to proceed.  In the following proposition, we conisder the first two sums in \eqref{top-hard2-2}:
\begin{prop}\label{prop:esti-top-hard2-2-1}
For each $0\leq s_{3}\leq 1$, it holds that
\begin{align}
	\label{esti-top-hard2-2-11}&\sum_{(k,l,\ell)\in\cB_{\xi}'}\big|\big\langle \Phi_{s_{3}},\tilde{a}_{l-\xi}^{*}b_{l}^{*}(K_{l}e_{\xi})b_{\ell}^{*}(K_{\ell}e_{\eta})\tilde{a}_{-\eta}^{*}b_{k}(K_{k}e_{\xi})\Phi_{s_{3}}\big\rangle\big|,\\
	\label{esti-top-hard2-2-12}&\sum_{(k,l,\ell)\in\cB_{\xi}'}\big|\big\langle\Phi_{s_{3}},\tilde{a}_{l-\xi}^{*}b_{l}^{*}(K_{l}e_{\xi})b_{\ell}^{*}(K_{\ell}e_{\eta})[b_{k}(K_{k}e_{\xi}),\tilde{a}_{-\eta}^{*}]\Phi_{s_{3}}\big\rangle\big|,
\end{align}
are bounded by $C_{\delta,V}\big(k_{F}^{-3/2}+k_{F}^{-3/2+\delta}\cQ\big)m(\xi)$.
\end{prop}
\begin{proof}
For \eqref{esti-top-hard2-2-11}, by \eqref{excitation-bdd}--\eqref{excitation-bdd5}, \eqref{E2-com-single-esti}--\eqref{E2-com-single-esti1}, Cauchy-Schwarz inequality, Lemmas \ref{lem:Gronwall}, \ref{lem:excitation-Gronwall}, Proposition \ref{prop:prod-Kk-esti}, the relations $[b_{l,p},\cN_{E}]=b_{l,p}$ and the operator norm $\|\tilde{a}_{q}\|=1$, we use the same computation in \eqref{esti-top-hard2-1-comp} to obtain
\begin{align}
	&\sum_{(k,l,\ell)\in\cB_{\xi}'}\big|\big\langle \Phi_{s_{3}},\tilde{a}_{l-\xi}^{*}b_{l}^{*}(K_{l}e_{\xi})b_{\ell}^{*}(K_{\ell}e_{\eta})\tilde{a}_{-\eta}^{*}b_{k}(K_{k}e_{\xi})\Phi_{s_{3}}\big\rangle\big|\nonumber\\
	&\leq \sum_{(k,l,\ell)\in\cB_{\xi}'}\|K_{\ell}e_{\eta}\|\|K_{k}e_{\xi}\|\|\cN_{k}^{1/2}\Phi_{s_{3}}\|\|\cN_{\ell}^{1/2}b_{l}(K_{l}e_{\xi})\Phi_{s_{3}}\|\nonumber\\
	&\leq Ck_{F}^{-1}\sum_{(k,l,\ell)\in\cB_{\xi}'}\frac{\hat{V}_{\ell}\hat{V}_{k}}{\sqrt{\lambda_{\ell,\eta}\lambda_{k,\xi}}}\|\cN_{k}^{1/2}\Phi_{s_{3}}\|\|\cN_{\ell}^{1/2}b_{l}(K_{l}e_{\xi})\Phi_{s_{3}}\|\nonumber\\
	&\leq C_{V}k_{F}^{-1}m(\xi)^{1/2}\sum_{l\in\Z_{*}^{3}}\chi_{L_{l}}(\xi)\|b_{l}(K_{l}e_{\xi})(\cN_{E}+1)\Phi_{s_{3}}\|
	\leq C_{V}k_{F}^{-3/2}m(\xi).
\end{align}

Next, for \eqref{esti-top-hard2-2-12}, we first compute the commutator
\begin{align}
	[b_{k}(K_{k}e_{\xi}),\tilde{a}_{-\eta}^{*}]&=\chi_{L_{k}}(-\eta)\langle K_{k}e_{\xi},e_{-\eta}\rangle\tilde{a}_{-\eta-k},
\end{align}
where we have used \cite[Eq.(4.23)]{CHN-23} and the fact that $\eta\in B_{F}^{c}$.  By substituting the definition $\eta=\xi-k+\ell$, it follows from the same argument above that
\begin{align}\label{esti-top-hard2-2-121}
	&\sum_{(k,l,\ell)\in\cB_{\xi}'}\big|\big\langle\Phi_{s_{3}},\tilde{a}_{l-\xi}^{*}b_{l}^{*}(K_{l}e_{\xi})b_{\ell}^{*}(K_{\ell}e_{\eta})[b_{k}(K_{k}e_{\xi}),\tilde{a}_{-\eta}^{*}]\Phi_{s_{3}}\big\rangle\big|\nonumber\\
	&\leq \sum_{(k,l,\ell)\in\cB_{\xi}'}\chi_{L_{k}}(-\eta)\big|\langle K_{k}e_{\xi},e_{-\eta}\rangle\big|\|K_{l}e_{\xi}\|\|\cN_{l}^{1/2}b_{\ell}(K_{\ell}e_{\eta})\Phi_{s_{3}}\|\|\tilde{a}_{-\eta-k}\Phi_{s_{3}}\|\nonumber\\
	&\leq k_{F}^{-3/2}\cQ\sum_{(k,l,\ell)\in\cB_{\xi}'}\chi_{L_{k}}(-\eta)\frac{\hat{V}_{k}\hat{V}_{l}}{(\lambda_{k,\xi}+\lambda_{k,-\eta})\sqrt{\lambda_{l,\xi}}}\|\cN_{l}^{1/2}b_{\ell}(K_{\ell}e_{\eta})\Phi_{s_{3}}\|\nonumber\\
	&\leq C_{V}k_{F}^{-3/2}\cQ m(\xi)\sum_{k,\ell\in\Z_{*}^{3}}\frac{\hat{V}_{k}\hat{V}_{\ell}\chi_{L_{k}}(\xi)\chi_{L_{k}'}(-\xi-\ell)}{\sqrt{\lambda_{k,k-\ell-\xi}\lambda_{\ell,\xi-k+\ell}}}\|\cN_{\ell}^{1/2}(\cN_{E}+1)\Phi_{s_{3}}\|\nonumber\\
	&\leq C_{V}k_{F}^{-3/2}\cQ m(\xi)\sum_{\ell\in\Z_{*}^{3}}\hat{V}_{\ell}\|\cN_{\ell}^{1/2}(\cN_{E}+1)\Phi_{s_{3}}\|\nonumber\\
	&\quad\quad\quad\quad\quad\quad\times\sqrt{\sum_{k\in\Z_{*}^{3}}\hat{V}_{k}^{2}}\sqrt{\sum_{k\in\Z_{*}^{3}}\chi_{L_{k}'}(-\xi-\ell)m(k-\ell-\xi)^{2}}.
\end{align}
Using the same argument as in the proof of Lemma \ref{lem:gap-sum} and lower bound \eqref{gap}, the last factor in \eqref{esti-top-hard2-2-121} yields for each $\delta>0$ that
\begin{align}
	\sum_{k\in\Z_{*}^{3}}\chi_{L_{k}'}(-\xi-\ell)m(k-\ell-\xi)^{2}&\leq C_{\delta}k_{F}^{1+\delta}m(\xi+\ell)\leq C_{\delta}k_{F}^{1+\delta}.
\end{align}
Hence, \eqref{esti-top-hard2-2-121} becomes
\begin{align*}
	&\sum_{(k,l,\ell)\in\cB_{\xi}'}\big|\big\langle\Phi_{s_{3}},\tilde{a}_{l-\xi}^{*}b_{l}^{*}(K_{l}e_{\xi})b_{\ell}^{*}(K_{\ell}e_{\eta})[b_{k}(K_{k}e_{\xi}),\tilde{a}_{-\eta}^{*}]\Phi_{s_{3}}\big\rangle\big|\nonumber\\
	&\leq C_{\delta,V}k_{F}^{-1+\delta}\cQ m(\xi)\sqrt{\sum_{\ell\in\Z_{*}^{3}}\hat{V}_{\ell}^{2}}\sqrt{\sum_{\ell\in\Z_{*}^{3}}\|\cN_{\ell}^{1/2}(\cN_{E}+1)\Phi_{s_{3}}\|^{2}}\leq C_{\delta,V}k_{F}^{-1+\delta}\cQ m(\xi).
\end{align*}
This completes the proof.
\end{proof}

Now, for the last sum in \eqref{top-hard2-2}, we compute using Lemma \ref{lem:CR-excitation} to obtain
\begin{align*}
	\tilde{a}_{l-\xi}^{*}&b_{l}^{*}(K_{l}e_{\xi})[b_{k}(K_{k}e_{\xi}),b_{\ell}^{*}(K_{\ell}e_{\eta})]\tilde{a}_{-\eta}^{*}\nonumber\\
	&=\delta_{k,\ell}\|K_{k}e_{\xi}\|^{2} \tilde{a}_{l-\xi}^{*}\tilde{a}_{-\xi}^{*}b_{l}^{*}(K_{l}e_{\xi})+\tilde{a}_{l-\xi}^{*}b_{l}^{*}(K_{l}e_{\xi})\varepsilon_{k,\ell}(K_{k}e_{\xi};K_{\ell}e_{-\eta})\tilde{a}_{-\eta}^{*}.
\end{align*}
where we have used the fact that $\eta=\xi$ when $k=\ell$.  Recall from \eqref{exchange-correction} that 
\begin{align*}
	-\varepsilon_{k,\ell}(K_{k}e_{\xi};K_{\ell}e_{\eta})&=\sum_{p\in L_{k}}\sum_{q\in L_{\ell}}\langle e_{\xi},K_{k}e_{p}\rangle\langle e_{q},K_{\ell}e_{\eta}\rangle\big(\delta_{p,q}\tilde{a}_{q-\ell}^{*}\tilde{a}_{p-k}+\delta_{p-k,q-\ell}\tilde{a}_{q}^{*}\tilde{a}_{p}\big)\nonumber\\
	&=\sum_{q\in L_{k}\cap L_{\ell}}\langle e_{\xi},K_{k}e_{q}\rangle\langle e_{q},K_{\ell}e_{\eta}\rangle\tilde{a}_{q-\ell}^{*}\tilde{a}_{q-k}\nonumber\\
	&\quad\quad\quad+\sum_{q\in L_{k}'\cap L_{\ell}'}\langle e_{\xi},K_{k}e_{q+k}\rangle\langle e _{q+\ell},K_{\ell}e_{\eta}\rangle\tilde{a}_{q+\ell}^{*}\tilde{a}_{q+k}
\end{align*}
We observe that, as suggested in \cite{CHN-23}, these sums take a common schematic form:
\begin{align}\label{schematic}
	\sum_{q\in S_{k}\cap S_{\ell}}\langle e_{\xi},K_{k}e_{q_{1}}\rangle\langle e_{q_{4}},K_{\ell}e_{-\eta}\rangle\tilde{a}_{q_{2}}^{*}\tilde{a}_{q_{3}},
\end{align}
where
\begin{align}\label{q-variable}
	(q_{1},q_{2},q_{3},q_{4})&=\begin{cases}
		(q, q-\ell, q-k,q)\quad&\text{ if }S_{k}=L_{k},\\
		(q+k,q+\ell,q+k,q+\ell)&\text{ if }S_{k}=L_{k}'.
	\end{cases}
\end{align}
In summary, to estimate the last commutator term in \eqref{top-hard2-2}, it suffices to estimate
\begin{align}
	&\sum_{k,l\in\Z_{*}^{3}}\chi_{L_{k}\cap L_{l}}(\xi)\|K_{k}e_{\xi}\|^{2}\tilde{a}_{l-\xi}^{*}\tilde{a}_{-\xi}^{*}b_{l}^{*}(K_{l}e_{\xi}),\\
	&\sum_{(k,l,\ell)\in\cB_{\xi}'}\sum_{q\in S_{k}\cap S_{\ell}}\langle e_{\xi},K_{k}e_{q_{1}}\rangle\langle e_{q_{4}},K_{\ell}e_{\eta}\rangle\tilde{a}_{l-\xi}^{*}b_{l}^{*}(K_{l}e_{\xi})\tilde{a}_{q_{2}}^{*}\tilde{a}_{q_{3}}\tilde{a}_{-\eta}^{*}.
\end{align}

\begin{prop}\label{prop:esti-top-hard2-2}
For each $0\leq s_{3}\leq 1$, it holds for each $\delta>0$ that 
\begin{align}
	\label{esti-top-hard2-2-1}&\sum_{k,l\in\Z_{*}^{3}}\chi_{L_{k}\cap L_{l}}(\xi)\|K_{k}e_{\xi}\|^{2}\big|\big\langle\Phi_{s_{3}},\tilde{a}_{l-\xi}^{*}\tilde{a}_{-\xi}^{*}b_{l}^{*}(K_{l}e_{\xi})\Phi_{s_{3}}\big\rangle\big|,\\
	\label{esti-top-hard2-2-2}&\sum_{(k,l,\ell)\in\cB_{\xi}'}\sum_{q\in S_{k}\cap S_{\ell}}\big|\langle e_{\xi},K_{k}e_{q_{1}}\rangle\big|\big|\langle e_{q_{4}},K_{\ell}e_{\eta}\rangle\big|\big|\langle\Phi_{s_{3}},\tilde{a}_{l-\xi}^{*}b_{l}^{*}(K_{l}e_{\xi})\tilde{a}_{q_{2}}^{*}\tilde{a}_{q_{3}}\tilde{a}_{-\eta}^{*}\Phi_{s_{3}}\big\rangle\big|,
\end{align}
are bounded by $C_{\delta,V}\big(k_{F}^{-3/2}\cQ^{1-\delta}+k_{F}^{-2+\delta}\big)m(\xi)$.
\end{prop}
\begin{proof}
For \eqref{esti-top-hard2-2-1}, by \eqref{excitation-bdd}--\eqref{excitation-bdd5}, \eqref{E2-com-single-esti}--\eqref{E2-com-single-esti1}, Cauchy-Schwarz inequality, Lemmas \ref{lem:Gronwall}, \ref{lem:excitation-Gronwall}, Proposition \ref{prop:prod-Kk-esti} and the operator norm $\|\tilde{a}_{q}\|=1$, we have
\begin{align}
	&\sum_{k,l\in\Z_{*}^{3}}\chi_{L_{k}\cap L_{l}}(\xi)\|K_{k}e_{\xi}\|^{2}\big\langle\Phi_{s_{3}},\tilde{a}_{l-\xi}^{*}\tilde{a}_{-\xi}^{*}b_{l}^{*}(K_{l}e_{\xi})\Phi_{s_{3}}\big\rangle\big|\nonumber\\
	&\leq \sum_{k,l\in\Z_{*}^{3}}\chi_{L_{k}\cap L_{l}}(\xi)\frac{\hat{V}_{k}^{2}}{\lambda_{k,\xi}}\|\tilde{a}_{l-\xi}b_{l}(K_{l}e_{\xi})\tilde{a}_{-\xi}\Phi_{s_{3}}\|\nonumber\\
	&\leq C_{V}k_{F}^{-3/2}m(\xi) \Big(\sum_{k\in\Z_{*}^{3}}\hat{V}_{k}^{2}\Big)\sqrt{\sum_{l\in\Z_{*}^{3}}\hat{V}_{l}^{2}}\sqrt{\sum_{l\in\Z_{*}^{3}}\|\cN_{l}^{1/2}\tilde{a}_{-\xi}\Phi_{s_{3}}\|^{2}}\nonumber\\
	&\leq C_{V}k_{F}^{-3/2}m(\xi)\|\tilde{a}_{-\xi}\cN_{E}\Phi_{s_{3}}\|\leq C_{\delta,V}k_{F}^{-3/2}\cQ^{1-\delta}m(\xi).
\end{align}

Next, for \eqref{esti-top-hard2-2-2}, we first put operators in the summand into normal order:
\begin{align}\label{esti-top-hard2-2-21}
	\tilde{a}_{q_{2}}^{*}\tilde{a}_{q_{3}}\tilde{a}_{-\eta}^{*}&=\delta_{q_{3},-\eta}\tilde{a}_{q_{2}}^{*}+\tilde{a}_{q_{2}}^{*}\tilde{a}_{-\eta}^{*}\tilde{a}_{q_{3}}.
\end{align}
For the part involving the first operator in \eqref{esti-top-hard2-2-21}, we note that, since $\eta\in B_{F}^{c}$ and $q_{3}=q-k\in B_{F}$ if $S_{k}=L_{k}$, this sum vanishes for $S_{k}=L_{k}$.  For $S_{k}=L_{k}'$, we substitute the definition \eqref{q-variable} and use a similar argument as in the proof of Lemma \ref{lem:gap-sum} to obtain for each $\delta>0$ that
\begin{align}
	&\sum_{(k,l,\ell)\in\cB_{\xi}'}\sum_{q\in S_{k}\cap S_{\ell}}\delta_{q_{3},-\eta}\big|\langle e_{\xi},K_{k}e_{q_{1}}\rangle\big|\big|\langle e_{q_{4}},K_{\ell}e_{\eta}\rangle\big|\big|\big\langle\Phi_{s_{3}},\tilde{a}_{l-\xi}^{*}b_{l}^{*}(K_{l}e_{\xi})\tilde{a}_{q_{2}}\Phi_{s_{3}}\big\rangle\big|\nonumber\\
	&\leq Ck_{F}^{-2}\sum_{(k,l,\ell)\in\cB_{\xi}'}\sum_{q\in S_{k}\cap S_{\ell}}\delta_{q_{3},-\eta}\frac{\hat{V}_{k}\chi_{L_{k}}(\xi)}{\lambda_{k,\xi}+\lambda_{k,q_{1}}}\frac{\hat{V}_{\ell}\chi_{L_{\ell}}(\eta)}{\lambda_{\ell,\eta}+\lambda_{\ell,q_{4}}}\|K_{l}e_{\xi}\|\|\cN_{l}^{1/2}\tilde{a}_{q_{2}}\Phi_{s_{3}}\|\nonumber\\
	\displaybreak
	&\leq Ck_{F}^{-5/2}m(\xi)^{1/2}\sum_{k,\ell\in\Z_{*}^{3}}\sum_{q\in S_{k}\cap S_{\ell}}\delta_{q_{3},-\eta}\frac{\hat{V}_{k}\chi_{L_{k}}(\xi)}{\lambda_{k,\xi}+\lambda_{k,q_{1}}}\frac{\hat{V}_{\ell}\chi_{L_{\ell}}(\eta)}{\lambda_{\ell,\eta}+\lambda_{\ell,q_{4}}}\|\tilde{a}_{q_{2}}\cN_{E}\Phi_{s_{3}}\|\nonumber\\
	&=Ck_{F}^{-5/2}\cQ m(\xi)^{1/2}\sum_{k,\ell\in\Z_{*}^{3}}\chi_{L_{k}'\cap L_{\ell}'}(-\xi-\ell)\frac{\hat{V}_{k}\chi_{L_{k}}(\xi)}{\sqrt{\lambda_{k,\xi}\lambda_{k,k-\ell-\xi}}}\frac{\hat{V}_{\ell}\chi_{L_{\ell}'}(\xi-k)}{\sqrt{\lambda_{\ell,\xi-k+\ell}\lambda_{\ell,-\xi}}}\nonumber\\
	&\leq Ck_{F}^{-5/2}\cQ m(\xi)\sum_{k\in\Z_{*}^{3}}\frac{\hat{V}_{k}\chi_{L_{k}}(\xi)}{\sqrt{\lambda_{k,\xi}}}\sqrt{\sum_{\ell\in\Z_{*}^{3}}\hat{V}_{\ell}^{2}}\sqrt{\sum_{\ell\in\Z_{*}^{3}}\chi_{L_{\ell}'}(\xi-k)m(\xi-k+\ell)^{2}}\nonumber\\
	&\leq C_{\delta,V}k_{F}^{-2+\delta}\cQ m(\xi)\sqrt{\sum_{k\in\Z_{*}^{3}}\hat{V}_{k}^{2}}\sqrt{\sum_{k\in\Z_{*}^{3}}\frac{\chi_{L_{k}}(\xi)}{\lambda_{k,\xi}}}\leq C_{\delta,V}k_{F}^{-3/2+\delta}\cQ m(\xi).
\end{align}

Now, for the part involving the second operator in \eqref{esti-top-hard2-2-21}, by \eqref{excitation-bdd}--\eqref{excitation-bdd5}, \eqref{E2-com-single-esti}--\eqref{E2-com-single-esti1}, Cauchy-Schwarz inequality, Lemmas \ref{lem:gap-sum}, \ref{lem:Gronwall}, \ref{lem:excitation-Gronwall}, Proposition \ref{prop:prod-Kk-esti} and the relations $[b_{l,p},\cN_{E}]=b_{l,p}$, we obtain
\begin{align}
	&\sum_{(k,l,\ell)\in\cB_{\xi}'}\sum_{q\in S_{k}\cap S_{\ell}}\big|\langle e_{\xi},K_{k}e_{q_{1}}\rangle\big|\big|\langle e_{q_{4}},K_{\ell}e_{\eta}\rangle\big|\big|\big\langle\Phi_{s_{3}},\tilde{a}_{l-\xi}^{*}b_{l}^{*}(K_{l}e_{\xi})\tilde{a}_{q_{2}}^{*}\tilde{a}_{-\eta}^{*}\tilde{a}_{q_{3}}\Phi_{s_{3}}\big\rangle\big|\nonumber\\
	&\leq Ck_{F}^{-2}\sum_{(k,l,\ell)\in\cB_{\xi}'}\sum_{q\in S_{k}\cap S_{\ell}}\frac{\hat{V}_{k}\chi_{L_{k}}(\xi)}{\sqrt{\lambda_{k,\xi}\lambda_{k,q_{1}}}}\frac{\hat{V}_{\ell}\chi_{L_{\ell}}(\eta)}{\sqrt{\lambda_{\ell,\eta}\lambda_{\ell,q_{4}}}}\|K_{l}e_{\xi}\|\|\cN_{l}^{1/2}\tilde{a}_{l-\xi}\tilde{a}_{q_{2}}\tilde{a}_{-\eta}\Phi_{s_{3}}\|\|\tilde{a}_{q_{3}}\Phi_{s_{3}}\|\nonumber\\
	&\leq C_{V}k_{F}^{-5/2}m(\xi)\sum_{k,\ell\in\Z_{*}^{3}}\sum_{q\in S_{k}\cap S_{\ell}}\frac{\hat{V}_{k}\hat{V}_{\ell}\chi_{L_{k}}(\xi)\chi_{L_{\ell}'}(\xi-k)}{\sqrt{\lambda_{k,q_{1}}\lambda_{\ell,\xi-k+\ell}\lambda_{\ell, q_{4}}}}\|\cN_{E}\tilde{a}_{q_{2}}\tilde{a}_{k-\ell-\xi}\Phi_{s_{3}}\|\|\tilde{a}_{q_{3}}\Phi_{s_{3}}\|\nonumber\\
	&\leq C_{V}k_{F}^{-5/2}m(\xi)\sum_{k,\ell\in\Z_{*}^{3}}\frac{\hat{V}_{k}\hat{V}_{\ell}\chi_{L_{k}}(\xi)\chi_{L_{\ell'}}(\xi-k)}{\sqrt{\lambda_{\ell,\xi-k+\ell}}}\nonumber\\
	&\quad\quad\quad\quad\quad\quad\quad\quad\times\sqrt{\sum_{q\in S_{\ell}}\|\tilde{a}_{q_{2}}\cN_{E}\tilde{a}_{k-\ell-\xi}\Phi_{s_{3}}\|^{2}}\sqrt{\sum_{q\in S_{k}}\|\tilde{a}_{q_{3}}\Phi_{s_{3}}\|^{2}}\nonumber\\
	&\leq C_{V}k_{F}^{-5/2}m(\xi)\sum_{k,\ell\in\Z_{*}^{3}}\hat{V}_{k}\hat{V}_{\ell}\chi_{L_{k}}(\xi)\chi_{L_{\ell}'}(\xi-k)m(\xi-k)^{1/2}\|\tilde{a}_{k-\xi-\ell}(\cN_{E}+1)^{3/2}\Phi_{s_{3}}\|\nonumber\\
	&\leq C_{V}k_{F}^{-5/2}m(\xi)\sum_{k\in\Z_{*}^{3}}\hat{V}_{k}\chi_{L_{k}}(\xi)m(\xi-k)^{1/2}\sqrt{\sum_{\ell\in\Z_{*}^{3}}\chi_{L_{\ell}'}(\xi-k)\|\tilde{a}_{k-\ell-\xi}(\cN_{E}+1)^{3/2}\Phi_{s_{3}}\|^{2}}\nonumber\\
	&\leq C_{V}k_{F}^{-5/2}m(\xi)\sqrt{\sum_{k\in\Z_{*}^{3}}\hat{V}_{k}^{2}}\sqrt{\sum_{k\in\Z_{*}^{3}}\chi_{L_{k}}(\xi)m(\xi-k)}\leq C_{\delta,V}k_{F}^{-2+\delta}m(\xi),
\end{align}
where we use a similar argument in the proof of Lemma \ref{lem:gap-sum} to obtain for each $\delta>0$ that
\begin{align}
	\sum_{k\in\Z_{*}^{3}}\chi_{L_{k}}(\xi)m(\xi-k)&\leq C_{\delta}k_{F}^{1+\delta}.
\end{align}
This completes the proof.
\end{proof}

Now, we come to the sum in \eqref{top-hard2-3}.  Again, we observe that the summand of \eqref{top-hard2-3} splits into two sums:
\begin{align}
	&-\sum_{\ell\in\Z_{*}^{3}}\sum_{q\in L_{\ell}}\{\varepsilon_{k,\ell}(K_{k}e_{\xi};e_{q}),b_{-\ell}^{*}(K_{-\ell}e_{-q})\}\tilde{a}_{k-\xi}\nonumber\\
	&=\sum_{\ell\in\Z_{*}^{3}}\sum_{q\in L_{k}\cap L_{\ell}}\langle e_{\xi},K_{k}e_{q}\rangle\{\tilde{a}_{q-\ell}^{*}\tilde{a}_{q-k},b_{-\ell}^{*}(K_{-\ell}e_{-q})\}\tilde{a}_{k-\xi}\nonumber\\
	&\quad\quad+\sum_{\ell\in\Z_{*}^{3}}\sum_{q\in L_{k}'\cap L_{\ell}'}\langle e_{\xi},K_{k}e_{q+k}\rangle\{\tilde{a}_{q+\ell}^{*}\tilde{a}_{q+k},b_{-\ell}^{*}(K_{-\ell}e_{-q-\ell})\}\tilde{a}_{k-\xi},
\end{align}
and these sums can be written into the schematic form as in \eqref{schematic}--\eqref{q-variable}:
\begin{align}
	\sum_{\ell\in\Z_{*}^{3}}\sum_{q\in S_{k}\cap S_{\ell}}\langle e_{\xi},K_{k}e_{q_{1}}\rangle\{\tilde{a}_{q_{2}}^{*}\tilde{a}_{q_{3}},b_{-\ell}^{*}(K_{-\ell}e_{q_{4}})\}\tilde{a}_{k-\xi},
\end{align}
where the variables $(q_{1},q_{2},q_{3},q_{4})$ are now redefined as
\begin{align*}
	(q_{1},q_{2},q_{3},q_{4})&=\begin{cases}
		(q,q-\ell,q-k,-q)\quad&\text{ if }S_{k}=L_{k},\\
		(q+k,q+\ell,q+k,-q-\ell)&\text{ if }S_{k}=L_{k}'.
	\end{cases}
\end{align*}
By normal-ordering this expression, we see that \eqref{top-hard2-3} take the form
\begin{align}
	&2\sum_{(k,l)\in\cB_{\xi}}\sum_{\ell\in\Z_{*}^{3}}\sum_{q\in S_{k}\cap S_{\ell}}\langle e_{\xi},K_{k}e_{q_{1}}\rangle\tilde{a}_{l-\xi}^{*}b_{l}^{*}(K_{l}e_{\xi})\tilde{a}_{q_{2}}^{*}b_{-\ell}^{*}(K_{-\ell}e_{q_{4}})\tilde{a}_{q_{3}}\tilde{a}_{k-\xi}\nonumber\\
	&\quad+\sum_{(k,l)\in\cB_{\xi}}\sum_{\ell\in\Z_{*}^{3}}\sum_{q\in S_{k}\cap S_{\ell}}\langle e_{\xi},K_{k}e_{q_{1}}\rangle\tilde{a}_{l-\xi}^{*}b_{l}^{*}(K_{l}e_{\xi})\tilde{a}_{q_{2}}^{*}[b_{-\ell}(K_{-\ell}e_{q_{4}}),\tilde{a}_{q_{3}}^{*}]^{*}\tilde{a}_{k-\xi},
\end{align}
so that it suffices to estimate these sums:
\begin{prop}\label{prop:esti-top-hard2-3}
For each $0\leq s_{3}\leq 1$, it holds that 
\begin{align}
	\label{esti-top-hard2-3-1}&\sum_{(k,l)\in\cB_{\xi}}\sum_{\ell\in\Z_{*}^{3}}\sum_{q\in S_{k}\cap S_{\ell}}\big|\langle e_{\xi},K_{k}e_{q_{1}}\rangle\big|\big|\big\langle\Phi_{s_{3}},\tilde{a}_{l-\xi}^{*}b_{l}^{*}(K_{l}e_{\xi})\tilde{a}_{q_{2}}^{*}b_{-\ell}^{*}(K_{-\ell}e_{q_{4}})\tilde{a}_{q_{3}}\tilde{a}_{k-\xi}\Phi_{s_{3}}\big\rangle\big|,\\
	\label{esti-top-hard2-3-2}&\sum_{(k,l)\in\cB_{\xi}}\sum_{\ell\in\Z_{*}^{3}}\sum_{q\in S_{k}\cap S_{\ell}}\big|\langle e_{\xi},K_{k}e_{q_{1}}\rangle\big|\big|\big\langle\Phi_{s_{3}},\tilde{a}_{l-\xi}^{*}b_{l}^{*}(K_{l}e_{\xi})\tilde{a}_{q_{2}}^{*}[b_{-\ell}(K_{-\ell}e_{q_{4}}),\tilde{a}_{q_{3}}^{*}]^{*}\tilde{a}_{k-\xi}\Phi_{s_{3}}\big\rangle\big|,
\end{align}
are bounded by $C_{V}k_{F}^{-3/2}m(\xi)$.
\end{prop}
\begin{proof}
For \eqref{esti-top-hard2-3-1}, by \eqref{excitation-bdd}--\eqref{excitation-bdd5}, \eqref{E2-com-single-esti}--\eqref{E2-com-single-esti1}, Cauchy-Schwarz inequality, Lemmas \ref{lem:gap-sum}, \ref{lem:Gronwall}, \ref{lem:excitation-Gronwall}, Proposition \ref{prop:prod-Kk-esti} and the relations $[b_{l,p},\cN_{E}]=b_{l,p}$, we obtain
\begin{align}\label{esti-top-hard2-3-1-1}
	&\sum_{(k,l)\in\cB_{\xi}}\sum_{\ell\in\Z_{*}^{3}}\sum_{q\in S_{k}\cap S_{\ell}}\big|\langle e_{\xi},K_{k}e_{q_{1}}\rangle\big|\big|\big\langle\Phi_{s_{3}},\tilde{a}_{l-\xi}^{*}b_{l}^{*}(K_{l}e_{\xi})\tilde{a}_{q_{2}}^{*}b_{-\ell}^{*}(K_{-\ell}e_{q_{4}})\tilde{a}_{q_{3}}\tilde{a}_{k-\xi}\Phi_{s_{3}}\big\rangle\big|\nonumber\\
	\displaybreak
	&\leq Ck_{F}^{-1}\sum_{(k,l)\in\cB_{\xi}}\sum_{\ell\in\Z_{*}^{3}}\sum_{q\in S_{k}\cap S_{\ell}}\frac{\hat{V}_{k}}{\sqrt{\lambda_{k,\xi}\lambda_{k,q_{1}}}}\|K_{l}e_{\xi}\|\|\cN_{l}^{1/2}b_{-\ell}(K_{-\ell}e_{q_{4}})\tilde{a}_{q_{2}}\Phi_{s_{3}}\|\|\tilde{a}_{q_{3}}\tilde{a}_{k-\xi}\Phi_{s_{3}}\|\nonumber\\
	&\leq Ck_{F}^{-3/2}m(\xi)\sum_{k,\ell\in\Z_{*}^{3}}\sum_{q\in S_{k}\cap S_{\ell}}\frac{\hat{V}_{k}\chi_{L_{k}}(\xi)}{\sqrt{\lambda_{k,q_{1}}}}\|K_{-\ell}e_{q_{4}}\|\|\cN_{-\ell}^{1/2}(\cN_{E}+1)\tilde{a}_{q_{2}}\Phi_{s_{3}}\|\|\tilde{a}_{q_{3}}\tilde{a}_{k-\xi}\Phi_{s_{3}}\|\nonumber\\
	&\leq Ck_{F}^{-2}m(\xi)\sum_{k,\ell\in\Z_{*}^{3}}\chi_{L_{k}}(\xi)\hat{V}_{k}\hat{V}_{-\ell}\sqrt{\sum_{q\in S_{k}}\|\tilde{a}_{q_{3}}\tilde{a}_{k-\xi}\Phi_{s_{3}}\|^{2}}\sqrt{\sum_{q\in S_{\ell}}\|\cN_{-\ell}^{1/2}(\cN_{E}+1)\tilde{a}_{q_{2}}\Phi_{s_{3}}\|^{2}}.
\end{align}
For the last factor in \eqref{esti-top-hard2-3-1-1}, we use the pull-through formula \eqref{pull-through} and the estimate $\tilde{a}_{p}^{*}\cN_{k}\tilde{a}_{p}\leq \cN_{k}^{1/2}\tilde{a}_{p}^{*}\tilde{a}_{p}\cN_{k}^{1/2}$ for each $p\in\Z^{3}$ and $k\in\Z_{*}^{3}$ (see \cite[Eq. (4.20)]{CHN-23}) to obtain
\begin{align*}
	\tilde{a}_{q_{2}}^{*}(\cN_{E}+1)\cN_{-\ell}(\cN_{E}+1)\tilde{a}_{q_{2}}&=\cN_{E}\tilde{a}_{q_{2}}^{*}\cN_{-\ell}\tilde{a}_{q_{2}}\cN_{E}\leq \cN_{E}\cN_{-\ell}^{1/2}\tilde{a}_{q_{2}}^{*}\tilde{a}_{q_{2}}\cN_{-\ell}^{1/2}\cN_{E}.
\end{align*}
Hence, the last factor in \eqref{esti-top-hard2-3-1-1} yields
\begin{align}
	\sum_{q\in S_{\ell}}\|\cN_{-\ell}^{1/2}(\cN_{E}+1)\tilde{a}_{q_{2}}\Phi_{s_{3}}\|^{2}&\leq \sum_{q\in S_{\ell}}\|\tilde{a}_{q_{2}}\cN_{-\ell}^{1/2}\cN_{E}\Phi_{s_{3}}\|^{2}\nonumber\\
	&\leq \|\cN_{E}^{1/2}\cN_{-\ell}^{1/2}\cN_{E}\Phi_{s_{3}}\|^{2}=\|\cN_{-\ell}^{1/2}\cN_{E}^{3/2}\Phi_{s_{3}}\|^{2}.
\end{align}
Thus, using the argument in \eqref{ex-esti}--\eqref{ex-esti2}, \eqref{esti-top-hard2-3-1-1} yields for each $\delta>0$ that
\begin{align}
	&\sum_{(k,l)\in\cB_{\xi}}\sum_{\ell\in\Z_{*}^{3}}\sum_{q\in S_{k}\cap S_{\ell}}\big|\langle e_{\xi},K_{k}e_{q_{1}}\rangle\big|\big|\big\langle\Phi_{s_{3}},\tilde{a}_{l-\xi}^{*}b_{l}^{*}(K_{l}e_{\xi})\tilde{a}_{q_{2}}^{*}b_{-\ell}^{*}(K_{-\ell}e_{q_{4}})\tilde{a}_{q_{3}}\tilde{a}_{k-\xi}\Phi_{s_{3}}\big\rangle\big|\nonumber\\
	&\leq Ck_{F}^{-2}m(\xi)\sum_{k,\ell\in\Z_{*}^{3}}\chi_{L_{k}}(\xi)\hat{V}_{k}\hat{V}_{-\ell}\|\tilde{a}_{k-\xi}\cN_{E}^{1/2}\Phi_{s_{3}}\|\|\cN_{-\ell}^{1/2}\cN_{E}^{3/2}\Phi_{s_{3}}\|\nonumber\\
	&\leq Ck_{F}^{-2}m(\xi)\sqrt{\sum_{k\in\Z_{*}^{3}}\hat{V}_{k}^{2}}\sqrt{\sum_{\ell\in\Z_{*}^{3}}\hat{V}_{-\ell}^{2}}\sqrt{\sum_{k\in\Z_{*}^{3}}\chi_{L_{k}}(\xi)\|\tilde{a}_{k-\xi}\cN_{E}^{1/2}\Phi_{s_{3}}\|^{2}}\sqrt{\sum_{\ell\in\Z_{*}^{3}}\|\cN_{-\ell}^{1/2}\cN_{E}^{3/2}\Phi_{s_{3}}\|^{2}}\nonumber\\
	&\leq C_{\delta,V}k_{F}^{-2}m(\xi)\|\cN_{E}\Phi_{s_{3}}\|\|\cN_{E}^{5/2}\Phi_{s_{3}}\|\leq C_{\delta,V}k_{F}^{-2}m(\xi).
\end{align}

Next, for \eqref{esti-top-hard2-3-2}, we first compute the commutator using \cite[Eq.(4.23)]{CHN-23}:
\begin{align*}
	[b_{-\ell}(K_{-\ell}e_{q_{4}}),\tilde{a}_{q_{3}}^{*}]&=\begin{cases}
		-\chi_{L_{-\ell}'}(q_{3})\langle K_{-\ell}e_{q_{4}},e_{q_{3}-\ell}\rangle\tilde{a}_{q_{3}-\ell}\quad&S_{k}=L_{k},\\
		\chi_{L_{-\ell}}(q_{3})\langle K_{-\ell}e_{q_{4}},e_{q_{3}}\rangle\tilde{a}_{q_{3}+\ell}& S_{k}=L_{k}',
	\end{cases}
\end{align*}
which satisfies the following bound (see \cite[Eq. (4.28)]{CHN-23})
\begin{align*}
	\begin{cases}
		\chi_{L_{-\ell}'}(q_{3})\big|\langle K_{-\ell}e_{q_{4}},e_{q_{3}-\ell}\rangle\big|\quad&S_{k}=L_{k}\\
		\chi_{L_{-\ell}}(q_{3})\big|\langle K_{-\ell}e_{q_{4}},e_{q_{3}}\rangle\big|& S_{k}=L_{k}'
	\end{cases}\leq C\frac{k_{F}^{-1}\hat{V}_{-\ell}}{\sqrt{\lambda_{k,q_{1}}\lambda_{-\ell,q_{4}}}}.
\end{align*}
It follows from the same argument above that
\begin{align}
	&\sum_{(k,l)\in\cB_{\xi}}\sum_{\ell\in\Z_{*}^{3}}\sum_{q\in S_{k}\cap S_{\ell}}\big|\langle e_{\xi},K_{k}e_{q_{1}}\rangle\big|\big|\big\langle\Phi_{s_{3}},\tilde{a}_{l-\xi}^{*}b_{l}^{*}(K_{l}e_{\xi})\tilde{a}_{q_{2}}^{*}[b_{-\ell}(K_{-\ell}e_{q_{4}}),\tilde{a}_{q_{3}}^{*}]^{*}\tilde{a}_{k-\xi}\Phi_{s_{3}}\big\rangle\big|\nonumber\\
	&\leq Ck_{F}^{-2}\sum_{(k,l)\in\cB_{\xi}}\sum_{\ell\in\Z_{*}^{3}}\sum_{q\in S_{k}\cap S_{\ell}}\frac{\hat{V}_{k}}{\sqrt{\lambda_{k,\xi}\lambda_{k,q_{1}}}}\frac{\hat{V}_{-\ell}}{\sqrt{\lambda_{k,q_{1}}\lambda_{-\ell,q_{4}}}}\|K_{l}e_{\xi}\|\|\cN_{l}^{1/2}\tilde{a}_{q_{2}}\tilde{a}_{q_{3}\mp\ell}\Phi_{s_{3}}\|\|\tilde{a}_{k-\xi}\Phi_{s_{3}}\|\nonumber\\
	&\leq Ck_{F}^{-5/2}m(\xi)\sum_{k,\ell\in\Z_{*}^{3}}\sum_{q\in S_{k}\cap S_{\ell}}\frac{\hat{V}_{k}\hat{V}_{-\ell}\chi_{L_{k}}(\xi)}{\sqrt{\lambda_{k,q_{1}}\lambda_{k,q_{1}}\lambda_{-\ell,q_{4}}}}\|\cN_{E}\tilde{a}_{q_{2}}\tilde{a}_{q_{3}\mp\ell}\Phi_{s_{3}}\|\|\tilde{a}_{k-\xi}\Phi_{s_{3}}\|\nonumber\\
	\displaybreak
	&\leq Ck_{F}^{-5/2}m(\xi)\sum_{k\in\Z_{*}^{3}}\sum_{q\in S_{k}}\frac{\hat{V}_{k}\chi_{L_{k}}(\xi)}{\sqrt{\lambda_{k,q_{1}}}}\|\tilde{a}_{k-\xi}\Phi_{s_{3}}\|\sqrt{\sum_{\ell\in\Z_{*}^{3}}\hat{V}_{-\ell}^{2}\frac{\chi_{S_{\ell}}(q)}{\lambda_{-\ell,q_{4}}}}\sqrt{\sum_{\ell\in\Z_{*}^{3}}\chi_{S_{\ell}}(q)\|\tilde{a}_{q_{2}}\cN_{E}\Phi_{s_{3}}\|^{2}}\nonumber\\
	&\leq Ck_{F}^{-5/2}m(\xi)\sum_{k\in\Z_{*}^{3}}\chi_{L_{k}}(\xi)\hat{V}_{k}\|\tilde{a}_{k-\xi}\Phi_{s_{3}}\|\|\cN_{E}^{3/2}\Phi_{s_{3}}\|\sqrt{\sum_{q\in L_{k}}\frac{1}{\lambda_{k,q}}}\sqrt{\sum_{\ell\in\Z_{*}^{3}}\hat{V}_{-\ell}^{2}\sum_{q\in L_{\ell}}\frac{1}{\lambda_{\ell,q}}}\nonumber\\
	&\leq C_{V}k_{F}^{-3/2}m(\xi)\sqrt{\sum_{k\in\Z_{*}^{3}}\hat{V}_{k}^{2}}\sqrt{\sum_{k\in\Z_{*}^{3}}\chi_{L_{k}}(\xi)\|\tilde{a}_{k-\xi}\Phi_{s_{3}}\|^{2}}\leq C_{V}k_{F}^{-3/2}m(\xi).
\end{align}
This completes the proof.
\end{proof}

Finally, after the above lengthy argument, we obtain the following theorem from Propositions \ref{prop:esti-top1}--\ref{prop:esti-top-hard2-3}:
\begin{thm}\label{thm:errorE2k}
For each $0\leq s_{1},s_{2}\leq 1$ and $(j_{1},j_{2})\in\Sigma_{*}^{2}$, each term in $\cW_{s_{1},s_{2}}^{(j_{1},j_{2})}(\widetilde{\cE}_{2,k})$ is bounded by
\begin{align*}
	C_{V}k_{F}^{-1}m(\xi)^{1/2}\cdot\sup_{0\leq \tau\leq 1}\|\tilde{a}_{\xi}\Phi_{\tau}\|+C_{\delta,V}\big(k_{F}^{-1}\cQ^{1-\delta}+k_{F}^{-1+\delta}\cQ+k_{F}^{-3/2+\delta}\cQ+k_{F}^{-3/2}\big)m(\xi),
\end{align*}
for each $\delta>0$.
\end{thm}

\section{Conclusion of proof of Theorem \ref{thm:main}}\label{sec:pf-main}
In this section, we complete the proof of Theorem \ref{thm:main}.  First, 
\begin{lemma}\label{lem:apriori-new}
The quantity $\cQ$ satisfies
\begin{align}\label{apriori-new}
	\cQ&\leq C_{\delta,V}k_{F}^{-1/2+\delta},
\end{align}
for any $\delta>0$.  
\end{lemma}
\begin{proof}
First, by the operator norm $\|\tilde{a}_{\xi}\|=1$, it is immediate that $\cQ\leq 1$ from its definition \eqref{bootstrap}.  Then, by Propositions \ref{prop:bosonized-momentum}, \ref{prop:ex-corr-momentum}, \ref{prop:esti-A5A6}, Theorems \ref{thm:esti-varepk}, \ref{thm:errorE1k} and \ref{thm:errorE2k}, we conclude for each $\delta>0$, $\xi\in\Z^{3}$ and $0\leq\tau\leq 1$ that
\begin{align}\label{Q-esti}
	\|\tilde{a}_{\xi}\Phi_{\tau}\|^{2}&\leq C_{V}k_{F}^{-1}\cQ m(\xi)^{1/2}+ C_{\delta,V}(k_{F}^{-1+\delta}+k_{F}^{-1}\cQ^{1-\delta}+k_{F}^{-1+\delta}\cQ)m(\xi)\nonumber\\
	&\leq C_{\delta,V}k_{F}^{-1+\delta}m(\xi)^{1/2}.
\end{align}
Since $m(\xi)\leq C$ uniformly in $\xi$ by its definition, we obtain
\begin{align}
	\cQ&=\sup_{\xi\in\Z^{3}}\sup_{0\leq \tau\leq 1}\|\tilde{a}_{\xi}\Phi_{\tau}\|\leq C_{\delta,V}k_{F}^{-1/2+\delta}.
\end{align}
This completes the proof.
%
\end{proof}

\begin{proof}[Proof of Theorem \ref{thm:main}]
By \eqref{momentum2}, we have
\begin{align}
	n(\xi)\equiv n_{1}(\xi)&=n_{\rm b}(\xi)+n_{{\rm ex},1}(\xi)+\cE_{1,1}(\xi)+\cE_{2,1}(\xi)+\cE_{3,1}(\xi).
\end{align}
where $n_{\rm b}(\xi):=n_{{\rm b},1}(\xi)$.  By Proposition \ref{prop:ex-corr-momentum}, we have
\begin{align}
	n_{{\rm ex},1}(\xi)=n_{\rm ex}(\xi)+\cE_{\rm ex}(\xi),
\end{align}
for some error term $|\cE_{\rm ex}(\xi)|\leq C_{V}k_{F}^{-3/2}m(\xi)$.  

Next, we estimate the error $\cE(\xi):=\cE_{1,1}(\xi)+\cE_{2,1}(\xi)+\cE_{3,1}(\xi)$ in \eqref{momentum2}.  By Proposition \ref{prop:esti-A5A6}, Theorems \ref{thm:esti-varepk}, \ref{thm:errorE1k}, \ref{thm:errorE2k} and Lemma \ref{lem:apriori-new}, it holds for each $\delta>0$ that 
\begin{align}\label{error-esti}
	\big|\cE(\xi)\big|&\leq C_{V}k_{F}^{-1}m(\xi)^{1/2}\sup_{0\leq \tau\leq 1}\|\tilde{a}_{\xi}\Phi_{\tau}\|+ C_{\delta,V}(k_{F}^{-1}\cQ^{1-\delta}+k_{F}^{-1+\delta}\cQ+k_{F}^{-1}\cQ)m(\xi)\nonumber\\
	&\leq C_{\delta,V}k_{F}^{-3/2+\delta}m(\xi)+C_{V}k_{F}^{-3/2}m(\xi)^{1/2},
\end{align}
where we have dropped any term $o(\delta)$ since we only consider small $\delta$.  This implies for each $0\leq \tau\leq 1$ that
\begin{align}
	\|\tilde{a}_{\xi}\Phi_{\tau}\|^{2}&\leq C_{\delta,V}k_{F}^{-3/2+\delta}m(\xi)^{1/2}.
\end{align}
By bootstrapping this inequality, \eqref{error-esti} becomes
\begin{align}
	\big|\cE(\xi)\big|&\leq C_{\delta,V}k_{F}^{-3/2+\delta}m(\xi)+C_{V}k_{F}^{-3/2}m(\xi)^{3/4},
\end{align}
which implies again for each $0\leq \tau\leq 1$ that
\begin{align}
	\|\tilde{a}_{\xi}\Phi_{\tau}\|^{2}&\leq C_{\delta,V}k_{F}^{-3/2+\delta}m(\xi)^{3/4}.
\end{align}
By continuing this process $n$-times, this implies for each $n\geq 1$ that
\begin{align}
	\|\tilde{a}_{\xi}\Phi_{\tau}\|^{2}&\leq C_{\delta,V}k_{F}^{-3/2+\delta}m(\xi)^{1-2^{-n}}.
\end{align}
Taking $n\rightarrow\infty$ implies that $\|\tilde{a}_{\xi}\Phi_{\tau}\|^{2}\leq C_{\delta,V}k_{F}^{-3/2+\delta}m(\xi)$.  The error estimate in \eqref{error-esti} thus implies that 
\begin{align}
	|\cE(\xi)|&\leq C_{\delta,V}k_{F}^{-3/2+\delta}m(\xi).
\end{align}
which, together with Propositions \ref{prop:bosonized-momentum} and \ref{prop:ex-corr-momentum}, gives the estimates in \eqref{boson-f-esti}--\eqref{error-f-esti} for each $f\in\ell^{\infty}(\Z^{3})$.

Now, we compute $n_{\rm b}(\xi)$.  Recall from Proposition \ref{prop:bosonized-momentum} that
\begin{align}
	n_{\rm b}(\xi)&=\sum_{(k,\zeta)\in\cC_{\xi}}\langle e_{\zeta},(\cosh(-2K_{k})-1)e_{\zeta}\rangle.
\end{align}
By definition, we have
\begin{align}\label{cosh(2Kk)-1}
	\cosh(-2K_{k})-1\DETAILS{=\frac{1}{2}\Big(h_{k}^{-1/2}\Big(h_{k}^{1/2}(h_{k}+2P_{k})h_{k}^{1/2}\Big)^{1/2}h_{k}^{-1/2}-1\Big)&\quad\quad\quad\quad\quad\quad\quad\quad+\frac{1}{2}\Big(h_{k}^{1/2}\Big(h_{k}^{1/2}(h_{k}+2P_{k})h_{k}^{1/2}\Big)^{-1/2}h_{k}^{1/2}-1\Big)\nonumber\\
	}
	&=\frac{1}{2}h_{k}^{-1/2}\Big[\big(h_{k}^{2}+2P_{u_{k}}\big)^{1/2}-h_{k}\Big]h_{k}^{-1/2}\nonumber\\
	&\quad\quad\quad\quad+\frac{1}{2}h_{k}^{1/2}\Big[\big(h_{k}^{2}+2P_{u_{k}}\big)^{-1/2}-h_{k}^{-1}\Big]h_{k}^{1/2},
\end{align}
where $P_{u_{k}}=\ket{u_{k}}\bra{u_{k}}$ with $u_{k}:=h_{k}^{1/2}v_{k}$.  For the first term in \eqref{cosh(2Kk)-1}, we apply \cite[Lemma A.8]{CHN-23} to obtain
\begin{align}
	&\Big\langle e_{\zeta},h_{k}^{-1/2}\Big(h_{k}^{1/2}(h_{k}+2P_{k})h_{k}^{1/2}\Big)^{1/2}h_{k}^{-1/2}e_{\zeta}\Big\rangle\nonumber\\
	&=\frac{4}{\pi\lambda_{k,\zeta}}\int_{0}^{\infty}\Big\langle (h_{k}^{2}+s^{2})^{-1}e_{\zeta},\Big(\frac{s^{2}}{1+2\langle u_{k},(h_{k}^{2}+s^{2})^{-1}u_{k}\rangle}\Big)P_{u_{k}}(h_{k}^{2}+s^{2})^{-1}e_{\zeta}\Big\rangle ds\nonumber\\
	&=\frac{2}{\pi}\frac{k_{F}^{-1}\hat{V}_{k}}{(2\pi)^{3}}\int_{0}^{\infty}\frac{s^{2}}{(\lambda_{k,\zeta}^{2}+s^{2})^{2}}\frac{1}{1+Q_{k}(s)}ds,
\end{align}
where we recall $Q_{k}(s)=\frac{k_{F}^{-1}\hat{V}_{k}}{(2\pi)^{3}}\sum_{p\in L_{k}}\frac{\lambda_{k,p}}{s^{2}+\lambda_{k,p}^{2}}$.  Similarly, for the second term in \eqref{cosh(2Kk)-1}, we use the identity from functional calculus to write
\begin{align}\label{inv-sqr}
	\big(h_{k}^{2}+2P_{u_{k}}\big)^{-1/2}&=\frac{2}{\pi}\int_{0}^{\infty}\frac{1}{h_{k}^{2}+2P_{u_{k}}+s^{2}}ds.
\end{align}
The resolvent for matrix with rank-1 perturbation is easily calculated using Sherman-Morrison formula as
\begin{align*}
	\frac{1}{h_{k}^{2}+2P_{u_{k}}+s^{2}}&=\frac{1}{h_{k}^{2}+s^{2}}-2\frac{(h_{k}^{2}+s^{2})^{-1}P_{u_{k}}(h_{k}^{2}+s^{2})^{-1}}{1+2\langle u_{k},(h_{k}^{2}+s^{2})^{-1}u_{k}\rangle}\nonumber\\
	&=\frac{1}{h_{k}^{2}+s^{2}}-\frac{k_{F}^{-1}\hat{V}_{k}}{(2\pi)^{3}}\frac{1}{1+Q_{k}(s)}\sum_{p,q\in L_{k}}M_{pq}(s)\ket{e_{p}}\bra{e_{q}},
\end{align*}
where $M_{pq}(s)=\frac{\lambda_{k,p}^{1/2}\lambda_{k,q}^{1/2}}{(s^{2}+\lambda_{k,p}^{2})(\lambda_{k,q}^{2}+s^{2})}$.  It follows that 
\begin{align*}
	\big(h_{k}^{2}+2P_{u_{k}}\big)^{-1/2}&=\frac{2}{\pi}\int_{0}^{\infty}\Big[\frac{1}{h_{k}^{2}+s^{2}}-\frac{k_{F}^{-1}\hat{V}_{k}}{(2\pi)^{3}}\frac{1}{1+Q_{k}(s)}\sum_{p,q\in L_{k}}M_{pq}(s)\ket{e_{p}}\bra{e_{q}}\Big]ds\nonumber\\
	&=h_{k}^{-1}-\frac{2}{\pi}\frac{k_{F}^{-1}\hat{V}_{k}}{(2\pi)^{3}}\int_{0}^{\infty}\frac{1}{1+Q_{k}(s)}\sum_{p,q\in L_{k}}M_{pq}(s)\ket{e_{p}}\bra{e_{q}}ds.
\end{align*}
Hence, the second term in \eqref{cosh(2Kk)-1} gives
\begin{equation}
	\begin{split}
		\Big\langle e_{\zeta},h_{k}^{1/2}\Big(\frac{1}{h_{k}^{2}+2P_{u_{k}}}-h_{k}^{-1}\Big)h_{k}^{1/2}e_{\zeta}\Big\rangle&=-\frac{2\lambda_{k,\zeta}}{\pi}\frac{k_{F}^{-1}\hat{V}_{k}}{(2\pi)^{3}}\int_{0}^{\infty}\frac{1}{1+Q_{k}(s)}M_{\zeta\zeta}(s)ds\nonumber\\
		&=-\frac{2\lambda_{k,\zeta}^{2}}{\pi}\frac{k_{F}^{-1}\hat{V}_{k}}{(2\pi)^{3}}\int_{0}^{\infty}\frac{1}{1+Q_{k}(s)}\frac{1}{(\lambda_{k,\zeta}^{2}+s^{2})^{2}}ds.
	\end{split}
\end{equation}
Putting everything together yields
\begin{equation}
	\langle e_{\zeta},(\cosh(-2K_{k})-1)e_{\zeta}\rangle=\frac{2}{\pi}\frac{k_{F}^{-1}\hat{V}_{k}}{(2\pi)^{3}}\int_{0}^{\infty}\frac{(s^{2}-\lambda_{k,\zeta}^{2})(\lambda_{k,\xi}^{2}+s^{2})^{-2}}{1+Q_{k}(s)}ds,
\end{equation}
which completes the proof.
\end{proof}

\appendix

\section{Comparison with work by Daniel and Voskov}\label{sec:compare-DV}
In the seminal work \cite{DanielVosko-60} by Daniel and Voskov, the momentum distribution of ground state in random phase approximation for electron gas in high density limit was obtained due to a Hellmann-Feynman type argument.  By a suitable change of variables, \cite[Eqs. (8), (9)]{DanielVosko-60} coincides exactly with \cite[Eq. (2.8)]{Lam-71}, which dealt with the same system at metallic densities.  

In this section, at least formally, we compare our result with the momentum distribution obtained by Daniel and Voskov, given in thermodynamic limit followed high density limit.  We shall focus on the case $\xi\in B_{F}^{c}$ since the reverse case can be treated similarly.  According to \cite[Eq. (8)]{DanielVosko-60}, the momentum distribution due to bosonization for $\xi\in B_{F}^{c}$ is given by (and we have adapted their formula in our notations)
\begin{align}\label{bosonization-DV}
	n_{\rm b}^{\rm (DV)}(\xi)&=\frac{k_{F}\alpha}{|\xi|}\int_{|\xi|-k_{F}}^{|\xi|+k_{F}}d|k| |k|\int_{0}^{\infty}\Big[\frac{|\xi|-|k|/2}{\big(|\xi|-|k|/2\big)^{2}+s^{2}}-\frac{\frac{|\xi|^{2}-k_{F}^{2}}{2|k|}}{\Big(\frac{|\xi|^{2}-k_{F}^{2}}{2|k|}\Big)^{2}+s^{2}}\Big]\nonumber\\
	&\quad\quad\quad\quad\quad\quad\quad\quad\quad\quad\quad\times\Big[|k|^{2}+\alpha k_{F}^{2} Q_{k}^{\rm (DV)}(s)\Big]^{-1}ds,
\end{align}
where
\begin{align*}
	Q_{k}^{\rm (DV)}(s)&=k_{F}^{-1}\int_{|p|<k_{F}}dp\int_{-\infty}^{\infty} dt e^{its|k|}\exp\Big[-|t|\Big(p\cdot k+\frac{1}{2}|k|^{2}\Big)\Big]\nonumber\\
	&=2\pi\Big[1+\frac{1}{2|k|k_{F}}\big(k_{F}^{2}-\frac{1}{4}|k|^{2}+s^{2}\big)\ln\Big(\frac{(k_{F}+|k|/2)^{2}+s^{2}}{(k_{F}-|k|/2)^{2}+s^{2}}\Big)\nonumber\\
	&\quad\quad\quad-\frac{s}{k_{F}}\arctan\Big(\frac{k_{F}+|k|/2}{s}\Big)-\frac{s}{k_{F}}\arctan\Big(\frac{k_{F}-|k|/2}{s}\Big)\Big].
\end{align*}
We remark that the variables $(k,q,p,u)$ in \cite{DanielVosko-60} corresponds to $(\xi/k_{F},k/k_{F},p/k_{F},s/k_{F})$ in our notations.  

To compare with our result, we denote $\cC_{\xi}:=\{k\in\Z_{*}^{3}\mid \xi\in L_{k}\}$ and replace our underlying configuration space by $[0,L]^{3}$, substitute
\begin{align*}
	k_{F}^{-1}\hat{V}_{k}\rightarrow 4\pi e^{2}|k|^{-2},\quad (2\pi)^{3}\rightarrow L^{3}\text{ (i.e., the volume of }[0,L]^{3}),
\end{align*}
and then formally take thermodynamic limit $L\rightarrow\infty$.  In this limit, we replace $L^{-3}\sum_{k}$ by integration $\int_{\R^{3}}dk$ and write $\lambda_{k,\xi}=\xi\cdot k-\frac{1}{2}|k|^{2}=|k|\big(\xi\cdot\hat{k}-\frac{1}{2}|k|\big)=:|k|\tilde{\lambda}_{k,\xi}$ so that our formula in \eqref{bosonization} becomes
\begin{align}\label{bosonization3}
	n_{\rm b}(\xi)&\sim \DETAILS{4\pi e^{2}\int_{\cC_{\xi}}\frac{dk}{|k|^{2}}\int_{0}^{\infty}\frac{(s^{2}-\lambda_{k,\xi}^{2})(s^{2}+\lambda_{k,\xi}^{2})^{-2}}{1+Q_{k}(s)}ds\nonumber\\
	&=}4\pi e^{2}\int_{\cC_{\xi}}\frac{dk}{|k|}\int_{0}^{\infty}\frac{(s^{2}-\tilde{\lambda}_{k,\xi}^{2})(s^{2}+\tilde{\lambda}_{k,\xi}^{2})^{-2}}{|k|^{2}+|k|^{2}Q_{k}(|k|s)}ds.
\end{align}
By replacing the Riemann summation $L^{-3}\sum_{p\in L_{k}}$ with the corresponding integral and then rewriting the resulting formula using Abel kernel, we obtain
\begin{align*}
	|k|^{2}Q_{k}(|k|s)\DETAILS{=\frac{4\pi e^{2}}{L^{3}k_{F}^{2}|k|^{2}}\sum_{p\in L_{k}}\frac{\lambda_{k,p}}{\lambda_{k,p}^{2}+|k|^{2}s^{2}}\nonumber\\}
	&\sim 4\pi e^{2}\int_{\overline{B}(k,k_{F})\cap B_{k_{F}}(0)^{c}}\frac{p\cdot k-\frac{1}{2}|k|^{2}}{(p\cdot k-|k|^{2}/2)^{2}+|k|^{2}s^{2}}dp\nonumber\\
	&=\frac{e^{2}}{2\pi}\int_{\overline{B}(0,k_{F})\cap B(-k,k_{F})^{c}}dp\int_{-\infty}^{\infty}dte^{-its|k|}\exp\big[-|t|\big(p\cdot k+\frac{1}{2}|k|^{2}\big)\big].
\end{align*}
For $|k|\geq 2k_{F}$, we see the integral region is simply $\overline{B}(0,k_{F})$ so that, in this regime, the integral is
\begin{align}
	\int_{|p|\leq k_{F}}dp\int_{-\infty}^{\infty}dt e^{-its|k|}\exp\Big[-|t|\big(p\cdot k+\frac{1}{2}|k|^{2}\big)\Big].
\end{align}
It follows that, in the regime where $k$ sufficiently far away from Fermi ball, we may identify $Q_{k}(|k|s)$ with $2\pi |k|^{-2}\alpha k_{F} Q_{k}^{\rm (DV)}(s)$.  

Now, we simplify the integral in \eqref{bosonization3} by switching to spherical coordinate of $k$ and following a similar argument as in \cite[Appendix C]{BeneLill-25}.  Since $\chi_{L_{k}}(\xi)$ is non-zero only when $|\xi|+k_{F}\leq |k|<|\xi|+k_{F}$, the radial integral over $|k|$ starts from $R_{\xi}:=|\xi|-k_{F}$ and can end at $|\xi|+k_{F}$.  The integration over $\theta$, measuring the angle between $\xi$ and $k$, runs from $0$ to $\theta_{\rm max}$ with $\cos\theta_{\rm max}\approx R_{\xi}/|k|=:\lambda_{\rm min}$.  It follows that
\begin{align}\label{bosonization4}
	&n_{\rm b}(\xi)\sim 8\pi^{2} e^{2}\int_{|\xi|-k_{F}}^{|\xi|+k_{F}}d|k||k|\int_{0}^{\theta_{\rm max}}\sin\theta d\theta\int_{0}^{\infty}\frac{s^{2}-(|\xi|\cos\theta-|k|/2)^{2}}{(s^{2}+(|\xi|\cos\theta-|k|/2)^{2})^{2}}\nonumber\\
	&\quad\quad\quad\quad\quad\quad\quad\quad\quad\quad\quad\quad\quad\quad\quad\quad\quad\quad\times\Big[|k|^{2}+|k|^{2}Q_{k}(|k|s)\Big]^{-1}ds\nonumber\\
	\displaybreak
	&=8\pi^{2} e^{2}\int_{|\xi|-k_{F}}^{|\xi|+k_{F}}d|k||k| \int_{0}^{\infty}\Big[\int_{\lambda_{\rm min}}^{1}\frac{s^{2}-(|\xi|\lambda-|k|/2)^{2}}{(s^{2}+(|\xi|\lambda-|k|/2)^{2})^{2}}d\lambda\Big]\nonumber\\
	&\quad\quad\quad\quad\quad\quad\quad\quad\quad\quad\quad\quad\quad\quad\quad\quad\quad\quad\times\Big[|k|^{2}+|k|^{2}Q_{k}(|k|s)\Big]^{-1}ds\nonumber\\
	&=\frac{8\pi^{2}e^{2}}{|\xi|}\int_{|\xi|-k_{F}}^{|\xi|+k_{F}}d|k||k|\int_{0}^{\infty}\Big[\frac{|\xi|-|k|/2}{s^{2}+(|\xi|-|k|/2)^{2}}-\frac{\frac{|\xi|^{2}-(|k|^{2}+2|\xi|k_{F}-|\xi|^{2})}{2|k|}}{s^{2}+\big(\frac{|\xi|^{2}-(|k|^{2}+2|\xi|k_{F}-|\xi|^{2})}{2|k|}\big)^{2}}\Big]\nonumber\\
	&\quad\quad\quad\quad\quad\quad\quad\quad\quad\quad\quad\quad\quad\quad\quad\quad\quad\quad\times\Big[|k|^{2}+|k|^{2}Q_{k}(|k|s)\Big]^{-1}ds.
\end{align}
In the region $|k|\sim |\xi|\gg k_{F}\gg 1$ (e.g., $|k|\sim |\xi|\sim k_{F}^{\varepsilon}$ for $\varepsilon>1$), we can formally write $k_{F}|\xi|^{-1}+|\xi|^{2}-|k|^{2}\sim o(1)$ and 
\begin{align*}
	\frac{|\xi|^{2}-\big(|k|^{2}-|\xi|^{2}+2|\xi|k_{F}\big)}{2|k|}&\approx\frac{|\xi|^{2}(1-o(k_{F}|\xi|^{-1}))+o(1)}{2|k|}\approx\frac{|\xi|^{2}}{2|k|}+o(1),
\end{align*}
where $o(1)$ denotes any formal small parameter.  The quantity $\frac{|\xi|^{2}-k_{F}^{2}}{2|k|}$ in \eqref{bosonization-DV} admits a similar approximation.  In this way, we may formally identify $n_{\rm b}^{\rm (DV)}(\xi)$ with our result in the high momentum regime.

Now, the second order correction for momentum distribution due to exchange correlation for $|\xi|>k_{F}$ was given explicitly in \cite[Eq. (12)]{DanielVosko-60} as
\begin{align}
	n_{\rm ex}^{\rm (DV)}(\xi)=-\frac{k_{F}^{2}\alpha^{2}}{4}\int_{|p+k|>k_{F}}\frac{dk}{|k|^{2}}\int_{|p|<k_{F}}\frac{dp}{[k\cdot(p-\xi)]^{2}|p-\xi|^{2}}.
\end{align}
In our result, by formally replacing summations with the corresponding integral as above and defining a new variable $p'=k-p$, the term \eqref{exchange} for $|\xi|>k_{F}$ can be written as
\begin{align}
	n_{\rm ex}(\xi)&\sim-\frac{16\pi^{2}e^{4}}{8}\int_{k\in\cC_{\xi}}\frac{dk}{|k|^{2}}\int_{\substack{|p'+k|>k_{F},\\ |p'|\leq k_{F}}}\frac{|p'-\xi|^{-2}}{(\lambda_{k,p'+k}+\lambda_{k,\xi})^{2}}dp\nonumber\\
	&=-2\pi^{2}e^{4}\int_{\substack{k\in\cC_{\xi},\\ |p'+k|>k_{F}}}\frac{dk}{|k|^{2}}\int_{|p'|\leq k_{F}}\frac{dp'}{|p'-\xi|^{2}[(p'-\xi)\cdot k]^{2}},
\end{align}
where we have used the following computation
\begin{align*}
	\lambda_{k,p}+\lambda_{k,\xi}&=\frac{1}{2}\big(|p'-k|^{2}-|p'|^{2}\big)+\frac{1}{2}\big(|\xi|^{2}-|\xi-k|^{2}\big)\\
	&=-p'\cdot k+\frac{1}{2}|k|^{2}+\xi\cdot k-\frac{1}{2}|\xi|^{2}=(\xi-p')\cdot k.
\end{align*}
We see that the integrand of $n_{\rm ex}(\xi)$ and $n_{\rm ex}^{\rm (DV)}(\xi)$ exactly coincides.  To conclude this section, we give some comments on the integral region.  For $k\in\cC_{\xi}$, one must have $|\xi-k|\leq k_{F}<|\xi|$ for $\xi\in B_{F}^{c}$.  For those $\xi$'s that are really far away from Fermi ball, those $k$'s that belong to $\cC_{\xi}$ must be far away from Fermi ball as well.  Since $p'\in B_{F}$, the condition $|p'+k|>k_{F}$ for these $k$'s is almost always guaranteed.  In this way, we can drop the symbol ``$k\in\cC_{\xi}$" in the above integral, which make $n_{\rm ex}(\xi)$ and $n_{\rm ex}^{\rm (DV)}(\xi)$ directly propositional to each other within this regime of $\xi$.

\bigskip

\noindent\textbf{Acknowledgements}.  The author is very greatful to Phan Th\`anh Nam for suggesting this work and many valuable feedbacks.  The author also thanks Niels Benedikter, Sascha Lill and Diwakar Naidu for stimulating discussions.  This work is partially supported by the European Research Council (ERC Consolidator Grant RAMBAS Project Nr. 10104424) and the Deutsche Forschungsgemeinschaft (TRR 352 Project Nr. 470903074).

\medskip

\noindent\textbf{Conflict of Interest}.  The author has no conflicts to disclose.

\medskip

\noindent\textbf{Data Availability}.
Data sharing is not applicable to this article as no new data were created or analyzed in this study.

\end{document}